\documentclass[12pt]{article}

\usepackage[english]{babel}
\usepackage[latin9]{inputenc}
\usepackage{subcaption}
\usepackage{graphicx}
\usepackage{comment}
\usepackage{amsmath}
\usepackage{amsfonts}
\usepackage{ctable}

\numberwithin{figure}{section}
\usepackage{subcaption}
\usepackage{booktabs}

\usepackage{natbib}
\usepackage{booktabs}
\usepackage{float}
\usepackage{color}
\usepackage{appendix}
\usepackage[inline]{enumitem}
\allowdisplaybreaks
\usepackage{upgreek}
\usepackage{mathtools}
\usepackage{mathabx} 
\usepackage{amssymb,amsthm}
\usepackage[]{bm}
\usepackage{bbm}
\usepackage{xcolor}
\usepackage{titlesec}
\usepackage[algosection]{algorithm2e}
\usepackage[T1]{fontenc}
\usepackage[latin9]{luainputenc}
\usepackage{mathrsfs}
\usepackage{babel}
\usepackage{multirow}

\titleformat{\subsubsection}[runin]
{\normalfont\bfseries}{\thesubsubsection}{1em}{}

\usepackage{xr-hyper}
\usepackage[colorlinks,citecolor=blue,urlcolor=blue]{hyperref}
\hypersetup{pdfborder = {0 0 0},colorlinks=true,linkcolor=blue,urlcolor=blue,citecolor=blue}

\DeclareMathOperator*{\argmin}{arg\!\min}
\newcommand{\calA}{\mathcal{A}}
\newcommand{\calB}{\mathcal{B}}
\newcommand{\calF}{\mathcal{F}}
\newcommand{\calG}{\mathcal{G}}
\newcommand{\calH}{\mathcal{H}}
\newcommand{\calI}{\mathcal{I}}
\newcommand{\calL}{\mathcal{L}}
\newcommand{\calN}{\mathcal{N}}

\newcommand{\calP}{\mathcal{P}}
\newcommand{\calQ}{\mathcal{Q}}
\newcommand{\calT}{\mathcal{T}}
\newcommand{\calZ}{\mathcal{Z}}
\newcommand{\Co}{\mathcal{o}}

\def\wh{\widehat}
\def\wt{\widetilde}
\def\wc{\widecheck}

\def\rmd{{\rm d}}

\def\RR{\mathbb{R}}
\def\EE{\mathbb{E}}

\def\bu{\mathbf{u}}
\def\bv{\mathbf{v}}
\def\bw{\mathbf{w}}
\def\ba{\bm{a}}
\def\bb{\bm{b}}
\def\bx{\bm{x}}
\def\bw{\bm{w}}
\def\bz{\bm{z}}

\def\bA{\mathbf{A}}
\def\bB{\mathbf{B}}
\def\bH{\mathbf{H}}

\def\bW{\mathbf{W}}
\def\bX{{\bm{X}}}
\def\bfF{{\bm{F}}}
\def\bff{{\bm{f}}}

\def\ua{\textup{a}}
\def\ub{\textup{b}}
\def\us{\textup{s}}
\def\ut{\textup{t}}

\def\bmu{\bm{u}}
\def\bXi{\bm{\Xi}}
\def\bxi{\bm{\xi}}

\def\ind{\mathbbm{1}}

\def\Co{{\scriptstyle{\mathcal{O}}}}

\def\NN{\mathbb{N}}
\def\PP{\mathbb{P}}
\def\PPn{\PP_n}
\def\calFn{\calF_n}
\def\calX{\mathcal{X}}

\def\rbar{{\overline{r}}}
\def\dbar{{\overline{d}}}
\def\sigmabar{{\overline{\sigma}}}

\def\bz{\bm{z}}
\def\uf{{\textup{f}}}
\def\nun{\nu_n}
\def\nuan{\nu_{\ua,n}}
\def\nusn{\nu_{\us,n}}
\def\deltaf{\delta_\uf}

\def\deltat{\delta_\ut}
\def\deltathat{\widehat{\delta}_\ut}
\def\cp{c_{\textup{p}}}
\def\pn{p_n}
\def\Rnull{R_n^{\textup{null}}}
\def\Rnullhat{\widehat{R}_n^{\textup{null}}}
\def\alphanull{\alpha_n^{\textup{null}}}
\def\rnull{r_{\alpha,n}^{\textup{null}}}
\def\rinf{r_{\textup{inf},n}}

\newtheorem{theorem}{Theorem}[section]
\newtheorem{lemma}[theorem]{Lemma}
\newtheorem{proposition}[theorem]{Proposition}

\newtheorem{assumption}{Assumption}[section]
\newtheorem{definition}{Definition}[section]
\newtheorem{remark}{Remark}[section]

\makeatletter
\@addtoreset{proofpart}{theorem}
\makeatother

\newcommand{\blind}{0}

\def\wh{\widehat}
\def\wt{\widetilde}
\def\wc{\widecheck}

\def\rmd{{\rm d}}

\def\RR{\mathbb{R}}
\def\EE{\mathbb{E}}
\def\NN{\mathbb{N}}

\def\bu{\mathbf{u}}
\def\bw{\mathbf{w}}
\def\ba{\bm{a}}
\def\bx{\bm{x}}

\def\bB{\mathbf{B}}
\def\bH{\mathbf{H}}

\def\bW{\mathbf{W}}
\def\bX{{\bm{X}}}
\def\bfF{{\bm{F}}}
\def\bff{{\bm{f}}}

\def\ua{\textup{a}}
\def\ub{\textup{b}}
\def\ud{\textup{d}}
\def\un{\textup{n}}

\def\us{\textup{s}}
\def\ut{\textup{t}}

\def\bmu{{\bm{u}}}
\def\bXi{{\bm{\Xi}}}
\def\bxi{{\bm{\xi}}}
\def\buc{\bu_{\bXi,\bz,d}}

\def\ind{\mathbbm{1}}

\def\Co{{\scriptstyle{\mathcal{O}}}}

\def\PP{\mathbb{P}}
\def\PPn{\PP_n}
\def\calFn{\calF_n}
\def\calFnbar{\overline{\calF}_n}
\def\calGn{\calG_n}
\def\calX{\mathcal{X}}

\def\rbar{{\overline{r}}}
\def\dbar{{\overline{d}}}

\def\bz{\bm{z}}
\def\uf{{\textup{f}}}
\def\nun{\nu_n}
\def\nuan{\nu_{\ua,n}}
\def\nusn{\nu_{\us,n}}
\def\deltaf{\delta_\uf}

\def\deltat{\delta_\ut}
\def\deltathat{\widehat{\delta}_\ut}
\def\cp{c_{\textup{p}}}
\def\pn{p_n}
\def\Rnull{R_n^{\textup{null}}}
\def\Rnullhat{\widehat{R}_n^{\textup{null}}}
\def\alphanull{\alpha_n^{\textup{null}}}
\def\rnull{r_{\alpha,n}^{\textup{null}}}
\def\rinf{r_{\textup{inf},n}}

\newcommand{\calE}{\mathcal{E}}
\newcommand{\calS}{\mathcal{S}}

\begin{document}

\def\spacingset#1{\renewcommand{\baselinestretch}
{#1}\small\normalsize} \spacingset{1.0}
\def\r#1{\textcolor{red}{\bf #1}}
\def\b#1{\textcolor{blue}{\bf #1}}

\if0\blind
{
  \title{\bf
 Conditional nonparametric variable screening by neural  factor regression}
  \author{Jianqing Fan\\
  Operations Research and Financial Engineering, Princeton University\\
    Weining Wang \\
    Department of Economics, Econometrics and Finance, University of Groningen\\
    Yue Zhao\\
    Department of Mathematics, University of York}
  \maketitle
} \fi

\if1\blind
{
  \bigskip
  \begin{center}
    {\LARGE\bf Conditional nonparametric variable screening via neural network factor regression}
\end{center}
  \medskip
} \fi

\bigskip

\bigskip
\begin{abstract}
High-dimensional covariates often admit linear factor structure. To effectively screen correlated covariates in high-dimension, we propose a conditional variable screening test based on non-parametric regression using neural networks due to their representation power.  We ask the question whether individual covariates have additional contributions given the latent factors or more generally a set of variables.
Our test statistics are based on the estimated partial derivative of the regression function of the candidate variable for screening  and  a observable proxy for the latent factors.  Hence, our test reveals how much predictors contribute additionally to the non-parametric regression after accounting for the latent factors.  Our derivative estimator is the convolution of a deep neural network regression estimator and a smoothing kernel.  We demonstrate that when the neural network size diverges with the sample size, unlike estimating the regression function itself, it is necessary to smooth the partial derivative of the neural network estimator to recover the desired convergence rate for the derivative.  Moreover, our screening test achieves asymptotic normality under the null after finely centering our test statistics that makes the biases negligible, as well as consistency for local alternatives under mild conditions.  We demonstrate the performance of our test in a simulation study and two real world applications.
\end{abstract}

\vfill

\newpage
\spacingset{1.2}

\section{Introduction}
\label{sec:intro}

\subsection{Background}
\label{sec:background}

Variable screening is a  powerful tool to expeditiously identify  the set of predictors that potentially affect the regression outcome \citep{FanLv2008}.  It can reduce a very large number of predictors to a smaller and more manageable set. Then, on this reduced set, one can apply some more refined but computationally demanding variable selection methods such as Lasso, SCAD, Danzig selector \citep{Tibshirani1996,FanLi2001,CandesTao2007,FanLiZhangZou2020} and their non-parametric counterpart FAST-NN \citep{FanGu2023factor}.  Conditional marginal screening in parametric regression \citep{barut2016conditional} further augments the screening by conditioning on a \textit{known} set of useful predictors to reduce the impact of the correlations among the predictors, thus making the important {predictors more visible and reducing the false positive and false negative rates in the vanilla screening method}.  Despite the aforementioned advances, conducting variable screening or selection in non-parametric regression with high-dimensional inputs remains a challenge due to curse of dimensionality, and
deep learning offers a promising solution thanks to its ability to adapt to unknown low-dimensional structures in multivariate non-parametric regression.

Deep learning has achieved tremendous empirical successes in numerous applications \citep{LeCunBengioHinton2015deep, GoodfellowBengioCourville2016deep}, for instance, in high-dimensional problems such as image recognition \citep{SimonyanZisserman2015, KrizhevskySutskeverHinton2017imagenet}, deep reinforcement learning \citep{mnih2015human}, and large language models \citep{kasneci2023chatgpt,thirunavukarasu2023large}.  There is now also a growing literature justifying theoretically the benefit of depth in deep neural networks \citep{Telgarsky2016benefits, Yarotsky2017, ElbrachterPerekrestenkoGrohsBolcskei2021} and their power in alleviating the curse of dimensionality in non-parametric regression via algorithmic learning of unknown low-dimensional structures within complex functions \citep{schmidt2020nonparametric, KohlerLanger2021}. Now being a component of the standard toolbox for statisticians, deep neural networks may nevertheless not be efficient if the dimension of the predictors is very high due to the fundamental limit of  multivariate non-parametric regression.  Recently \cite{FanGu2023factor} proposed a factor-augmented sparse throughput regression model that simultaneously leveraged the aforementioned adaptivity of deep neural networks and a factor model on the predictors to facilitate variable selection. To complement the variable selection effort in \cite{FanGu2023factor}, \cite{DinhHo2020}, and \cite{HoRichardsonTran2023}, in this paper, we investigate the issue of conditional variable screening with deep neural networks when facing potentially very high-dimensional inputs.

We will assess conditional contribution of a candidate variable for screening by examining its partial derivative in the multivariate non-parametric regression function with a given set of variables and construct our screening test statistics on the moment generating function (MGF) of the smoothed partial derivative of the  regression function estimator. We chose a deep neural network as our regression estimator due to its aforementioned algorithmic adaptation to the low-dimensional structure.
Thus, quantifying and improving the performance of derivatives of deep neural networks are integral to our study, and also form an interesting topic of its own right given the importance of derivative estimation in non-parametric regression across diverse research domains and practical applications \citep{GijbelsGoderniaux2005, RondonottiMarronPark2007, horel2020significance}.  Traditional non-parametric derivative estimation in general follows one of the following three methods: empirical derivative-type estimation \citep{MullerStadtmullerSchmitt1987, DeBrabanterDeBrabanterGijbelsDeMoor2013, LiuDeBrabanter2020}, kernel/local polynomial-type estimation \citep{GasserMuller1984,fan1996local}, and series/spline-type estimation \citep{Stone1985,ZhouWolfe2000}.  However, the first method heavily relies on the existence of an order among the predictor samples and hence naturally applies when the predictor dimension is just one, and the last two methods, just like their original regression estimation counterparts, suffer from the curse of dimensionality problem {when facing predictors of a moderate  dimension}.
	
Despite their deteriorated performances posed by high dimensionality, the traditional derivative estimation methods are relatively amenable to theoretical analysis due to their closed-form solutions.  Excluding the empirical derivative method, which avoids explicitly fitting a regression function, the closed-form solutions for the last two methods can be attributed to the {close} connection between estimating the original regression functions and their accompanying derivatives.  For instance, in the kernel and the spline methods, a {closed-form} derivative estimator can simply be obtained as the derivative of the original regression function estimator (see, for instance, the discussion between Eqs.~(2) and (3) in \cite{ZhouWolfe2000}).  Thus, unsurprisingly, in these methods the quality of the derivative estimation closely follows the quality of the original regression function estimation.  However, for deep neural networks, estimating the \textit{derivative} of a regression function can be quite different from the task on the regression function itself due to its smoothness.

Take for example any candidate regression function estimator $m$ within the canonical neural network class, precisely defined in \eqref{eq:base_NN_class} later, built from the popular ReLU activation function.  The first order partial derivatives of $m$ are necessarily piecewise constant due to the piecewise linear nature of the ReLU function, and the second order partial derivatives of $m$ are necessarily \textit{zero} almost everywhere, \textit{irrespective} of the underlying truth that the function $m$ may attempt to recover.  Numerically, this is easily demonstrable through standard software packages such as PyTorch.  In addition, the first order partial derivatives of $m$ could exhibit convergence behaviors qualitatively different from the convergence behavior of $m$ itself as we will explain in Section~\ref{sec:test_stat}.  Such phenomena are in sharp contrast to the traditional non-parametric estimators, and can in part be attributed to the features of deep neural networks: they are highly non-linear and lack easily interpretable closed-form expressions.  {Moreover, existing asymptotic results on functionals acting on deep neural networks almost invariably assume some continuity of the functionals with respect to their inputs.  However, when such functionals in effect act on the derivatives of neural networks, the assumed continuities {could} break down due to the aforementioned different convergence behaviors of the said derivatives.}  We will explain and address this discrepancy as we progress through the paper.

\subsection{Our method and contribution}
\label{sec:our_method}
Our main contribution in this paper is a non-parametric conditional variable screening test using deep neural networks when the ambient dimension of the inputs is potentially very high.  In addition to complementing the variable selection methods, our contribution is also an advance over the aforementioned paper by \cite{barut2016conditional} in that we work with non-parametric regression and our conditioning variables are not known in advance but instead are extracted from the inputs with the help of the factor model as in \cite{fan2022learning} and \cite{FanGu2023factor}.  We note that \cite{horel2020significance} has also conducted screening test based on partial derivatives (in low dimensions).  However their study focused on single-layer neural networks, and hence does not benefit from the power nor reveal the intricacies of deep neural networks in derivative estimation.

In addition to the general procedures outlined above, we also make the following contributions which could be of independent interest:
\begin{enumerate}[wide, labelwidth=!, labelindent=0pt, label=(\alph*)]
	\item[1)]
	We rigorously derive the size and power of our proposed {test statistics employing non-linear functionals of} truly deep neural networks whose architecture can become more complex as the sample size increases.  For computational ease, we propose a simple variance estimator whose properties we properly characterize.  Our resulting test statistics are straightforward to compute and sufficiently precise to accommodate local alternatives.  {Moreover, we address the aforementioned continuity issue of functionals acting on the derivatives of neural networks by proposing improved estimation of the said derivatives which naturally leads to our next contribution.}
	\item[2)]
We exploit the regression function algorithmically learned by deep neural networks in order to recover high-quality derivative estimation.  To achieve this goal, we employ a smoothing technique on deep neural network estimators to regularize their derivatives.  Moreover, out method is applicable to
regularize the derivatives of alternative machine learning techniques, thus paving the way for their application in derivative estimation.  Last but not least, although we focus on (conditional) marginal screening in the present paper, our method can be generalized easily to test higher-order effect and variable interaction using higher-order and mixed derivatives respectively.
\end{enumerate}

Through simulation studies, we illustrate the favorable size and power performance of our test statistics, and the benefit of the smoothing operation in generating accurate derivative estimators.  We also highlight the potential of our test statistics as a viable tool for model specification in nonlinear factor models via two empirical applications.

\subsection{Notations, conventions, and manuscript organization}
\label{sec:notation}

Let $\NN$ denote the set of positive integers.  For a vector $\bv = (v_1, \ldots, v_d)^\top \in \RR^d$, we let $\|\bv\|_p = (\sum_{k=1}^d |v_i|^p)^{1/p}$ be the $\ell_p$ norm of $\bv$.  For a matrix $\bA = (a_{j,k})_{1\le j\le l, 1\le k \le m}\in\RR^{l\times m}$, we define the operator norm $\|\bA\|_{\textup{op}} = \max_{\bv\in\RR^m: \|\bv\|_2 = 1} \|\bA\bv\|_2$.  We let $C$ denote an (absolute) constant that may change for each occurrence, and let ``$\lesssim$'' denote an inequality that holds up to such a multiplicative factor $C$; moreover, let $c$, $C$ and $M$ with super/subscripts denote constants with particular (though often non-specified) values.
Limits are taken as $n\to\infty$ unless otherwise stated.  For positive number sequences $(a_n: n\ge 1)$ and $(b_n: n\ge 1)$, we denote $a_n\lesssim b_n$ if there exists a positive constant $C$ such that $a_n/b_n\le C$ (for all $n$), and denote $a_n=\Co(b_n)$ (resp. $a_n\sim b_n$) if $a_n/b_n\to 0$ (resp. $a_n\lesssim b_n$ and $b_n\lesssim a_n$).  We use ``$\rightarrow_{d}$'' to denote convergence in distribution.  Finally, let $\|\cdot\|_{L_2}=\|\cdot\|_{L_2(\PP)}$ denote the $L_2$ norm of the argument function with respect to the measure $\PP$ to be formally introduced in Section~\ref{sec:reg_model}, so $\|f\|_{L_2}= \{ \int f^2 \rmd\PP \}^{1/2}$.  Sections in the supplement are labelled alphabetically.

We organize the remainder of our paper as follows.  Section~\ref{sec:testhypothesis} specifies our regression model and screening test, and develop our derivative estimators and test statistics.  Section~\ref{sec:regression_function} depicts the accompanying theoretical properties of the derivative estimators and test statistics from Section~\ref{sec:testhypothesis}. Section~\ref{sec:simulations} presents a simulation study. Section~\ref{sec:emp_app} applies our test to an empirical example in asset pricing and another in macroeconomics time series.  Section~\ref{sec:conclusion} concludes and suggests several extensions.  Additional results for the empirical application, proofs and supporting details are deferred to the supplementary materials.

\section{Derivative estimator and test statistics}
\label{sec:testhypothesis}

\subsection{A conditional screening test through latent factors}
\label{sec:reg_model}

Our screening test aims to tackle the high-dimensional regime where the ambient dimension $d$ of our observed predictors $\bX\in\RR^d$ can increase with the sample size $n$. Even with the remarkable capacity of deep neural networks to represent complex functions, our task is still infeasible if $d$ grows too fast.
To address this issue, we assume that $\bX$ admits the following factor model \citep{FanGu2023factor}:
\begin{align}
\label{eq:factor_observation}
\bX= \bB \bfF+\bmu,\quad \EE(\bmu|\bfF)=\mathbf{0} ,
\end{align}
where $\bfF \in \RR^r$ is a vector of latent factors, $\bB\in\RR^{d\times r}$ (usually $d\gg r$ for the high-dimensional $\bX$) is the factor loading matrix, and  $\bmu$ is the vector of the idiosyncratic noises.   We let $(Y,\bfF,\bX)$ have joint distribution $\PP$ which  also determines the distribution of $\bmu$ by \eqref{eq:factor_observation}.

The conditional marginal screening is to see whether a component  $X_j$ of  $\bX$ has additional contributions to the response variable $Y$ given $\bfF$.  Then, the problem involves the working regression function $m_0(\cdot):\RR^{r+1}\rightarrow\RR$ defined through
\begin{equation}
\begin{gathered}
 Y = m_0(\bfF,X_j) + \epsilon , \quad \text{where}~m_0(\bfF,X_j) \equiv \EE[Y|\bfF,X_j] .
\end{gathered}
\label{eq:reg_model}
\end{equation}
The additional contribution of variable $X_j$ is measured through the partial derivative
\begin{align*}
    \textstyle m_{0,j} = m_{0,j}(\bff,x_j) \equiv \frac{\partial}{\partial x_j} m_0(\bff,x_j) .
\end{align*}
We impose a blanket notational convention that a function with subscript $j$ denotes the partial derivative with respect to the last argument, which will almost always be $x_j$, while keeping the other argument (in this case, $\bff$) fixed.  Then, to examine whether $X_j$ has additional contribution, we propose the null and alternative  hypotheses
\begin{gather}\label{hypo_null}
  H_0: \text{for all}~\bff, x_j\in\RR^{r+1},~m_{0,j}(\bff,x_j) = 0,~\text{against} \\
  \label{hypo_alt}
  H_A: \text{for some}~\bff, x_j\in\RR^{r+1},~\text{we have}~m_{0,j}(\bff,x_j) \neq 0 .
\end{gather}

We will assume that $(Y,\bfF,\bX)$ is a $\RR\times[-b,b]^{r+d}$ valued random vector.  The function $m_0$ is akin to $\EE[Y|X_j]$ in condition~C in \cite{FanFengSong2011} on \textit{un}conditional non-parametric screening, but here we exercise finer control over the potential role of $X_j$ by introducing $\bff$ in $m_0=m_0(\bff,x_j)$.  Note that {working with} $m_0$ involves estimating the latent factors $\bfF$, which will be discussed in Section~\ref{sec:highd}.  By Remark~\ref{rmk:specification_test} in the supplement, the screening hypotheses for $X_j$ in \eqref{hypo_null} and \eqref{hypo_alt} are also equivalent to a significance test for the $j$-th idiosyncratic term $u_j$.  We deliberately avoid working with the full conditional expectation of $Y$, namely $m_0^*(\bfF,\bX) \equiv \EE[Y|\bfF,\bX]$, because this is likely infeasible if the dimension of $\bX$ is too high {without} additional structure.  In contrast, the total number of variables $r+1$ in our working regression function $m_0$ and its derivative $m_{0,j}$ is much less than $d$, which allows us to circumvent the problem of testing a potentially very large number of coordinates of $\bX$ simultaneously \citep{FanGu2023factor}.  Nevertheless, our model retains its validity even when the regression outcome $Y$ involves more coordinates of $\bX$, as we will explain in Remark~\ref{rmk:misspecification}.

Although we shall focus on the high-dimensional $\bX$ case with growing $d$, our screening test also easily accommodates the low-dimensional $\bX$ case, as we will comment in Section~\ref{sec:lowd}.  At the other end of the spectrum, our findings can be generalized to testing $\bX$ over a fixed or expanding set of coordinates that form a subset of $\{1,\cdots,d\}$, as we will briefly outline in Section~\ref{sec:conclusion}.

\subsection{Initial regression estimator through diversified projection}
\label{sec:highd}

We focus on multi-layer feed-forward neural networks that are fully connected between adjacent layers \cite[p.~75]{AnthonyBartlett1999}.  In this paper we exclusively consider the Rectified Linear Unit (ReLU) activation function defined as $\sigma=\sigma_{\text{ReLU}}(x)=\max \{x, 0\}$ due to its widespread popularity \citep[Section~6.1]{GoodfellowBengioCourville2016deep}.

For now, we consider neural network functions with a generic input dimension $\dbar$.  We finalize the structure of our neural networks by specifying the tuple $(L, \mathbf{k})$ where $L\in\NN$ represents the number of hidden layers and the width vector $\mathbf{k}=(\dbar, k_1, \ldots, k_L, k_{L+1}=1) \in \NN^{L+2}$ specifies the number of nodes (i.e., neurons) in each hidden layer and the input/output dimensions.  More precisely, such a neural network function $f:\RR^\dbar\rightarrow\RR$ with architecture $(L, \mathbf{k})$ is given by
\begin{align}
    \label{eq:NN_function_output}
    \textstyle f(\bx)= \calL_{L+1} \circ \sigmabar_L \circ \calL_L \circ \sigmabar_{L-1} \circ \dots  \circ \calL_2 \circ \sigmabar_1 \circ \calL_1(\bx) ;
\end{align}
here $L_\ell(\bz)=\bw_\ell \bz +\bb_\ell$ is an affine map with weight matrix {$\bw_{\ell}\in\RR^{k_\ell \times k_{\ell-1}}$ if $\ell\geq 2$, $\bw_{1}\in\RR^{k_1 \times \overline{d}}$} and bias vector $\bb_\ell\in\RR^{k_\ell}$, and the function $\sigmabar_\ell:\RR^{k_\ell}\rightarrow\RR^{k_\ell}$ applies the ReLU activation function $\sigma$ entry-wise.  Additionally, we will truncate the (univarite) output of $f$ at a pre-specified constant level $M>0$ with the truncation operator $T_M(x) \equiv \textup{sgn}(x) \min\{|x|, M\}$.  We denote such a collection of truncated $f$ by $\calFn(\dbar)$ where for brevity of notation we only retain the dependence on the input dimension $\dbar$:
\begin{align}
\calFn(\dbar) = \big\{ T_M(f): f\text{ is of the form \eqref{eq:NN_function_output}}~
\text{with $L$ hidden layers and width vector $\mathbf{k}$} \big\},
\label{eq:base_NN_class}
\end{align}
where both $L$ and $\mathbf{k}$ can scale with the sample size $n$.

Next, we address the  latency issue of the unobserved factor $\bfF$.  For a given diversified weight $\bw \in\RR^d$,  \cite{fan2022learning} note  from \eqref{eq:factor_observation} that within the projection $\bw^T \bX$, the projected idiosyncratic terms $\bw^T \bmu$ will be negligible in high-dimension due to the law of the averages, so $\bw^T \bX$ yields an approximate linear combination of $\bfF$.  To yield $r$ factors, we need at least $r$ projections.  Since $r$ is {unknown}, one specifies an upper bound $\rbar$ {of} $r$.

Let $\bW\in\RR^{d\times \rbar}$ be a pre-trained \textit{diversified projection matrix} as termed by Definition~3 in \cite{FanGu2023factor}, and let $\bH\equiv d^{-1}\bW^\top \bB\in\RR^{\rbar\times r}$.
Then, by \eqref{eq:factor_observation}, we have
$$
\wt\bfF \equiv d^{-1}\bW^\top \bX = \bH \bfF + d^{-1}\bW^\top \bmu \approx \bH \bfF  \in\RR^{\rbar},
$$
by the law of the averages.  Hence, $\bH^\dagger \wt\bfF \approx \bfF$ under some appropriate conditions, where $\bH^\dagger$ is the pseudo-inverse of $\bH$.  We call $\wt\bfF$ as the \textit{diversified factor}, which is observable and a proxy of the latent factor $\bfF$.
In practice, acquiring $\bW$ in advance is necessary, either through domain expertise or data-driven methods. For instance, Proposition~1 in \cite{FanGu2023factor} proposes estimating $\bW$ via pretraining, where a tiny portion of size approximately $\log(n)$ of the samples is reserved to {extract the top $\rbar$ principal components in order to} construct a $\bW$ that will satisfy Assumption~\ref{ass:W} with high probability. Henceforth we shall assume that $\bW$ is pre-determined, exogenous and fixed.

With the proxies of the latent factors, for each given $j$, the coordinate of the candidate variable for screening, we can compute $\{Y_i, \wt\bfF_i, X_{i,j}\}_{i=1}^n$ based on the observed data, where $\wt\bfF_i= d^{-1}\bW^\top \bX_i$.  Then we fit the neural network regression
\begin{align}
	\label{eq:def_wh_m_n_highd}
	\textstyle \wh g_n = \argmin_{g\in\calFn(\rbar+1)} \frac{1}{n}\sum_{i=1}^n \{Y_i-g(\wt\bfF_i,X_{i,j})\}^2 = \argmin_{g\in\calFn(\rbar+1)} \PPn \ell(\cdot;g)
\end{align}
where $\ell(y,\wt\bff,x_j;g)=\{y-g(\wt\bff,x_j)\}^2$ is the square loss and $\PPn$ is the empirical distribution.
Intuitively, $\wh g_n$ is a neural network estimation of
\begin{align}
\label{eq:def_g0}
    g_0 = g_0(\wt\bff,x_j) \equiv m_0(\bH^\dagger \bff, x_j).
\end{align}
Indeed, Lemma~\ref{lemma:factor} shows that $g_0$ approximates $m_0$ well.

\subsection{Initial test statistic and its improvements}
\label{sec:test_stat}

To test our screening hypotheses \eqref{hypo_null} and \eqref{hypo_alt}, we start from $\wh g_{n,j}$, which by the convention in Section~\ref{sec:reg_model} is the partial derivative of $\wh g_n$ in \eqref{eq:def_wh_m_n_highd} with respect to $x_j$: $\wh g_{n,j}(\wt\bff,x_j)=\partial\wh g_n(\wt\bff,x_j) / \partial x_j$.  Now, $\wh g_{n,j}$ can be regarded as an estimator of the derivative $m_{0,j}$.  We then consider the following \textit{initial} MGF/exponentially tilted test statistic:
\begin{align}
\label{eq:eta_t_naive}
\eta_t(\wh g_n)=\PPn\{ \exp(t\,\wh g_{n,j}) - 1 \}.
\end{align}
Partly owing to the fact that a MGF uniquely characterizes the distribution of the underlying random variable (specifically when considering the MGF over an interval including zero), MGF-based tests are popular in the literature \citep{EppsSingletonPulley1982,BaringhausEbnerHenze2017}.  For instance, they have been extensively employed in testing (multivariate) normality \citep{EbnerHenze2020}.  Under the null hypothesis $H_0: m_{0,j}=0$, we expect $\wh g_{n,j}$ to be close to zero and so $\eta_t(\wh g_n)$ is also centered around zero.

To enhance the quality of both the derivative estimator $\wh g_n$ and the test statistic $\eta_t(\wh g_n)$, we will further conduct the following refinements sequentially:
\begin{enumerate}[wide, labelwidth=!, labelindent=0pt, label=(\alph*), topsep=0pt,itemsep=-1ex,partopsep=1ex,parsep=1ex,leftmargin =0.2 in]
    \item {\bf Smoothing}.
    \label{improve:smooth}
    As already alluded to in the introduction, in derivative estimation, the straightforward plugin estimator $\wh g_{n,j}$ derived from $\wh g_n$ in \eqref{eq:def_wh_m_n_highd} may perform poorly due to irregularities of the neural network functions.  For example, for the $L$-time iterated sawtooth function $\zeta_L$ (see Lemma~2.4 in \cite{Telgarsky2015} and Figure~\ref{fig:sawtooth}) which is implementable by a neural network of no more than $L$ hidden layers and three nodes per layer,  we have {$\|\dot\zeta_L\|_{L_2}\ge C 2^L \|\zeta_L\|_{L_2}$}, which implies that the quality of derivative estimation can become significantly worse than that of regression function estimation as the depth $L$ increases, a regime precisely of interest for deep neural networks.  This phenomenon results from the increasingly oscillatory behaviour of the neural network functions as the depth increases, and hence is not tied specifically to the ReLU activation function.  To address this issue, we will refine the initial estimator $\wh g_{n,j}$ to obtain a smoothed derivative estimator $\wh g_{n,j}^\us$, which will allow us to recover a faster convergence rate; see Theorem~\ref{thm:m_hat_est_master} and Theorem~\ref{thm:m_check_est_master}; this will in turn address the continuity issue {arising from second-order terms in our MGF test statistics} acting on the derivatives of neural networks.  Throughout this paper, we adhere to the convention that the superscript ``$\us$'' in upright font denotes the smoothed variant of the preceding function.  We will provide more detailed calculation for our observation on the sawtooth function and more comprehensive rationale behind the smoothing operation in Section~\ref{sec:sawtooth}.
    \item{\bf Centering}.
    \label{improve:tweaking}
    We must further refine the initial $\wh g_{n,j}^\us$ {to arrive at $\wc g_{n,j}^\us$ in \eqref{eq:f_check} in order} to center our test statistics.  This debiasing step is analogous to the concept of targeted machine learning in the literature \citep{van2011targeted, VanderLaanRose2018}, and also relates to other research on debiased machine learning \citep{quintas2021riesznet, kennedy2022semiparametric}.  However, we will  justify this refinement step independently without directly referencing the targeted/debiased machine learning literature.
    \item{\bf Truncation}.
    \label{improve:truncate}
   To treat the technical possibility of the unboundedness of $\wc g_{n,j}^\us$, we will  truncate $\wc g_{n,j}^\us$ appropriately in our final test statistics.
    \item{\bf Uniformity over $t$}.
    \label{improve:uniform}
  We will extend the fixed-$t$ statistic to the statistics aggregated over a range of $t$ to enhance the power of the test \citep{fan2001generalized, fan2015power}.
\end{enumerate}

We focus solely on the smoothing step~\ref{improve:smooth} in this section, and will introduce step~\ref{improve:tweaking} in Section~\ref{sec:tweaking}, and both \ref{improve:truncate} and \ref{improve:uniform} in Section~\ref{sec:uniform}.

We first describe the smoothing operation applied to a generic function $g:\mathbb{R}^{\bar{r}+1}\rightarrow\mathbb{R}$ (or analogously $m:\mathbb{R}^{r+1}\rightarrow\mathbb{R}$). Let $K$ be a univariate continuously differentiable kernel  supported on $[-1,1]$ with derivative $\dot K$.
Then, for the generic function $g$, denote its smoothed version in the variable $x_j$  by $g^\us=g^\us(\wt\bff, x_j)
= \int g(\wt \bff, z) K_h(x_j - z)dz$, where $K_h(\cdot)=K(\cdot/h) / h$ with a bandwidth parameter $h$.  Now, the smoothed function $g^\us$ becomes differentiable in $x_j$ everywhere with the partial derivative given by
\begin{align}
\label{eq:smooth_master}
\textstyle g_j^\us(\wt\bff, x_j) = \frac{\partial}{\partial x_j} g^\us(\wt\bff, x_j) = \int_{x_j-h}^{x_j+h} g(\wt\bff,z) \dot K_h(x_j-z)\rmd z = \frac{1}{h} \int_{-1}^1 g(\wt\bff,x_j-ah) \dot K(a) \rmd a.
\end{align}
Accordingly, we let our initial smoothed derivative estimator based on $\wh g_n$ be $\wh g_{n,j}^\us$.  In practice, we can select the bandwidth $h$ through cross-validation; see our Remark~\ref{rmk:bandwidth}.

\subsection{Estimating the score function}
\label{sec:Riesz_est}
To test the screening hypotheses \eqref{hypo_null} and \eqref{hypo_alt}, we proceed to estimate the score function corresponding to the statistic $\eta_t^\us(\wh g_n)\equiv\PPn\{ \exp(t\,\wh g_{n,j}^\us) - 1 \}$ now refined over \eqref{eq:eta_t_naive}.  This score function $\alpha_{t,n}^*:\RR^{\rbar+1}\rightarrow\RR$ satisfies, under the null and for any $L_2(\PP)$-integrable $\alpha$, the equality $\int \alpha\alpha_{t,n}^* \rmd\PP = \int_{\Omega_h} \alpha_j^\us(\wt\bff,x_j) \rmd\PP$, where $\alpha_j^\us$ is the smoothed derivative of $\alpha$ and $\Omega_h=\{\omega\in\Omega: X_j(\omega)\in\calB_h\}$ {is the interior sample space corresponding to} the interior set $\calB_h=[-b+h,b-h]$.  (Such an $\alpha_{t,n}^*$ exists because it is the Riesz representer of the directional derivative functional of our test statistic, as we will formally explain in Section~\ref{sec:variance_prep}.)  Therefore, the mean-square error (MSE) can be expressed as
\begin{align}
	&\EE(\alpha-\alpha_{t,n}^*)^2(\wt\bfF,X_j) = \EE \alpha^2(\wt\bfF,X_j) - 2 \EE (\alpha \alpha_{t,n}^*)(\wt\bfF,X_j) + \EE \alpha_{t,n}^{*2}(\wt\bfF,X_j) \nonumber \\
	& \textstyle {\stackrel{\text{under}~H_{0}}{=}} \textstyle \EE \alpha^2(\wt\bfF,X_j) - 2 \int_{\Omega_h} \alpha_j^\us(\wt\bff,x_j) \rmd\PP + \EE \alpha_{t,n}^{*2}(\wt\bfF,X_j).
	\label{eq:Rnull}
\end{align}
Discarding the last term $\EE[\alpha_{t,n}^{*2}(\widetilde{\mathbf{F}},X_j)]$, which does not depend on $\alpha$, the MSE can be estimated by the empirical loss
\begin{align}
    \label{eq:Rhat}
    \textstyle \Rnullhat(\alpha) = \frac{1}{n} \sum_{i=1}^n \alpha^2(\wt\bfF_i, X_{i,j}) - 2 \frac{1}{n} \sum_{i\in\calI_h} \alpha_j^\us(\wt\bfF_i, X_{i,j}),
\end{align}
where $\calI_h=\{i\in\{1,\dots,n\}: X_{i,j}\in\calB_h\}$ is the interior index set.
Hence, we estimate $\alpha_{t,n}^*$ by
\begin{align}
\label{eq:alpha_hat}
    \textstyle \wh\alpha_n = \argmin_{\alpha\in\calFn(\rbar+1)} \Rnullhat(\alpha),
\end{align}
which is a neural network estimator of the score function.

We tailor the loss function $\Rnullhat$ and the estimator $\wh\alpha_n$ to the null hypothesis \eqref{hypo_null}.  One notable computational advantage is that $\Rnullhat$ and hence $\wh\alpha_n$ are independent of both the response and the value of $t$: A single optimization suffices to yield $\wh\alpha_n$ for all $t$.  However, the potential trade-off is a bias of $\wh\alpha_n$ with respect to $\alpha_{t,n}^*$ induced under the alternative hypothesis, which will be characterized in Theorem~\ref{thm:Riesz_est}.  Numerically, computing the right-hand side in \eqref{eq:Rhat} is straightforward, as detailed in Lemma~\ref{lemma:alpha_j_numerical} in the supplement.

\subsection{Centering test statistics}
\label{sec:tweaking}

Our refined statistic $\eta_t^\us(\wh g_n)$ needs to be further centered to ensure asymptotic normality.  To achieve this, ideally we adjust $\wh g_n$ in the direction of the score function $\alpha_{t,n}^*$ to attain a smaller loss than in \eqref{eq:def_wh_m_n_highd}.  With the estimate $\wh\alpha_n$ of $\alpha_{t,n}^*$ given by \eqref{eq:alpha_hat}, one naturally uses
\begin{align}
	\label{eq:f_check}
	\wc g_n = \wh g_n + \deltathat \wh \alpha_n
\end{align}
to arrive at a debiased statistic $\eta_t^\us(\wc g_n)$, where
 \begin{align}
	\textstyle \deltathat & \textstyle = \frac{1}{\sum_{i=1}^n \wh\alpha_n^2(\wt\bfF_i,X_{i,j})}  \sum_{i=1}^n \left\{ Y_i - \wh g_n(\wt\bfF_i,X_{i,j}) \right\} \wh\alpha_n(\wt\bfF_i,X_{i,j})
	\label{eq:k_check}
\end{align}
is chosen to minimize $\PPn \ell(\cdot;\wh g_n + \deltat \wh\alpha_n)$ in $\deltat$, as derived in Section~\ref{sec:proof_prop:tweaking}.   We can interpret this update step as a form of the targeted machine learning \citep{VanderLaanRose2018} specialized to our deep neural network context.
In Section~\ref{sec:tweaking_proof} we will show that $\wc g_n$ indeed satisfies an approximate minimization condition: for an infinitesimal stepsize $\rinf$,
\begin{align}
	\label{eq:min_con}
	\PPn\ell(\cdot; \wc g_n)-\PPn\ell(\cdot; \wc g_n \pm\rinf\alpha_{t,n}^*) \le \rinf b_n
\end{align}
for a small tolerance $b_n$.  This will in turn center our test statistics.

\subsection{Final test statistic and its uniform extensions}
\label{sec:uniform}

To alleviate the technical possibility of the unboundnessed of our derivative estimator, we let $\Psi: \mathbb{R} \rightarrow \mathbb{R}$ be a bounded truncation function that also satisfies Assumption~\ref{ass:truncation}; we provide an example of such a truncation function below the assumption.  We then define the final (fixed-$t$) test statistic to be $\wc\eta_t^\us(\wc g_n)$ where the operator $\wc\eta_t^\us(\cdot)$ acts on functions $g:\RR^{\rbar+1}\rightarrow\RR$ as
\begin{gather}
\textstyle \wc\eta_t^\us(g) = \frac{1}{n}\sum_{i\in\calI_h} \left[ \exp\{t\,\Psi(g_j^\us(\wt\bfF_i, X_{i,j})) \} - 1 \right] .
\label{eta_t_s_wc}
\end{gather}
Its standard error is estimated as $t \|\wh\epsilon \wh\alpha_n\|_{L_2(\PPn)}= t [ \sum_{i=1}^n \{ \wh\epsilon_i \wh\alpha_n(\wt\bfF_i,X_{i,j}) \}^2 / n ]^{1/2}$, where $\wh\epsilon_i=Y_i-\wc g_n(\wt\bfF_i,X_i)$ is the $i$-th residual.
This leads to a level-$\alpha$ test of the null hypothesis~\eqref{hypo_null} with the critical region
\begin{align}
    \label{test:fixed_t}
    \left| \frac{ \sqrt{n} }{t} \frac{\wc\eta_t^\us(\wc g_n) }{ \|\wh\epsilon \wh\alpha_n\|_{L_2(\PPn)} } \right| > z_{1-\alpha/2}
\end{align}
where $z_{1-\alpha/2}$ is the $1-\alpha/2$ quantile of a standard normal distribution; the choice of this critical value is supported by the first part of \eqref{eq:H0_normality_Ln} in Theorem~\ref{thm:main_null_master}.

To increase the power of our test, we aggregate the ensemble of test statistics $\wc\eta_t^\us(\wc g_n)$ over a range of $t$. We consider two types of aggregation: a sup-statistic and an integrated weighted square-statistic (or simply a square-statistic).  In both cases, let $\delta > 0$ be a small but fixed constant, let $T > \delta$ be a potentially large but also fixed constant, and define  $\calT_\delta= [-T, -\delta] \cup [\delta, T]$.  The two aforementioned aggregations of test statistics and their associated critical regions are as follows.
\begin{enumerate}[wide, labelwidth=!, labelindent=0pt, label=(\alph*), leftmargin =0.2 in]
\item[1)] { \bf Sup-statistic}:
\begin{equation}\label{sup}
 \wh Z  = \sqrt{n} \sup_{t\in\calT_\delta} \left| \dfrac{ \wc\eta_t^\us(\wc g_n ) }{t} \right|, \quad \text{with critical region:}~\dfrac{ \wh Z }{ \|\wh\epsilon \wh\alpha_n\|_{L_2(\PPn)} } > z_{1-\alpha/2} .
\end{equation}

\item[2)] { \bf Square statistic}:
\begin{equation}\label{square}
\wh\chi^2  = \dfrac{n}{ \int_{\calT_\delta} w(t) \rmd t } \int_{\calT_\delta} \dfrac{1}{t^2} \left\{ \wc\eta_t^\us(\wc g_n) \right\}^2 w(t) \rmd t, \quad \text{with critical region}~\dfrac{ \wh\chi^2 }{ \|\wh\epsilon \wh\alpha_n\|_{L_2(\PPn)}^2 } > \chi^2_{1,1-\alpha}
\end{equation}
where $w$ is an integrable weight function such as  $w(t)=\exp(-\beta t^2)$ for a given $\beta>0$, and
$\chi^2_{1,1-\alpha}$ is the $1-\alpha$ quantile of a chi-square distribution with one degree of freedom.
\end{enumerate}
The critical values are supported by the last two parts of \eqref{eq:H0_normality_Ln} in Theorem~\ref{thm:main_null_master}.

\subsection{Summary of proposed methods}

The major implementation steps of our partial derivative estimation and conditional screening test procedures in Section~\ref{sec:testhypothesis} are summarized in Algorithm~\ref{algorithmhighd}.
{\small \spacingset{1.5}
\begin{algorithm}
\hrulefill
\caption{Partial derivative estimation \& conditional screening test by deep neural network}
\SetKwInput{KwData}{Input}
\KwData{Observed samples $Y_i, \bX_i \in \RR^d$, $i\in \{1,\cdots, n\}$, and an exogenous diversified projection matrix $\bW\in\RR^{d\times\rbar}$; specification of the coordinate $j$ for screening.}
\KwResult{Preliminary and refined regression estimators $\wh g_n(\cdot)$ and $\wc g_n(\cdot)$; their  smoothed derivative estimators $\wh g_{n,j}^\us(\cdot)$ and $\wc g_{n,j}^\us(\cdot)$; $\wh\alpha_n(\cdot)$, and test results.}
\begin{enumerate} \itemsep -0.15in
\item
Compute $\wt\bfF_i = d^{-1} \bW^\top \bX_i$, for all $i \in \{1,\dots,n\}$.
\item
Estimate the regression function $m_0$ in \eqref{eq:reg_model} by the neural network estimator $\wh g_n$ from \eqref{eq:def_wh_m_n_highd}.  {Conduct cross validation (see Remark~\ref{rmk:bandwidth}), or otherwise, to find a smoothing bandwidth $h$.} \\
\item
Define the loss function $\Rnullhat(\cdot)$ as
\begin{align*}
    \textstyle \Rnullhat(\alpha) = \frac{1}{n} \sum_{i=1}^n \alpha^2(\wt\bfF_i,X_{i,j}) - 2 \frac{1}{n} \sum_{i\in\calI_h} \alpha_j^\us(\wt\bfF_i,X_{i,j}) ,
\end{align*}
which can be approximated via Lemma~\ref{lemma:alpha_j_numerical},  and let the estimator $\wh\alpha_n$ of $\alpha_{t,n}^*$ be
\begin{align*}
    \textstyle \wh\alpha_n = \argmin_{\alpha\in\calFn(\rbar+1)} \Rnullhat(\alpha).
 \end{align*}
\item
   {Return $ \wc g_n= \wh g_n + \deltathat {\wh\alpha_n}$, where}
\begin{align*}
\textstyle \deltathat = \frac{1}{\frac{1}{n} \sum_{i=1}^n {\wh\alpha_n}^2(\wt\bfF_i,X_{i,j})} \frac{1}{n} \sum_{i=1}^n \left\{ Y_i - \wh g_n(\wt\bfF_i,X_{i,j}) \right\} {\wh\alpha_n}(\wt\bfF_i,X_{i,j}) .
\end{align*}
\item
Let $\|\wh\epsilon \wh\alpha_n\|_{L_2(\PPn)}$ be the standard deviation estimator of our test statistic $\wc\eta_t^\us(\wc g_n)$ (with $\wc\eta_t^\us$ from \eqref{eta_t_s_wc}); specifically, $\|\wh\epsilon \wh\alpha_n\|_{L_2(\PPn)}= [ \frac{1}{n} \sum_{i=1}^n \{ \wh\epsilon_i \wh\alpha_n(\wt\bfF_i,X_{i,j}) \}^2 ]^{1/2}$ where $\wh\epsilon_i=Y_i-\wc g_n(\wt\bfF_i,X_i)$ is the $i$-th residual.  Then, for a prescribed significance level $0<\alpha<1$, reject the null hypothesis \eqref{hypo_null} in our conditional screening hypotheses: \vspace*{-0.15in}
\begin{itemize}\itemsep -0.15in
    \item[a)]
    For the fixed-$t$ test, if \eqref{test:fixed_t} holds;
\item[b)]
For the sup test, if the decision rule in \eqref{sup} holds;
\item[c)]
For the square test, if the decision rule in \eqref{square} hold. \vspace*{-0.2in}
\end{itemize}
\end{enumerate}
\label{algorithmhighd}
\hrulefill
\end{algorithm}
}

\section{Theoretical properties of estimators and tests}
\label{sec:regression_function}

\subsection{Initial regression and smoothed derivative estimators}
\label{sec:init_est_theory}

We first collect the assumptions most relevant to the construction and rates of our estimators $\wh g_n$ and $\wh g_{n,j}^\us$; we defer the more technical assumptions to Section~\ref{sec:additional_ass_sec_3}.
In particular, Assumption~\ref{ass:DGP} imposes mild conditions on the data generating processes, and Assumption~\ref{ass:NN_scaling} specifies the function class for $m_0$ and its associated neural networks.

\begin{assumption}[Data generating processes]\label{ass:DGP}
\begin{enumerate*}[label=(\roman*)]
\item\label{ass:DGP:con_1}
$(\bfF, X_j)$ take value on the bounded support $[-b,b]^r\times[-b,b]$;
\item\label{ass:DGP:con_2}
{The dimension $\rbar$ of the diversified factor $\wt\bfF$ is a fixed constant and satisfies $\rbar\ge r$;}
\item\label{ass:DGP:con_3}
Regression function $m_0^*(\bff,\bx)$, defined below \eqref{hypo_alt}, is bounded in magnitude by $M_\infty$ that satisfies $M_\infty\le M$ for $M$ the truncation level in \eqref{eq:base_NN_class}.
\end{enumerate*}
\end{assumption}

\begin{assumption}[Function class for $m_0$ and neural network scaling]
\label{ass:NN_scaling}
Regression function $m_0$ belongs to the hierarchical composition model $\calH(r+1,l,\calP,C_\calH)$ in Definition~\ref{def:Hier_comp} \citep{schmidt2020nonparametric,KohlerLanger2021}.  The tuple $(L, \mathbf{k}=(\rbar+1, k_1, \ldots, k_L,1))$ for the structure of $\calFn(\rbar+1)$ in \eqref{eq:base_NN_class} satisfies $k_1=\dots=k_L\equiv k_{\textup{0}}$ for some $k_{\textup{0}}\in\NN$ and $L \cdot k_{\textup{0}} \sim n^{\frac{1}{4\kappa+2}} \log^{\frac{4\kappa-1}{2\kappa+1}}(n)$ where $\kappa$ is the dimension-adjusted degree of smoothness defined above Definition~\ref{def:Hier_comp}.
\end{assumption}

\begin{assumption}[Idiosyncratic terms]\label{asumidio}
The idiosyncratic terms $\bmu\in\RR^d$ satisfy $\sum_{j=1}^d\EE[u_j^2]\lesssim d_u$ for $d_u\leq d$ and a weak dependence condition $\sum_{j\neq{j}^{\prime}}|\EE\left[u_{j}u_{j^{\prime}}\right]|\lesssim d_u$.
\end{assumption}

In Theorem~\ref{thm:m_hat_est_master} we establish the convergence of $\wh g_n-m_0$ for regression function estimation and of $\wh g_{n,j}^\us - m_{0,j}$ for derivative estimation. The following rates will appear:
\begin{align}
    \label{eq:nu_n_concrete}
    \pn = \cp n^{\frac{1}{2\kappa+1}} \log^{2\frac{4\kappa-1}{2\kappa+1}+{2}}(n), \quad \nun = n^{-\frac{\kappa}{2\kappa+1}} \log^{\frac{6\kappa}{2\kappa+1}}(n) , \quad  \deltaf= \{\rbar d_u/d^2\}^{1/2}
\end{align}
where $\cp$ is the constant appearing in \eqref{eq:V_calF_bound}, $d_u$ appears in Assumption~\ref{asumidio}, and $\deltaf \to 0$ as $d\to \infty$.  Moreover, because the target of the \textit{smoothed} derivative estimator $\wh g_{n,j}^\us$ is the \textit{smoothed} derivative $m_{0,j}^\us$, not $m_{0,j}$, naturally a bias
\begin{align*}
    \textstyle r_{\textup{b},m,j} \equiv \{ \int_{\Omega_h} ( m_{0,j}^\us-m_{0,j})^2\,\rmd\PP\}^{1/2} = \left\{ \int_{\Omega_h} [ \int_{-h}^h \{ m_{0,j}(\bff,x_j-z) - m_{0,j}(\bff, x_j)\} K_h(z) \rmd z ]^2\rmd\PP \right\}^{1/2}
\end{align*}
occurs.  However, this bias vanishes under the null in \eqref{hypo_null} because there $m_{0,j}=0$ everywhere.

\begin{theorem}
\label{thm:m_hat_est_master}
Under Assumptions~\ref{ass:DGP} to \ref{asumidio} and \ref{assum_fun_class} to \ref{ass:kernel}, for $\pn$, $\nun$ and $\deltaf$ in \eqref{eq:nu_n_concrete}, on an event $\calA_m'$ with $\PP(\calA_m') \ge 1 - C {\exp(- \pn )}$, and for a constant $c_{m,1}' $, the initial regression function estimator $\wh g_n$ from \eqref{eq:def_wh_m_n_highd} satisfies
\begin{align}
\label{eq:rn1}
\textstyle \left[ \int \{ \wh g_n(\wt\bff,x_j ) - m_0(\bff, x_j) \}^2 \rmd\PP \right]^{1/2} \le c_{m,1}' ( \nun + \deltaf ) \equiv r_{m,n} .
\end{align}
In addition, its associated smoothed derivative estimator $\wh g_{n,j}^\us$ satisfies, on the same event $\calA_m'$ and with a possibly different constant $c_{m,2}'$,
\begin{align}
\textstyle \left[ \int_{\Omega_h} \{ \wh g_{n,j}^\us(\wt\bff, x_j)- m_{0,j}(\bff, x_j) \}^2 \rmd\PP \right]^{1/2} \le c_{m,2}'  h^{-1} ( \nun + \deltaf ) + r_{\textup{b},m,j} .
\label{eq:rn1_deriv}
\end{align}
\end{theorem}

The proof of Theorem~\ref{thm:m_hat_est_master} is deferred to Section~\ref{sec:proof_thm:m_hat_est_master}. With the rate $r_{m,n}$ in \eqref{eq:rn1} for the regression estimator, the rate \eqref{eq:rn1_deriv} for the smoothed \textit{derivative} estimator follows from Lemma~\ref{lemma:derivative_bound_via_original} that applies to the smoothed derivatives of all estimators.  Next, the rate $r_{m,n}$ consists of two parts: $\nun$ represents the combined effect of stochastic error and neural network approximation bias, while $\deltaf$ represents the error induced by relying on the diversified factors in place of the latent factors. Theorem~\ref{thm:m_hat_est_master} also covers as a special example the low-dimensional $\bX$ case in Section~\ref{sec:lowd} by simply setting $\deltaf=0$.

\subsection{Estimating the asymptotic variance}
\label{sec:variance}

\subsubsection{Definition of the score function}
\label{sec:variance_prep}

In this section, to complete our discussion in Section~\ref{sec:Riesz_est}, we formally introduce the score function $\alpha_{t,n}^*$ of our test statistics and the bias of $\wh\alpha_n$ relative to $\alpha_{t,n}^*$.  Define the population version of the operator $\wc\eta_t^\us(g)$ in \eqref{eta_t_s_wc} as
\begin{align}
\label{labeltildeg}
\textstyle \wt\eta_t^\us(g) = \int_{\Omega_h} [ \exp\{ t\,\Psi(g_j^\us(\wt\bff, x_j)) \} - 1 ] \rmd\PP ,
\end{align}
and its associated directional derivative (in the direction $v=v(\wt\bff,x_j)$)
\begin{align*}
    \textstyle \frac{\partial \wt\eta_t^\us(g_0)}{\partial g}[v] = \frac{\partial \wt\eta_t^\us(g_0+\tau v)}{\partial\tau}\Big|_{\tau=0} .
\end{align*}
By \eqref{eq:target_functional_derivative_master} in the proof of Lemma~\ref{lem:In2_2_master}, $\frac{\partial \wt\eta_t^\us(g_0)}{\partial g}[v] / t$ is a bounded linear functional of $v$.  Hence, there exists a \textit{Riesz representer} that becomes our $\alpha_{t,n}^*$ and that satisfies $\frac{\partial \wt\eta_t\left(g_0\right)}{\partial g}[v] / t = {\int v  \alpha_{t,n}^* \rmd\PP}$ for all $v$ \cite[Theorem~3.4, Chapter~1]{Conway1990}.  In particular, by \eqref{eq:target_functional_derivative_master}, the effect of $\alpha_{t,n}^*$ is captured explicitly as
\begin{align}
& \textstyle \int v \alpha_{t,n}^* \rmd\PP = \frac{1}{t} \frac{\partial \wt\eta_t^\us(g_0)}{\partial g}[v] = \int_{\Omega_h} \exp\{t \,g_{0,j}^\us(\wt\bff,x_j) \} v_j^\us(\wt\bff,x_j) \rmd\PP \textstyle {\stackrel{\text{under}~H_{0}}{=}} \int_{\Omega_h} v_j^\us(\wt\bff,x_j) \rmd\PP.
\label{eq:Riesz_H0_master}
\end{align}
Because the loss $\Rnullhat$ relies on the last step of \eqref{eq:Riesz_H0_master}, $\wh\alpha_n$ consistently estimates $\alpha_{t,n}^*$ under the null, as confirmed by Theorem~\ref{thm:Riesz_est}.  However $\wh\alpha_n$ may not be consistent for $\alpha_{t,n}^*$ under the alternative.  Instead, it is not hard to show that
it is consistent for a population limit $\alphanull$ which exists and satisfies $\int v \alphanull \rmd\PP = \int_{\Omega_h} v_j^\us(\wt\bff,x_j) \rmd\PP$ for all $v$. See the remark below Eq.~\eqref{eq:target_functional_derivative_master} for details.

To obtain an intuitive idea of what the Riesz representer $\alpha_{t,n}^*$ could look like, assume that $\wt\bfF, X_j$ admit a joint density {$p(\wt\bff,x_j)$} that is differentiable in $x_j$ with the derivative being {$p_j(\wt\bff,x_j)$}.  Then, by Lemma~\ref{lem:Riesz_analytic}, under the null hypothesis and in the limit $h\rightarrow 0$, $\alpha_{t,n}^*(\wt\bff,x_j)= - p_j(\wt\bff,x_j) / p(\wt\bff,x_j)$; in this case, if further $\wt\bfF, X_j$ are jointly Gaussian {on their support}, then $\alpha_{t,n}^*(\wt\bff,x_j)$ is simply a linear function in $\wt\bff$ and $x_j$.

\subsubsection{Convergence rate for the score function estimator}
\label{sec:variance_thm}

To ensure the rate of $\wh\alpha_n$, our Assumption~\ref{ass:u_t_n} mirrors our earlier Assumptions~\ref{ass:DGP} and \ref{ass:NN_scaling} for estimating the regression function $m_0$ with deep neural networks and in particular imposes a hierarchical composition model on $\alphanull$ that takes the diversified predictors $(\wt\bfF,X_j)$ as arguments.  Then, Assumption~\ref{ass:alpha_h} places a mild condition on the smoothing bandwidth.
As in Section~\ref{sec:init_est_theory}, we defer the more technical assumptions to Section~\ref{sec:additional_ass_sec_3}.

\begin{assumption}[Function class and neural network scaling for estimating $\alphanull$ and $\alpha_{t,n}^*$]\label{ass:u_t_n}
The function $\alphanull$ is bounded in magnitude by $M_\infty$, and when we restrict the support to $[-c_b b, c_b b]^{\rbar}\times[-b,b]$ for a constant $c_b>0$, belongs to the hierarchical composition model $\calH(\rbar+1,l,\calP,C_\calH)$ on the same support and (without loss of generality) with the same $l,\calP,C_\calH$ as in Assumption~\ref{ass:NN_scaling}.  The class $\calFn(\rbar+1)$ in \eqref{eq:alpha_hat} satisfies the same structural scaling as in Assumption~\ref{ass:NN_scaling}.
\end{assumption}

\begin{assumption}[Rate of bandwidth]
\label{ass:alpha_h}
The bandwidth $h$ satisfies $h\ge 1/(\sqrt{n}\nun)$.
\end{assumption}

To characterize the bias of $\alphanull$ with respect to $\alpha_{t,n}^*$ under the alternative, define the random variable that represents the signal of the alternative as
\begin{align}
\label{eq:Z_t}
Z_{t,j,h} = [ \exp\{t\,g_{0,j}^\us(\wt\bfF,X_j)\} - 1 ] \ind\{X_j\in\calB_h\} .
\end{align}
Note that $Z_{t,j,h}$ always equals zero under the null hypothesis.

\begin{theorem}
\label{thm:Riesz_est}
Under Assumptions~\ref{ass:DGP} to \ref{ass:alpha_h} and \ref{assum_fun_class} to \ref{ass:truncation}, for $\pn$, $\nun$ and $\deltaf$ in \eqref{eq:nu_n_concrete}, on an event $\calA_\alpha$ with $\PP(\calA_\alpha) \ge 1 - C \exp(-\pn)$, and for a constant $c_{\alpha,1}$,
\begin{align}
&\textstyle \left[ \int (\wh\alpha_n - \alphanull)^2(\wt\bff,x_j ) \rmd\PP \right]^{1/2} \le c_{\alpha,1} \{ ( h^{-1} \nun + \deltaf ) \wedge M \} \equiv \rnull .
\label{eq:r_n2_null}
\end{align}
Moreover, for a constant $c_{\alpha,2}$ not dependent on $t$, the bias of $\alphanull$ relative to $\alpha_{t,n}^*$ is
\begin{align}
    \textstyle \forall t\in\RR, \quad \left[ \int (\alphanull - \alpha_{t,n}^*)^2(\wt\bff,x_j ) \rmd\PP \right]^{1/2} \le c_{\alpha,2} ( h^{-1} \| Z_{t,j,h} \|_{L_2} \wedge M ) \equiv  r_{\alpha,t,\ub} .
    \label{eq:r_alpha}
\end{align}
Consequently, on the event $\calA_\alpha$,
\begin{align}
&\textstyle \forall t\in\RR, \quad \left[ \int (\wh\alpha_n - \alpha_{t,n}^*)^2(\wt\bff,x_j ) \rmd\PP \right]^{1/2} \le \rnull + r_{\alpha,t,\ub} \equiv r_{\alpha,t,n}.
\label{eq:r_n2}
\end{align}
\end{theorem}

The proof of Theorem~\ref{thm:Riesz_est} is deferred to Section~\ref{sec:proof_thm:Riesz_est}. Compared with the rate \eqref{eq:rn1} in Theorem~\ref{thm:m_hat_est_master} for the regression estimator $\wh g_n$, the rates for the score function estimator $\wh\alpha_n$ mainly differ in two aspects: first, a factor $h^{-1}$ precedes $\nun$ which is the consequence of the smoothing operation in the loss function $\Rnullhat$ in \eqref{eq:Rhat}; second, under the alternative hypothesis a bias $h^{-1} \| Z_{t,j,h} \|_{L_2}$ in \eqref{eq:r_alpha} is induced relative to $\alpha_{t,n}^*$.

\subsection{Centering and adjusted estimator}
\label{sec:tweaking_proof}

Recall from Section~\ref{sec:tweaking} that the adjusted estimator $\wc g_n$ in \eqref{eq:f_check} was introduced to center our test statistics through a suitable minimization condition \eqref{eq:min_con}.  In this section Proposition~\ref{prop:tweaking} first shows that a small $\deltathat$ is sufficient to arrive at $\wc g_n$ that satisfies \eqref{eq:min_con} with a small $b_n$.

\begin{assumption}[Centering test statistics]
\label{ratexx}
$n$ is large enough such that
\begin{enumerate*}[label=(\roman*)]\label{ratexx:whole}
\item\label{ratexx:con_1}
$c_{\ut,1} (\deltaf + \nun/h) \le 1$ for a large enough constant $c_{\ut,1}$;
\item\label{ratexx:con_2}
$c_{\ut,2} r_{m,n} \le 1$ for the constant $c_{\ut,2}$ in Proposition~\ref{prop:tweaking}; also, let {the} infinitesimal positive sequence $\rinf$ satisfy $\rinf=\Co(r_{m,n} \nun)$.
\end{enumerate*}
\end{assumption}

\begin{proposition}
\label{prop:tweaking}
Under Assumptions~\ref{ass:DGP} to \ref{ratexx}\ref{ratexx:con_1} and  \ref{assum_fun_class} to \ref{ass:truncation},
on an event $\calA_{\ut,1}$ with $\PP(\calA_{\ut,1}) \ge 1 - C \exp(-\pn)$, for a constant $c_{\ut,2}$, $\deltathat$ given by \eqref{eq:k_check} satisfies $|\deltathat| \le c_{\ut,2} r_{m,n}$.  If furthermore Assumption~\ref{ratexx}\ref{ratexx:con_2} holds, then on an event $\calA_{\ut,2}$ with $\PP(\calA_{\ut,2}) \ge 1 - C \exp(-\pn)$,
for a constant $c_{\ut,3}$, condition \eqref{eq:min_con} holds uniformly at all $t\in\calT_\delta$ and $b_n=c_{\ut,3} r_{m,n} (\rnull + r_{\alpha,t,\ub})$.
\end{proposition}

The proof of Proposition~\ref{prop:tweaking} is deferred to Section~\ref{sec:proof_prop:tweaking}.  In the proposition, under the null the tolerance $b_n=r_{m,n} \rnull$, which is faster than $n^{-1/2}$ under appropriate conditions and will ensure the asymptotic normality of our test statistics.  Under the alternative, $b_n$ is not necessarily faster than $n^{-1/2}$, but is still fast enough to ensure {consistency under the local alternatives}.

We will start working mostly with the adjusted estimator $\wc g_n$ from now on.   The next theorem is the counterpart of Theorem~\ref{thm:m_hat_est_master} and shows that $\wc g_n$ and its smoothed derivative estimator $\wc g_{n,j}^\us$ maintain convergence rates similar to the unadjusted ones.
\begin{theorem}
\label{thm:m_check_est_master}
Under Assumptions~\ref{ass:DGP} to \ref{ratexx} and \ref{assum_fun_class} to \ref{ass:truncation}, for $\pn$, $\nun$ and $\deltaf$ in \eqref{eq:nu_n_concrete}, on an event $\calA_m$ with $\PP(\calA_m) \ge 1 - C \exp(-\pn)$, and for a constant $c_{m,1}$,
\begin{align}
\label{eq:m_check_rate_master}
\textstyle \left[ \int \{ \wc g_n(\wt\bff,x_j ) - m_0(\bff,\bx ) \}^2 \rmd\PP \right]^{1/2} \le c_{m,1} ( \nun + \deltaf ) .
\end{align}
In addition, its smoothed derivative estimator $\wc g_{n,j}^\us$ satisfies, on the same event $\calA_m$ and with a possibly different constant $c_{m,2}$,
\begin{align}
\textstyle \left[ \int_{\Omega_h} \{ \wc g_{n,j}^\us(\wt\bff, x_j)- m_{0,j}(\bff, x_j) \}^2 \rmd\PP \right]^{1/2} \le c_{m,2}  h^{-1} ( \nun + \deltaf ) + r_{\textup{b},m,j} .
\label{eq:m_check_rate_master_deriv}
\end{align}
\end{theorem}
The proof of Theorem~\ref{thm:m_check_est_master} is deferred to Section~\ref{sec:proof_thm:f_check_est}.

\subsection{Properties of the conditional screening tests}
\label{sec:test}

This section gives the result on the asymptotic null distributions of the test statistics.  The results for consistency against the local alternatives are given in Section~\ref{sec:alter_supp}.

\begin{theorem}
\label{thm:main_null_master}
Suppose that Assumptions~\ref{ass:DGP} to \ref{ratexx} and \ref{assum_fun_class} to \ref{ass:truncation} hold, and in addition condition
\begin{enumerate*}[label=($\ast$)]
\item\label{thm:main_null_master:con_1}
holds: $\deltaf + h^{-2} (\nun^2+\deltaf^2) = \Co(n^{-1/2})$.
\end{enumerate*}
Then,  with $\calT_\delta$ and tests given in Section~\ref{sec:uniform}, we have
\begin{align}
\label{eq:H0_normality_L2}
    \text{for all fixed}~t\in\calT_\delta, \frac{ \sqrt{n} }{t} \frac{\wc\eta_t^\us(\wc g_n) }{\| \epsilon \alphanull\|_{L_2} }\rightarrow_{d}\calZ , \quad \frac{\wh Z}{\| \epsilon \alphanull\|_{L_2}} \rightarrow_{d} |\calZ|, \quad  \frac{\wh\chi^2}{\| \epsilon \alphanull\|_{L_2}^2} \rightarrow_{d} \chi_1^2 ,
\end{align}
where $\calZ$ stands for a standard normal random variable and $\chi_1^2$ is a chi-square random variable with one degree of freedom.  Moreover, we are free to replace $\| \epsilon \alphanull\|_{L_2}$ in \eqref{eq:H0_normality_L2} above by its empirical counterpart $\|\wh\epsilon \wh\alpha_n\|_{L_2(\PPn)}$ from Section~\ref{sec:uniform}, to conclude that
\begin{align}
\label{eq:H0_normality_Ln}
    \text{for all fixed}~t\in\calT_\delta, \frac{ \sqrt{n} }{t} \frac{ \wc\eta_t^\us(\wc g_n) }{\| \wh\epsilon \wh\alpha_n \|_{L_2(\PPn)} }\rightarrow_{d} \calZ ,~~\frac{\wh Z}{ \| \wh\epsilon \wh\alpha_n \|_{L_2(\PPn)} } \rightarrow_{d} |\calZ|,~~\frac{\wh\chi^2}{ \| \wh\epsilon \wh\alpha_n \|_{L_2(\PPn)}^2 } \rightarrow_{d} \chi_1^2 .
\end{align}
\end{theorem}
The proof of Theorem~\ref{thm:main_null_master} is deferred to Section~\ref{sec:Proof_Thm_thm:main_null_master}.  The extra condition~\ref{thm:main_null_master:con_1} in Theorem~\ref{thm:main_null_master} is natural due to the presence of the $\sqrt{n}$ scaling factor in the our tests and is mild.  Note that if the bandwidth $h$ is held as a constant, then condition~\ref{thm:main_null_master:con_1} simply reduces to $\deltaf=\Co(n^{-1/2})$, and $\nun=\Co(n^{-1/4})$ which is in turn implied by $\kappa>1/2$ for $\kappa$ in Assumption~\ref{ass:NN_scaling} on the hierarchical composition model.

\section{Simulation studies}\label{sec:simulations}
In this section, we conduct simulations to assess the size and power of our conditional screening tests across different scenarios.  Recall that we summarize the implementations of these tests in Algorithm~\ref{algorithmhighd}.  We consistently employ the Quartic/biweight kernel.

We choose $t=1$ for the fixed-$t$ test statistic, and simply set $\calT_\delta = \{-1.25,-0.5,0.5,1.25\}$ in the sup statistic and the square statistic in \eqref{sup} and \eqref{square} respectively.  Then, we implement the sup statistic in \eqref{sup} as
$\mbox{sup}_{t \in T_{\delta}}| {\wh\eta_t^\us(\wc g_n)} |/t$, and the square statistic as $\wh\chi^2_{\delta} = n \sum_{t\in \calT_\delta} \left\{ {\wh\eta_{t}^\us(\wc g_n)} \right\}^2 / t^2$ (we take $w(t)=1$ in \eqref{square}).

We set $d=200$ and $d=400$ when $n=256$ and $n=512$ respectively, and the full regression model as $Y=m_0^*(\bfF,X_1,X_3)+\epsilon^*$ (see $m_0^*$ defined below \eqref{hypo_alt}).  To illustrate the performance of our conditional screening test, we have designed $m_0^*$ so that a significant portion of its variation is accounted for by the factors.
We further consider a nonlinear and a linear model of $m_0^*$.  Specifically, under the null hypothesis, we set:
\begin{equation}
\begin{gathered}
 \mbox{nonlinear}: m_0^*(\cdot) = m_0^{*(null)}(\cdot) = \mbox{sin}(f_1+u_1) + \log(8+f_2) \times \log(8+f_3) + \exp(-f_4^2/2), \\
 \mbox{linear}: m_0^*(\cdot) = m_0^{*(null)}(\cdot) = f_1-f_2+f_3+f_4-f_5 .
\end{gathered}
\label{eq:sim_null}
\end{equation}
For our screening test, we always select $X_3$ as the variable of interest.  Accordingly, under the alternative hypothesis we add signals in $X_3$ to $m_0^{*(null)}$ above:
\begin{equation}
\begin{gathered}
\mbox{nonlinear}:  m_0^*(\cdot) = m_0^{*(null)}(\cdot) + X_3^2 /4, \quad
\mbox{linear}: m_0^*(\cdot) = m_0^{*(null)}(\cdot)+ X_3/16.
\end{gathered}
\label{eq:sim_alt}
\end{equation}
The signal under the linear model is very weak and allows us to discern different test settings.
 While the signal under the nonlinear model may not appear weak at first, here the \textit{average} derivative $\EE m_{0,j}=0$, so detecting departure from the null critically depends on the higher-order, nonlinear effect of our MGF test statistics.  Under the nonlinear model $r=4$ while under the linear model $r=5$.  In both models we set $\rbar=r$.

Under the nonlinear model, $u_1$ is incorporated for a richer structure that specifically leads to the working regression function $m_0(\bfF,X_3)$ being different from the full regression function $m_0^*(\bfF,X_1,X_3)$ by design; see Remark~\ref{rmk:misspecification}.  We incorporate $u_1 = X_1 - \bB_{1\cdot}\bfF$ instead of $X_1$ directly because $u_1$ will be enforced to be independent of $X_3$ under screening.  Moreover, under the null hypothesis, the selection of $X_3$ is purely for clarity: under the nonlinear model, we could conduct our test on any $X_j$ for any $j\in\{2,\dots,d\}$, and under the linear model, for any $j\in\{1,\dots,d\}$.

For the high-dimensional predictor $\bX$ in the factor model~\eqref{eq:factor_observation}, the factor loading matrix $\bB$ is generated with i.i.d.\,Unif$[-\sqrt{3},\sqrt{3}]$ entries; the factor $\bfF$ and the idiosyncratic terms $\bmu$ both have i.i.d.~$\calN(0,0.6)$ entries.  We further draw the noise as $\epsilon^*\sim\calN(0,0.3)$.  The quantities $\bB$, $\bfF$, $\bmu$ and $\epsilon^*$ are all drawn independently.  We pre-train $\bW$ with samples of size $100$.

We set the hyper-parameters for the neural network fitting as follows: for the regression estimator $\wh g_n$ in \eqref{eq:def_wh_m_n_highd}, we employ neural networks with $L=5$ hidden layers, a common width of $k_{\textup{0}}=16$ per layer, and the ReLU activation function.  Training proceeds over $800$ epochs, utilizing a batch size of $256$ and a constant learning rate of $0.005$.  For the score function estimator $\wh\alpha_n$ in \eqref{eq:alpha_hat}, we maintain the same neural network fitting parameters, except that we set $L=2$, batch size to $64$ and the number of epochs to $400$.  We employ early stopping as the only regularization technique, and terminate training if there's no improvement after 20 epochs ($\text{patience}=20$) on a validation set.

We conduct our tests at significance levels of either $5\%$ or $10\%$ under both the nonlinear and the linear models, resulting in a total of four combinations summarized in Tables~\ref{tab:farlog1} to \ref{tab:farlinear2}. In each table we present in alternating rows the sizes and powers (the latters in parentheses) of the tests in \eqref{test:fixed_t}, \eqref{sup} and \eqref{square}, and of the same tests but without centering (that is, tests employing the non-adjusted estimator $\wh g_n$), under different sample sizes. The columns, arranged from left to right, correspond to the four different bandwidths.  Each entry is calculated based on $500$ Monte Carlo repetitions.

\begin{table}[htbp]
    \centering
    \begin{tabular}{c| c|c|c|c|c|c}
        \hline\hline
    \multirow{9}{*}{$n=256$}&   & $h$ & 0.1 & 1.0 & 1.5 & 2.0 \\ \hline
   & \multirow{2}{*}{fixed-$t$ test} & Non-centered & 0.16 (0.67) & 0.14 (0.57) & 0.13 (0.56) & 0.12 (0.51) \\
	&	    &	Centered & 0.11 (0.69) & 0.10 (0.60) & 0.08 (0.55) & 0.07 (0.50) \\ \cline{2-7}
   &  \multirow{2}{*}{sup test}		    &	Non-centered & 0.26 (0.97) & 0.21 (0.93) & 0.19 (0.92) & 0.18 (0.87) \\
   &	    &Centered & 0.19 (0.97) & 0.16 (0.93) & 0.13 (0.91) & 0.13 (0.85) \\ \cline{2-7}
   & \multirow{2}{*}{square test}		&    Non-centered & 0.16 (0.86) & 0.14 (0.79) & 0.14 (0.74) & 0.12 (0.65) \\
	&	    & Centered & 0.11 (0.82) & 0.09 (0.78) & 0.09 (0.70) & 0.09 (0.63) \\
		    \hline
    \end{tabular}
 \begin{tabular}{c|c|c|c|c|c|c}
 \multirow{7}{*}{$n=512$}  &		   \multirow{2}{*}{fixed-$t$ test} &  Non-centered & 0.11 (0.88) & 0.10 (0.80) & 0.10 (0.77) & 0.08 (0.70) \\
		&	& Centered & 0.09 (0.88) & 0.07 (0.83) & 0.07 (0.76) & 0.05 (0.70) \\ \cline{2-7}
	 &  \multirow{2}{*}{sup test}		& Non-centered & 0.19 (1.00) & 0.16 (0.98) & 0.15 (0.97) & 0.11 (0.96) \\
	 &		& Centered & 0.15 (1.00) & 0.12 (0.97) & 0.11 (0.97) & 0.08 (0.95) \\  \cline{2-7}
     &  \multirow{2}{*}{square test}	   &	Non-centered & 0.12 (0.96) & 0.10 (0.91) & 0.09 (0.87) & 0.08 (0.85) \\
 	 &							 &  Centered & 0.08 (0.96) & 0.07 (0.91) & 0.05 (0.87) & 0.05 (0.86) \\ \hline \hline
    \end{tabular}
    \caption{Performance summary of our test statistics under the nonlinear model in \eqref{eq:sim_null} (for size under the null) and \eqref{eq:sim_alt} (for power under the alternative) at the significance level $5\%$. Specifically, we provide the size and power (the latter displayed in parentheses) of our tests under various combinations of the test statistic (fixed-$t$, sup, or squared statistic), sample size, bandwidth, and the use of either the non-centered (employing the non-adjusted estimator $\widehat g_n$) or the centered (employing the adjusted $\widecheck g_n$) test statistics. Each value in the table represents the average over $500$ Monte Carlo repetitions.}
      \label{tab:farlog1}
\end{table}

\begin{table}[htbp]
    \centering
  \begin{tabular}{c|c|c|c|c|c|c}\hline\hline
  	\multirow{9}{*}{$n=512$}&   & $h$ & 0.1 & 1.0 & 1.5 & 2.0 \\ \hline
    & \multirow{2}{*}{fixed-$t$ test} & 	Non-centered & 0.16 (0.84) & 0.14 (0.78) & 0.14 (0.76) & 0.14 (0.76) \\
			& & Centered & 0.08 (0.87) & 0.09 (0.82) & 0.08 (0.80) & 0.07 (0.79) \\
        \cline{2-7}
	&  \multirow{2}{*}{sup test} &	Non-centered & 0.25 (0.85) & 0.21 (0.79) & 0.18 (0.77) & 0.18 (0.76) \\
		&  &	Centered & 0.18 (0.88) & 0.13 (0.83) & 0.11 (0.80) & 0.11 (0.80) \\
	      \cline{2-7}
	& \multirow{2}{*}{square test} &	Non-centered & 0.14 (0.74) & 0.14 (0.72) & 0.15 (0.72) & 0.14 (0.72) \\
		&  &	Centered & 0.09 (0.77) & 0.09 (0.75) & 0.08 (0.73) & 0.07 (0.75) \\
		\hline\hline
    \end{tabular}
    \caption{   Performance summary of our test statistics under the linear model in \eqref{eq:sim_null} (for size under the null) and \eqref{eq:sim_alt} (for power under the alternative) at the significance level $5\%$.}
    \label{tab:farlinear1}
\end{table}

 \begin{table}[htbp]
 	\centering
 	\begin{tabular}{c| c|c|c|c|c|c}
 		\hline\hline
 		\multirow{9}{*}{$n=512$}&   & $h$ & 0.1 & 1.0 & 1.5 & 2.0 \\ \hline
		 & \multirow{2}{*}{fixed-$t$ test}  & Non-centered & 0.17 (0.90) & 0.18 (0.84) & 0.18 (0.81) & 0.14 (0.76) \\
		&  & 	Centered & 0.13 (0.91) & 0.13 (0.86) & 0.13 (0.81) & 0.13 (0.77) \\ \cline{2-7}
        & \multirow{2}{*}{sup test}  & Non-centered & 0.27 (1.00) & 0.24 (0.99) & 0.24 (0.98) & 0.20 (0.97) \\
		& & 	Centered & 0.24 (1.00) & 0.19 (0.99) & 0.19 (0.98) & 0.17 (0.98) \\
		\cline{2-7}
		&  \multirow{2}{*}{square test}  & Non-centered & 0.19 (0.99) & 0.18 (0.96) & 0.16 (0.93) & 0.13 (0.90) \\
		&&	Centered & 0.13 (0.98) & 0.12 (0.95) & 0.12 (0.92) & 0.11 (0.90) \\
		\hline\hline
    \end{tabular}
    \caption{ The same caption as Table~\ref{tab:farlog1} except for significant level 10\%.}
     \label{tab:farlog2}
\end{table}

\begin{table}[htbp]
   \centering
   \begin{tabular}{c|c|c|c|c|c|c}\hline\hline
  	\multirow{9}{*}{$n=512$}&   & $h$ & 0.1 & 1.0 & 1.5 & 2.0 \\ \hline
  	& \multirow{2}{*}{fixed-$t$ test} &
    	Non-centered & 0.23 (0.89) & 0.21 (0.85) & 0.20 (0.82) & 0.21 (0.82) \\
    & & 	Centered & 0.14 (0.93) & 0.16 (0.88) & 0.13 (0.87) & 0.12 (0.85) \\ \cline{2-7}
	& \multirow{2}{*}{sup test} &  	Non-centered & 0.34 (0.90) & 0.26 (0.86) & 0.26 (0.82) & 0.26 (0.83) \\
			&& Centered & 0.26 (0.93) & 0.21 (0.89) & 0.17 (0.88) & 0.16 (0.86) \\\cline{2-7}
	& \multirow{2}{*}{square test}  &	Non-centered & 0.21 (0.83) & 0.21 (0.80) & 0.20 (0.78) & 0.21 (0.79) \\
	& &		Centered & 0.15 (0.86) & 0.15 (0.84) & 0.12 (0.83) & 0.13 (0.82) \\
		\hline\hline
    \end{tabular}
    \caption{ The same caption as Table~\ref{tab:farlinear1} except for significant level 10\%.}
    \label{tab:farlinear2}
\end{table}

We make a few observations from the results in the tables.  First, centering generally improves the size of our tests without affecting their power. Moreover, under the null, sizes arrive at the nominal level if we further increase the bandwidth; this is especially evident for the sup test, which naturally tends to reject more often (under both the null and alternative). This shows the improvement of the power of the centered, smoothed test compared to the initial test \eqref{eq:eta_t_naive} if both tests are calibrated to have the same size.  Of course, if we increase $h$ more than $1.5$, we expect the biases will arrive {under the alternative}.  In the nonlinear case (Tables~\ref{tab:farlog1} and  \ref{tab:farlog2}), the ensemble tests by aggregating various $t$ also outperform the fixed-$t$ test in power.

\section{Empirical Applications}
\label{sec:emp_app}

\subsection{Asset Pricing}
\label{sec:asset_pricing}

In our first empirical analysis, we examine a comprehensive dataset containing both returns and specific characteristics of firms.  Twelve monthly returns in $2021$ for the 100 largest financial institutions are drawn from the CRSP database as our response variable, along with a set of $d=48$ firm-specific characteristics as predictors; the resulting sample size is {$n=1200$}. The dataset's foundation is credited to \cite{chen2021open}.

We select the smoothing bandwidth from the sequence $\{0.4, 0.6, \dots, 1.8, 2.0\}$ using the bandwidth selection algorithm suggested in Remark~\ref{rmk:bandwidth}.  We set $\rbar=5$ by the conventional practice in finance, the hidden size $k_{\textup{0}}=24$, and the learning rate $\tau=10^{-3}$; we keep the other neural network fitting parameters identical to Section~\ref{sec:simulations}.

In Figure~\ref{assetpricing} we plot for each of the $44$ continuous predictors the intervals whose right ends are the test statistics normalized by their corresponding critical values in \eqref{test:fixed_t}, \eqref{sup} and \eqref{square}.  Specifically, we plot the intervals
$$ [ 0, n^{1/2} |\wc\eta_t^\us(\wc g_n)| / ( t z_{1-\alpha/2} \|\wh\epsilon \wh\alpha_n\|_{L_2(\PPn)} ) ], \quad [ 0, \wh Z/(z_{1-\alpha/2}\|\wh\epsilon \wh\alpha_n\|_{L_2(\PPn)})], \quad [ 0, \sqrt{\wh\chi^2}/(\sqrt{\chi^2_{1,1-\alpha}}\|\wh\epsilon \wh\alpha_n\|_{L_2(\PPn)})],$$
for the fixed-$t$, sup, and square tests respectively.  We also plot the same intervals but for the non-centered versions of these tests.  In all cases, an interval covering one indicates statistical significance.  The results suggest that the presence of some idiosyncratic terms related to tail risk, such as ReturnSkewCAPM and DownsideBeta, exert additional influence on returns and complement the explanatory power of the five factors.

\begin{figure}
    \includegraphics[height=20cm, width = 1.0\textwidth]{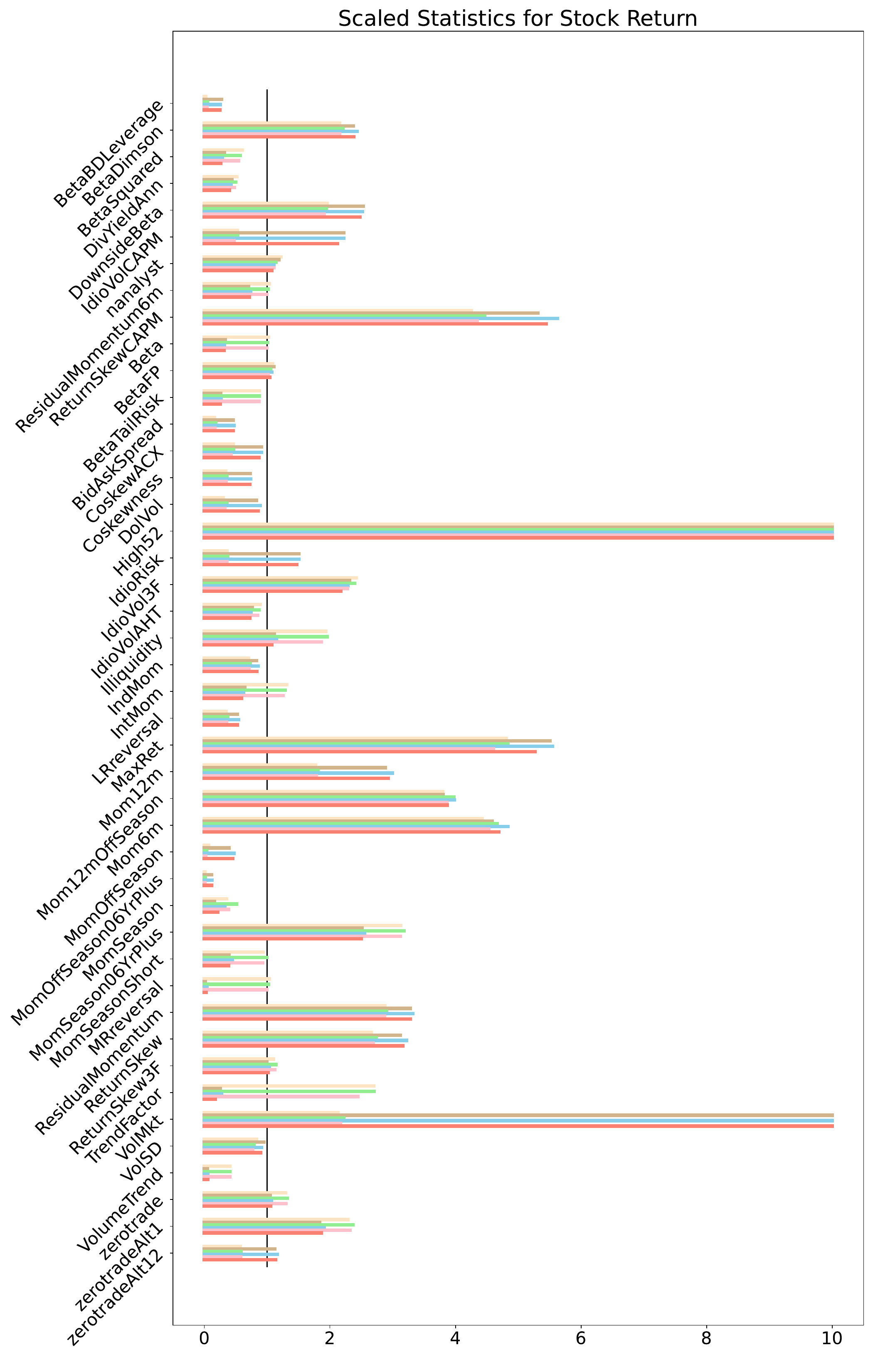}
\caption{Asset Pricing Dataset: Significance of {predictors/idiosyncratic terms} at the $5\%$ significance level. Intervals covering one (indicated by the vertical black line) correspond to the significant {predictors/idiosyncratic terms}. For each variable, six intervals are plotted in the order of: fix-$t$ test, sup test and square test, and within each test the {\color{black}non-centered version is plotted before the centered version}.\color{black}}
\label{assetpricing}
\end{figure}

\subsection{Macroeconomics Time Series}
\label{sec:applications:FRED-MD}

In this section, to substantiate Section~\ref{sec:asset_pricing}, we illustrate our conditional screening test with another empirical application, this time on the macroeconomics dataset FRED-MD introduced in \cite{mccracken2016fred} and also later studied in \cite{FanGu2023factor}.  The dataset collects $d=127$ monthly U.S.\,macroeconomic variables, such as unemployment rate and real personal income, starting from 1959/01.  It is shown in \cite{mccracken2016fred} that these variables can be explained well by several latent factors.

Our analysis setup is in general similar to Section~C in \cite{FanGu2023factor}. Our target variables are UEMP15T26,  TB3SMFFM or TB6SMFFM, and we aim to identify which variables contribute to predicting the target variables beyond the latent factors.  The variable UEMP15T26 represents the civilians unemployed for $15-26$ weeks. The variable TB3SMFFM (TB6SMFFM) measures the 3-month (6-month) treasury bill rate minus the effective federal funds rate. For each target response variable $y_{t+1}$ in \{UEMP15T26,  TB3SMFFM, TB6SMFFM\}, we regress $y_{t+1}$ on $\bx_t\in\RR^{127}$ where $\bx_t$ is the vector of all variables at the previous month. We choose the $n=330$ valid sample pairs $\{\bx_t, y_{t+1}\}$ between January 1980 and July 2022. We employ the same neural network fitting parameters as Section~\ref{sec:asset_pricing} except that we revert back to a hidden size $k_0=16$ as in our simulation studies.

For each variable, we plot the same six intervals derived from our tests as in Section~\ref{sec:asset_pricing} and we recall that an interval covering one indicates statistical significance.
Our analysis, based on the variables categorized by \cite{mccracken2016fred}, underscores the predominant impact of a limited set of variables on the three response variables.   As there are 127 variables, it takes four sub-figures to display a figure that is similar Figure~\ref{assetpricing} for each given response variable.
Figures~\ref{FRED10} {to} \ref{FRED13} in the appendix display the intervals, with the response variable being UEMP15T26 (civilian unemployment).
Notably, most variables fail to reach the $5\%$ significance level, except for some variables such as RETAILx, NDMANEMP, and COMPAPFFx. Additional significant variables for the response variables TB3SMFFM and TB6SMFFM are illustrated from Figures~\ref{FRED1} to \ref{FRED8}. As for TB3SMFFM, there are many significant variables. For example, the variables such as IPDMAT in the output and income variable group and DTCTHFNM in the prices group contribute additionally to the factors. TB6SMFFM demonstrates slightly higher susceptibility to variables compared to TB3SMFFM, which, in turn, is notably influenced by variables related to the stock market, employment, interest rates, and money and credit. Our overall findings validate our regression model based on a latent factors plus sparse idiosyncratic terms structure.

\section{Conclusion and further work}
\label{sec:conclusion}

We have introduced a conditional variable screening test for non-parametric regression using deep neural networks; the inputs to the networks are obtained with the help of a factor model that further enables us to handle very high-dimensional predictors.  In our test statistics, we employ high-quality estimators of the partial derivatives of the non-parametric regression function, which could be of independent interest.  To demonstrate the versatility of our test, we apply it to assess the adequacy of non-parametric factor regression.  An intriguing avenue for further exploration involves extending this framework to dependent data, other statistical machine learning losses, and simultaneous testing for multiple variables.  Relevant examples for the latter direction include the $\ell_{\infty}$ statistics proposed by \cite{chen2022inference} and the $\ell_2-\ell_{\infty}$ type statistics discussed by \cite{li2022ell}.

This paper focuses on the popular feed-forward networks with the ReLU activation function.  The derivative irregularity is not limited to the ReLU activation function, and hence, addressing this issue can benefit various neural network classes.  Moreover, in our approach, we apply the smoothing procedure after we have obtained a neural network estimator.  This is different from other derivative regularizations, for instance, Sobolev training \citep{Czarneck2017sobolev}, that employ roughness penalties during optimization, which could form a potential future topic.  Last but not least, we could also potentially conduct different tests for testing higher- or mixed-derivatives.  For instance, testing the monotonicity in a variable in the regression is equivalent to testing the sign of the associated partial derivative, a task already studied in the context of an one-dimensional predictor by, for instance, \cite{BowmanJonesGijbels1998testing}, \cite{GhosalArusharka2000testing} and \cite{HallHuang2001nonparametric}.

{\footnotesize \spacingset{1.0}
\bibliographystyle{apalike}
\bibliography{biball}
}

\appendix
\newpage

\addtocontents{toc}{\protect\thispagestyle{empty}}
\thispagestyle{empty}

\newpage

\pagestyle{plain}

The supplement is organized as follows.  Section~\ref{app_sec:misc} collects definitions, additional assumptions, and some miscellaneous results and discussions. Section~\ref{app_sec:proof_sec_3} provides supporting details and preliminary proofs for the main theorems in Section~\ref{sec:regression_function}, while Section~\ref{app_sec:proof_sec_4} presents the proofs for the said main theorems and the power for the fixed-$t$ test. Section~\ref{sec:lowd_supp} presents the low-dimensional regime as a special case of the high-dimensional regime. Section~\ref{sec:app_results_more} provides additional graphs for the empirical application in Section~\ref{sec:applications:FRED-MD}.  We introduce the additional shorthand notation that $\left\langle \cdot, \cdot \right\rangle$ represents the inner product with respect to the distribution $\PP$ (that is, $\left\langle g_1, g_2 \right\rangle = \int g_1 g_2 \rmd\PP$).

\section{Definitions, additional assumptions, miscellaneous results and discussions}
\label{app_sec:misc}

\subsection{Definitions}
\label{sec:def_supp}

\begin{definition}[H\"{o}lder class]
\label{def:Holder}
\normalfont Let $\beta$ and $C_\calH$ be two positive real numbers and let $\lfloor\beta\rfloor$ be the largest integer \textit{strictly} less than $\beta$.  The $(\beta, p, C_\calH)$-smoothness H\"{o}lder class on $\RR^p$ is the collection of functions $f:\RR^p\rightarrow\RR$ whose mixed derivatives $(\partial^{\lfloor\beta\rfloor}/\partial z_1^{\omega_1} \dots z_p^{\omega_p}) f(\bz)$, where $\bz=(z_1,\dots,z_p)^\top$, $\omega_1,\dots,\omega_p\ge 0$ and $\omega_1+\dots+\omega_p=\lfloor\beta\rfloor$, all exist, and moreover such derivatives are H\"{o}lder-continuous with constant $C_\calH$ and exponent $\beta-\lfloor\beta\rfloor$, or more precisely, $\left| (\partial^{\lfloor\beta\rfloor}/\partial z_1^{'\omega_1} \dots z_p^{'\omega_p}) f(\bz') -  (\partial^{\lfloor\beta\rfloor}/\partial z_1^{\omega_1} \dots z_p^{\omega_p}) f(\bz) \right| \le C_\calH \|\bz'-\bz\|^{\beta-\lfloor\beta\rfloor}$, $\forall \bz, \bz'$.
\end{definition}

\begin{definition}[sub-Gaussian norm; see, for instance, Definition~1 in \cite{Adamczak2008}]
\label{def:subGaussian}
Let the function $\psi_2:[0,\infty)\rightarrow[0,\infty)$ be $\psi_2(x) = \exp(x^2)-1$.  Then, the sub-Gaussian norm $\|\cdot\|_{\psi_2}$ of a random variable $X$ is $\|X\|_{\psi_2} = \inf\{\lambda>0: \EE \psi_2(|X|/\lambda)\le 1\}$.
\end{definition}

\begin{definition}[VC-subgraph class; see, for instance, Section~2.6.2 in \cite{vVW1996}]
\label{def:VC_subgraph}
The subgraph of a function $f:\calX\rightarrow\RR$ is the subset of $\calX\times\RR$ given by $\{(x,t): x\in\calX, t<f(x)\}$.  A collection $\calF$ of measurable functions is called a VC-subgraph class, or simply a VC-class, if the collection $\calS$ of the subgraphs corresponding to all functions $f\in\calF$ form a VC class of sets in $\calX\times\RR$.  Then, we define $V_{\calF}$ for the function class $\calF$ as the VC-index of the associated collection $\calS$ of subgraphs.  We refer the readers further to, for instance, Section~2.6.1 in \cite{vVW1996} for the definitions of VC classes of sets and the accompanying VC-index.
\end{definition}

\begin{definition}[VC-type of functions; see, for instance, the beginning of Section~2.1 in \cite{GineMason2007AoS}]
\label{def:VC}
Let $(S,\calS)$ be a measurable space.  A class $\calG$ of measurable functions on $(S,\calS)$ is \textit{VC-type} with respect to an envelope $G$, if for some $A\ge 3$ and index $v\ge 1$, the covering number $N(\calG,L_2(Q),\tau)$ satisfies
\begin{align}
\label{eq:covering_number_bound_master}
N( \calG, L_2(Q), \tau  ) \le \left( \dfrac{A \, \|G\|_{L_2(Q)}}{\tau} \right)^{v}, \quad 0<\tau\le \|G\|_{L_2(Q)},
\end{align}
for every probability measure $Q$ on $\calS$ for which $G\in L_2(Q)$.
\end{definition}

Finally we recap a definition coming out from recent deep neural network approximation theory.  It is well known that the optimal minimax rate for estimating functions in the $(\beta, \dbar, C_\calH)$-H\"{o}lder class (Def.~\ref{def:Holder}) scales with the sample size $n$ as $n^{-{\beta}/(2\beta+\dbar)}$ \citep{Stone1982}, which is slow when the dimension $\dbar$ is large compared to $\beta$.  To mitigate this issue, \cite{schmidt2020nonparametric} and \cite{KohlerLanger2021} showed the excellent approximation power of the neural network class $\calFn(\dbar)$ in \eqref{eq:base_NN_class} to the \textit{hierarchical composition model} $\calH(\dbar, l, \calP, C_\calH)$ defined in Def.~\ref{def:Hier_comp}.  This model consists of functions that are compositions in $l\in\NN$ layers of individual functions in various $(\beta, p, C_\calH)$-H\"{o}lder classes, where the $(\beta, p)$ tuples belong to a set $\calP \subset [1, \infty) \times \NN$ (for simplicity we assume a common $C_\calH$ across the composition functions).  Then, under the scaling suggested by Proposition~3.4 in \cite{fan2022noise} and recorded in Assumption~\ref{ass:NN_scaling}, the neural network class $\calF(\dbar)$ can approximate the class $\calH(\dbar, l, \calP, C_\calH)$ adaptively at a rate $(L k_{\textup{0}})^{-2\kappa}$ up to $\log$ factors, where $\kappa$ is the worst-case dimension-adjusted smoothness measure $\kappa\equiv\min_{(\beta, p)\in\calP} \beta/p$.  This rate depends only on the \textit{individual} composition functions, and not on the ambient dimension $\dbar$.

\begin{definition}
\label{def:Hier_comp}
The hierarchical composition model $\calH(\dbar, l, \calP, C_\calH)$ of functions from $\RR^\dbar$ to $\RR$, where $\dbar, l \in \NN$ and $\calP\subset[1, \infty) \times {\NN}$, is defined as follows. For $l=1$,
\begin{align*}
    & \calH(\dbar, 1, \calP, C_\calH)=\big\{ h: \RR^\dbar \rightarrow \RR: h(\bx)=g\left(x_{\pi(1)}, \ldots, x_{\pi(p)}\right) \text {, where }~g: \RR^p \rightarrow \RR\\
    & \quad \text{is $(\beta, p, C_\calH)$-H\"{o}lder for some $(\beta, p) \in\calP$ and $C_\calH<\infty$}, ~\text{and}~\pi: \{1,\dots,p\}\rightarrow\{1,\dots,\dbar\} \big\} .
\end{align*}
For $l>1$, $\calH(\dbar, l, \calP, C_\calH)$ is defined recursively as
\begin{align*}
&\calH(\dbar, l, \calP, C_\calH) = \big\{ h: \RR^\dbar \rightarrow \RR: h(\bx)=g\left(f_1(\bx), \ldots, f_p(\bx)\right),~\text{where}~g: \RR^p \rightarrow \RR \\
& \quad \text{is $(\beta, p, C_\calH)$-H\"{o}lder for some $(\beta, p) \in\calP$ and $C_\calH<\infty$},~\text{and each}~{f_k} \in \calH(\dbar, l-1, \calP, C_\calH) \big\}.
\end{align*}
\end{definition}

\subsection{Additional technical assumptions}
\label{sec:additional_ass_sec_3}

\begin{assumption}[Additional smoothness of $m_0$ and conditions on noise]\label{assum_fun_class}
	The true regression function $m_0(\bff,x_j)$ is Lipschitz in the first argument, precisely in the sense that $|m_0(\bff,x_j)-m_0(\bff',x_j)| \le C_{\textup{L}} {\|\bff-\bff'\|_2}$ for a common positive constant $C_{\textup{L}}$ for all $x_j$, $\bff$ and $\bff'$; in addition, for a positive constant $M_\infty'$, the derivative of $m_0$ in the last argument, namely $m_{0,j}$, exists everywhere and satisfies $\sup_{(\bff,x_j)\in \RR^{r+1}} |m_{0,j}(\bff,x_j)| < M_\infty'$.  The noise $\epsilon= Y - m_0(\bfF,X_j)$ is sub-Gaussian (see Def.~\ref{def:subGaussian}) with sub-Gaussian norm $C_\epsilon<\infty$.
\end{assumption}

\begin{assumption}[Diversified projection matrix]
	\label{ass:W}
	The diversified projection matrix $\bW=(\textup{W}_{j,k})_{1\le j\le d,1\le k\le\rbar}\in\RR^{d\times\rbar}$ is exogenous, so is independent of $(Y_i, \bfF_i,\bX_i)$, $i=1,\dots,n$. For some constants $c_{\textup{W},1}$, $c_{\textup{W},2}$, and $c_{\textup{W},3}$, $\max_{j,k}\left|\textup{W}_{j,k}\right| \le c_{\textup{W},1}$, and the minimum and maximum singular values of $\bH$ (recall that $\bH=d^{-1}\bW^\top \bB$), denoted by $v_{\min}(\bH)$ and $v_{\max}(\bH)$ respectively, satisfy $0<c_{\textup{W},2} \le v_{\min}(\bH) \le v_{\max}(\bH) \le c_{\textup{W},3}$.
	Finally,
	\begin{align}
		\label{eq:cond_exp_stability}
		\EE[ \{ \EE[Y|\wt\bfF,X_j] - m_0(\bfF,X_j) \}^2 ] \lesssim \deltaf^2 .
	\end{align}
\end{assumption}
Assumption~\ref{ass:W} above imposes conditions on $\bW$ beyond exogeneity to ensure in particular $\bH^\dagger \wt\bfF \approx \bfF$.
Moreover, Eq.~\eqref{eq:cond_exp_stability} in Assumption~\ref{ass:W} is a high-level condition regarding the two different conditional expectations $\EE[Y|\wt\bfF,X_j]$ and (from \eqref{eq:reg_model}) $m_0(\bfF,X_j) = \EE[Y|\bfF,X_j]$.  This condition is necessary because while our regression model \eqref{eq:reg_model} involves $m_0$, our estimation procedures instead rely on the (sample of the) diversified factors $\wt\bfF$ and $X_j$, so $\EE[Y|\wt\bfF,X_j]$ becomes a more appropriate measure of the ``center'' of $Y$.  Intuitively, because $\bfF,X_j$ and $\wt\bfF,X_j$ should generate similar $\sigma$-algebras, $\EE[Y|\wt\bfF,X_j]$ should be close to $m_0(\bfF,X_j)$ as characterized by \eqref{eq:cond_exp_stability} \citep{WlodzimierzWlodzimierz1992}.  We will provide conditions in Lemma~\ref{lemma:cond_exp_stability} for the validity of \eqref{eq:cond_exp_stability}.

\begin{assumption}[Smoothing kernel]
	\label{ass:kernel}
	The kernel $K:\RR\rightarrow\RR$ is compactly supported on $[-1,1]$, satisfies $\int_{-1}^1 K(a) \rmd a=1$ and vanishes at the boundary points $-1$ and $1$, and is continuously differentiable with uniformly bounded derivative $\dot K$.  The bandwidth $h$ satisfies $h\le b$ (and thus the interior sets $\calB_h$, $\calI_h$ and $\Omega_h$ defined in Section~\ref{sec:test_stat} are not degenerate).
\end{assumption}

\begin{assumption}[Miscellaneous conditions for estimating $\alpha_{t,n}^*$]
	\label{ass:u_t_n_misc}
	The class of functions $\{\alpha_{t,n}^*:t\in\calT_\delta\}$ indexed by $t$ is of VC-type (see Definition~\ref{def:VC}) with an index of at most $V_{\calFn(\rbar+1)}$, and is uniformly bounded in magnitude by $M_\infty<\infty$.  Moreover, for each $k\in\{1,\dots,\rbar\}$, for $\bH_{k\cdot}$ the $k$-th row of $\bH$, $\PP(|\bH_{k\cdot} \bfF| > c_b b/2)\le\deltaf^2/\rbar$.
\end{assumption}

In Assumption~\ref{ass:u_t_n_misc}, $V_{\calFn(\rbar+1)}$ is the VC-index of the subgraphs associated with the functions in the neural network class $\calFn(\rbar+1)$, which is well-defined because $\calFn(\rbar+1)$ turns out to be a VC-\textit{subgraph} class of functions as in Definition~\ref{def:VC_subgraph}.  On the other hand, the VC-\textit{type} assumption on the class $\{\alpha_{t,n}^*:t\in\calT_\delta\}$ refers to Definition~\ref{def:VC}; the ``index'' therein is not explicitly defined in terms of subgraphs and Assumption~\ref{ass:u_t_n_misc} simply states that this index is the \textit{number} $V_{\calFn(\rbar+1)}$.  Technically, the VC-subgraph and the VC-type classes of functions are related: a VC-subgraph class of functions with a VC-index $v$ admits a similar covering number as a VC-type class of functions that admits the same index $v$.  However, conceptually, these two classes of functions are not equal.  Next, in Assumption~\ref{ass:u_t_n_misc}, $\|\bH_{k\cdot}\|_2$ is upper bounded by $v_{\max}(\bH)$ which is in turn bounded by a constant as stipulated in Assumption~\ref{ass:W}.  Thus, because moreover $\bfF$ is coordinate-wise bounded, the tail condition on $\bH_{k\cdot} \bfF$ in Assumption~\ref{ass:u_t_n_misc} is satisfied if $c_b$ is large enough.

\begin{assumption}[Truncation function]\label{ass:truncation}
	The truncated function $\Psi$ is twice differentiable, bounded, and have bounded first- and second-order derivatives $\dot\Psi$ and $\ddot\Psi$ respectively.
	In addition, $\Psi$ satisfies $\Psi(t)=t$ for $t\in[- \|K\|_{L_1} M_\infty', \|K\|_{L_1} M_\infty']$ (where $M_\infty'$ is defined in Assumption~\ref{assum_fun_class} and $\|K\|_{L_1}$ is the $L_1$ norm of the kernel $K$).
\end{assumption}
For an example of a truncation function $\Psi$ satisfying Assumption~\ref{ass:truncation}, we can simply ``insert'' a linear segment of length $2\gamma$, where $\gamma>\|K\|_{L_1} M_\infty'$, around the origin into a centered and scaled logistic function; specifically, this $\Psi$ is given by $\Psi(t)=t$ for $|t|\le \gamma$ and $\Psi(t) = \text{sgn}(t) \gamma - 1/2 + 1/[ 1 + \exp\{-4(t- \text{sgn}(t) \gamma)\}]$ for $|t|>\gamma$.

\subsection{Miscellaneous results}
\label{sec:results_supp}

\begin{lemma}
\label{lemma:derivative_bound_via_original_specific}
Under Assumption~\ref{ass:kernel}, for generic functions $g=g(\wt\bff,x_j):\RR^{\rbar+1}\rightarrow\RR$ and $m=m(\bff,x_j):\RR^{r+1}\rightarrow\RR$, there exists a constant $c^\us$ (that does not depend on $g$ or $m$) such that
\begin{equation}
\begin{gathered}
    \textstyle \int_{\Omega_h} g_j^\us(\wt\bff,x_j)^2\,\rmd\PP \le \frac{(c^\us)^2}{h^2} \int g(\wt\bff,x_j)^2 \rmd\PP, \\ \textstyle \int_{\Omega_h} \{ g_j^\us(\wt\bff,x_j) - m_j^\us(\bff,x_j) \}^2\,\rmd\PP \le \frac{(c^\us)^2}{h^2} \int \{ g(\wt\bff,x_j) - m(\bff,x_j) \}^2 \rmd\PP .
\end{gathered}
\label{eq:derivative_bound_via_original}
\end{equation}
\end{lemma}
\begin{proof}
We only prove the first half of \eqref{eq:derivative_bound_via_original}; the second half will follow by a similar argument.  We have, starting from \eqref{eq:smooth_master} and the Cauchy-Schwarz inequality,
\begin{align*}
    & \textstyle \int_{\Omega_h} g_j^\us(\wt\bff,x_j)^2\, \rmd\PP = \textstyle \int_{\Omega_h} \left\{ \frac{1}{h} \int_{-1}^1 g(\wt\bff,x_j-ah) \dot K(a) \rmd a \right\}^2\, \rmd\PP \\
    & \textstyle \le \frac{1}{h^2} \int_{\Omega_h} \left\{ \int_{-1}^1 g^2(\wt\bff,x_j-ah) \rmd a \right\} \left\{ \int_{-1}^1 \dot K(a)^2 \rmd a \right\}\, \rmd\PP \lesssim \frac{1}{h^2} \int_{\Omega_h}  \left\{ \int_{-1}^1 g^2(\wt\bff,x_j-ah) \rmd a \right\}\, \rmd\PP \\
    & \textstyle = \frac{1}{h^2} \int_{-1}^{1} \left\{ \int_{\Omega_h} g^2(\wt\bff,x_j - a h)\, \rmd\PP \right\} \rmd a \le \frac{2}{h^2} \|g\|_{L_2}^2 . \qedhere
\end{align*}
\end{proof}

\begin{lemma}[Numerical approximation for candidate Riesz representer derivative in Section~\ref{sec:Riesz_est}]
\label{lemma:alpha_j_numerical}
Let $N$ be a large but fixed integer (we used $N=50$ throughout our numerical studies) and consider the grid points $a_1,\dots,a_{2N}$ where $a_l=-1-\frac{1}{2N}+\frac{l}{N}$ (these grid points are the centers of the $2N$ equally-sized intervals that form a partition of the interval $[-1,1]$).  Then, $\alpha_j^\us(\wt\bfF_i,X_{i,j})$ in Section~\ref{sec:Riesz_est} can be numerically approximated via a middle Riemann sum over the said grid points as
\begin{align*}
& \textstyle \alpha_j^\us(\wt\bfF_i,X_{i,j}) \approx \frac{1}{Nh} \sum_{l=1}^{2N} \dot{K}\left(a_l\right) \alpha\left(\wt\bfF_i,X_{i,j} - a_l h\right) . \nonumber
\end{align*}
Then, the sample average of the $\alpha_j^\us(\wt\bfF_i,X_{i,j})$'s in $\Rnullhat$ in \eqref{eq:Rhat} can be approximated as
\begin{align*}
& \textstyle \frac{1}{n} \sum_{i\in\calI_h} \alpha_j^\us(\wt\bfF_i,X_{i,j}) \approx \frac{1}{nNh} \sum_{l=1}^{2N} \dot{K}\left(a_l\right) \left\{\sum_{i\in\calI_h} \alpha\left(\wt\bfF_i,X_{i,j} - a_l h\right) \right\} . \nonumber
\end{align*}
\end{lemma}

\begin{proof}
We only prove the first equation display, because the second equation display is straightforward given the first.  Starting from the right-most term in \eqref{eq:smooth_master}, we have
\begin{align*}
\textstyle \alpha_j^\us(\wt\bfF_i,X_{i,j}) & \textstyle = \frac{1}{h} \int_{-1}^1 \alpha(\wt\bfF_i,X_{i,j} - a h) \dot K(a) \rmd a \approx \frac{1}{Nh} \sum_{l=1}^{2N} \dot{K}\left(a_l\right)\alpha\left(\wt\bfF_i,X_{i,j} - a_l h\right) ,
\end{align*}
where the approximation follows by replacing the integral with the middle Riemann sum over the grid points in the lemma statement, where each grid point carries a weight $2/(2N)=1/N$.
\end{proof}

\begin{lemma}[Existence and analytical form of the Riesz representer in \eqref{eq:Riesz_H0_master}]
\label{lem:Riesz_analytic}
Under Assumptions~\ref{ass:kernel} and \ref{ass:truncation}, the Riesz representer $\alpha_{t,n}^*$ that satisfies \eqref{eq:Riesz_H0_master} exists.  Furthermore, if $\wt\bfF, X_j$ admits a joint density $p=p(\wt\bff,x_j)$, then $\alpha_{t,n}^*$ takes the form
\begin{align}
\alpha_{t,n}^*(\wt\bff,x_j) = \dfrac{ \int_{\calB_h} \dot K_h(z-x_j) e^{t\,g_{0,j}^\us(\wt\bff,z) } p(\wt\bff,z) \rmd z}{ p(\wt\bff,x_j) } . \label{eq:def_v}
\end{align}
Finally, if $p$ is bounded and is differentiable in $x_j$ with the derivative being $p_j$, then under $H_0$ and in the limit $h\rightarrow 0$, \eqref{eq:def_v} simplifies to $-p_j(\wt\bff,x_j)/p(\wt\bff,x_j)$ .
\end{lemma}

\begin{proof}
We first verify that $\frac{1}{t} \frac{\partial \wt\eta_t^\us(g_0)}{\partial g}[v]$ is a bounded linear functional.  For any fixed $t$ and $h$,
\begin{align*}
& \textstyle \frac{\partial \wt\eta_t^\us(g_0)}{\partial g}[v] =\frac{\partial \wt\eta_t^\us(g_0+\tau v)}{\partial\tau}\Big|_{\tau=0} \textstyle = \left[ \frac{\partial}{\partial\tau} \int_{\Omega_h} e^{t\,\Psi\{(g_0+\tau v)_j^\us(\wt\bff,x_j)\}} \rmd\PP \right] \Big|_{\tau=0} \nonumber \\
& \textstyle = \int_{\Omega_h} \left[ e^{t\,\Psi\{(g_0+\tau v)_j^\us(\wt\bff,x_j)\}} \dot\Psi\{(g_0+\tau v)_j^\us(\wt\bff,x_j)\} \right] \Big|_{\tau=0}t v_j^\us(\wt\bff,x_j) \rmd\PP \nonumber \\
& \textstyle = t \int_{\Omega_h} e^{t\,\Psi\{ g_{0,j}^\us(\wt\bff,x_j)\} } \dot\Psi\{ g_{0,j}^\us(\wt\bff,x_j)\} v_j^\us(\wt\bff,x_j) \rmd\PP .
\end{align*}
Now, we note that equivalent to \eqref{eq:smooth_master} one can also write
\begin{align*}
	\textstyle g_{0,j}^\us(\wt\bff, x_j) = \int_{-1}^1 m_{0,j}(\bH^\dagger \wt\bff,x_j-ah) K(a) \rmd a .
\end{align*}
Thus, clearly
$|g_{0,j}^\us(\wt\bff, x_j)|\le \|K\|_{L_1} \|m_{0,j}\|_{L_\infty}$ whenever $x_j\in\calB_h$.  Analogously, $|m_{0,j}^\us(\bff, x_j)|\le \|K\|_{L_1} \|m_{0,j}\|_{L_\infty}$ whenever $x_j\in\calB_h$ as well.  Thus, by Assumptions~\ref{assum_fun_class} and \ref{ass:truncation}, $\Psi\{ g_{0,j}^\us(\wt\bff,x_j)\}=g_{0,j}^\us(\wt\bff,x_j)$ and $\dot\Psi\{ g_{0,j}^\us(\wt\bff,x_j)\}=1$; hence, we conclude that
\begin{align}
	& \textstyle \frac{\partial \wt\eta_t^\us(g_0)}{\partial g}[v] = t \int_{\Omega_h} e^{t\, g_{0,j}^\us(\wt\bff,x_j) } v_j^\us(\wt\bff,x_j) \rmd\PP .
	\label{eq:target_functional_derivative_master}
\end{align}
Thus, the functional $\frac{1}{t} \frac{\partial \wt\eta_t^\us(g_0)}{\partial g}[v]$ is clearly linear in its input and, by the boundedness of $g_{0,j}^\us$ just shown and Lemma~\ref{lemma:derivative_bound_via_original_specific}, is also bounded.
Therefore, the Riesz representer that satisfies \eqref{eq:Riesz_H0_master} exists \cite[Theorem~3.4, Chapter~1]{Conway1990}.  As a side remark, the modified functional that, when acting on any $v$, always returns $\int_{\Omega_h} v_j^\us(\wt\bff,x_j) \rmd\PP$ (the last term in \eqref{eq:Riesz_H0_master}) is also a bounded linear functional.  Hence, this modified functional admits its own Riesz representer.  It's also easy to show that the latter Riesz representer minimizes the population loss $\Rnull \equiv \int \Rnullhat(\alpha) \rmd\PP$.  Thus the latter Riesz representer is precisely $\alphanull$ which we have defined in Section~\ref{sec:variance_prep} as the population limit of $\wh\alpha_n$.

Next, starting from the first two steps of \eqref{eq:Riesz_H0_master}, with the assumed density $p$,
\begin{align*}
    & \textstyle t \langle v, \alpha_{t,n}^*\rangle = \frac{\partial \wt\eta_t^\us(g_0)}{\partial g}[v] = t \int_{[-b,b]^\rbar} \int_{\calB_h} \{ \int_{[-b,b]} v(\wt\bff,z) \dot K_h(x_j-z) \rmd z \} e^{t\,g_{0,j}^\us(\wt\bff,x_j)} p(\wt\bff,x_j) \rmd x_j \rmd\wt\bff \\
    & \textstyle = t \int_{[-b,b]^{\rbar+1}} v(\wt\bff,z) \{ \int_{\calB_h} \dot K_h(x_j-z) e^{t\,g_{0,j}^\us(\wt\bff,x_j)} p(\wt\bff,x_j) \rmd x_j \} \rmd z  \rmd \wt\bff \\
    & = t \int_{[-b,b]^{\rbar+1}} v(\wt\bff,z) \left\{ \dfrac{ \int_{\calB_h} \dot K_h(x_j-z) e^{t\,g_{0,j}^\us(\wt\bff,x_j)} p(\wt\bff,x_j) \rmd x_j}{ p(\wt\bff,z) } \right\} p (\wt\bff,z) \rmd z  \rmd \wt\bff .
\end{align*}
After interchanging the dummy integration variables $z$ and $x_j$ in the last line above, the term in the curly bracket then gives the form of $\alpha_{t,n}^*$ in \eqref{eq:def_v}.  Finally,
\begin{align*}
\alpha_{t,n}^*(\wt\bff,x_j) = \frac{\int_{-b+h}^{b-h}\dot K_h\left(z-x_j\right) e^{t\,g_{0,j}^\us(\wt\bff,z)} p(\wt\bff,z) \rmd z }{p(\wt\bff,x_j) }  {\stackrel{\text{under}~H_{0}}{=}} \frac{\int_{-b+h}^{b-h}\dot K_h\left(z-x_j\right) p(\wt\bff,z) \rmd z }{p(\wt\bff,x_j) } .
\end{align*}
For $x_j\in[-b+2h,b-2h]$, the numerator can be treated as
\begin{align*}
    & \textstyle \int_{-b+h}^{b-h}\dot K_h\left(z-x_j\right) p(\wt\bff,z) \rmd z = - \int_{x_j-h}^{x_j+h}\dot K_h\left(x_j-z\right) p(\wt\bff,z) \rmd z = - \frac{\partial}{\partial x_j}\int_{x_j-h}^{x_j+h} K_h\left(x_j-z\right) p(\wt\bff,z) \rmd z \\
    & \textstyle = - \frac{\partial}{\partial x_j}\int_{x_j-h}^{x_j+h} K_h\left(z\right) p(\wt\bff,x_j-z) \rmd z = - \int_{x_j-h}^{x_j+h} K_h\left(z\right) p_j(\wt\bff,x_j-z) \rmd z \xrightarrow{h\rightarrow 0} - p_j(\wt\bff,x_j) ,
\end{align*}
where we have applied twice differentiation under the integral sign \cite[Theorem 3, Chapter 8]{ProtterCharles2012}, followed by the usual approximation to the identity theory \cite[Theorem~2.1]{SteinShakarchi2009}.  This completes the proof of the Lemma.
\end{proof}

\subsection{Additional remarks}
\label{sec:remarks_supp}

\begin{remark}[Specification test in factor models]
\label{rmk:specification_test}
As discussed in Section~\ref{sec:reg_model}, our hypotheses \eqref{hypo_null} and \eqref{hypo_alt} are equivalent to determining whether the $j$-th idiosyncratic term $u_j$ of $\bmu$ (defined in \eqref{eq:factor_observation}) contributes to the regression function at the population level.  In the framework of \cite{fan2022learning}, this test is further similar to testing their Factor Augmented Regression (FAR) model against their Factor Augmented Sparse Throughput (FAST) model for high-dimensional regression. To establish the aforementioned equivalence, let's define the regression function
\begin{align*}
    m_0^\dagger(\bfF,u_j)=\EE[Y|\bfF,u_j]
\end{align*}
and the associated partial derivative with respect to the last argument:
\begin{equation}
\label{eq:FAST_m_j}
    \textstyle m_{0,j}^\dagger(\bff,v_j) = \frac{\partial}{\partial v_j} m_0^\dagger(\bff,v_j).
\end{equation}
Then, similar to our hypotheses \eqref{hypo_null} and \eqref{hypo_alt}, the new hypotheses for $u_j$ can be formulated as
\begin{align}
\label{null}
H_0^\dagger: m_{0,j}^\dagger(\bff,v_j)=0~\mbox{for all}~\bff,v_j, \quad \mbox{against} \quad H_A^\dagger: m_{0,j}^\dagger(\bff,v_j)\neq 0~\mbox{for some}~\bff,v_j .
\end{align}

Let $\bB_{j\cdot}$ denote the $j$-th row of the loading matrix $\bB$, so $X_j=\bB_{j\cdot}\bfF+u_j$.  Clearly, $\bfF,X_j$ and $\bfF,u_j$ generate the same $\sigma$-algebra, and hence the functions $m_0$ and $m_0^\dagger$ as conditional expectations are essentially a re-writing of each other with different arguments.  Specifically,
\begin{align*}
m_0(\bff, x_j) = m_0^\dagger(\bff, x_j-\bB_{j\cdot}\bff) = m_0^\dagger(\bff,v_j).
\end{align*}
Thus, the partial derivatives $m_{0,j}$ and $m_{0,j}^\dagger$ are linked as follows: at $x_j-\bB_{j\cdot}\bff=v_j$,
\begin{align*}
\textstyle m_{0,j}(\bff, x_j) = \frac{\partial}{\partial x_j} m_0(\bff, x_j) = \frac{\partial}{\partial x_j} m_0^\dagger(\bff, x_j-\bB_{j\cdot}\bff) = m_{0,j}^\dagger(\bff,v_j).
\end{align*}
Therefore, at least at the population level, our conditional screening hypotheses \eqref{hypo_null} and \eqref{hypo_alt} for $m_{0,j}$ are equivalent to the similar ones in \eqref{null} for $m_{0,j}^\dagger$.
\end{remark}

\begin{remark}[Relationship between \eqref{eq:reg_model} and the full model]
\label{rmk:misspecification}
Note that our reduced regression model \eqref{eq:reg_model} obviously does not preclude the ``full'' regression model $Y = m_0^*(\bfF,\bX) + \epsilon^*$ where $m_0^*(\bfF,\bX) = \EE[Y|\bfF,\bX]$ is the full conditional expectation.  It is apparent that our reduced model and the full model are linked as
\begin{align*}
    \epsilon = \epsilon^* + \EE[Y|\bfF,\bX] - \EE[Y|\bfF,X_j] = \epsilon^* + m_0^*(\bfF,\bX) - m_0(\bfF,X_j) .
\end{align*}
Moreover, by the definition of the conditional expectation, $\epsilon$ is centered and uncorrelated with any function of $\bfF$ and $X_j$.
\end{remark}

\begin{remark}[Boundary region]
\label{rmk:boundary}
Technically, we need to exclude the boundary region, which under Assumption~\ref{ass:DGP} is the union $[-b,-b+h]\cup[b-h,b]$, from our smoothing operation, because smoothing at a $x_j$ in this region will extend beyond the support of $X_j$ which is problematic.  Thus the sets $\calB_h$, $\calI_h$, $\Omega_h$ relevant to the interior regions in different contexts were introduced.  For instance, the quantity $\int_{\Omega_h} g \rmd\PP$ signifies the integral of the function $g$ in the interior region only.
\end{remark}

\begin{remark}[Cross validation for bandwidth selection]
\label{rmk:bandwidth}
In this remark we provide some brief guideline on how to choose an optimal bandwidth $h$ through cross validation.  Suppose that we split the data evenly into $K$ folds (for instance, $K=5$) indexed by $k=1, \dots, K$, and denote the collection of sample indices $i\in\{1,\dots,n\}$ within the $k$-th fold by $\calI_k$.  We let $\wh g_n^{(-k)}=\wh g_n^{(-k)}(\wt\bff,x_j)$ be the regression function estimator similar to that in \eqref{eq:def_wh_m_n_highd}, except that here we exclude the samples in the $k$-th fold; then, let $\wh g_{n,j,h}^{(-k),\us}=\wh g_{n,j,h}^{(-k),\us}(\wt\bff,x_j)$ be the corresponding estimator for the partial derivative (with respect to $x_j$) that is smoothed at bandwidth $h$.  Next, let $\wh g_n^{(k)}$ be the regression function estimator using only the sample in the $k$-th fold, and let $\wh g_{n,j}^{(k)}$ be the corresponding \textit{unsmoothed} estimator for the partial derivative (with respect to $x_j$).  Then, we define the cross validation error at $h$ as $\textup{CV}(h)=\sum_{k=1}^K \sum_{i\in \calI_k}[\wh g_{n,j,h}^{(-k),\us}(\wt\bfF_i, X_{i,j})-\wh g_{n,j}^{(k)}(\wt\bfF_i, X_{i,j})]^2$.  Finally, we could look for the value of $h$ that minimizes $\textup{CV}(h)$ as the optimal $h$.
\end{remark}

\subsection{Further motivation on the smoothing operation}
\label{sec:sawtooth}

In this section, we first provide detailed descriptions for our observation in Section~\ref{sec:test_stat} on the sawtooth function $\zeta_L$ which demonstrates the intricacy of derivative estimation with deep neural networks, and then further motivate our smoothing operation.  By Proposition~2 in \cite{Yarotsky2017}, the triangular function $\zeta:[0,1]\rightarrow[0,1]$ defined as
\begin{align*}
    \zeta(t) = \begin{cases}
        2t, &0\le t<1/2 , \\
        2(1-t), & 1/2\le t\le 1, \\
        0, & \text{otherwise}
    \end{cases}
\end{align*}
is given in terms of the ReLU function $\sigma$ by $\zeta(t)=2\sigma(t)-4\sigma(t-1/2)+2\sigma(t-1)$.
Next, let $\zeta_L(\cdot)=\frac{2}{2^L}\zeta_L\circ\zeta_L\circ\cdots\circ\zeta_L$ denote the $L$-fold iteration of $\zeta$ that is further scaled by an overall factor $\frac{2}{2^L}$.  The function $\zeta_L$, described in Lemma~2.4 in \cite{Telgarsky2015}, is often referred to as the ``sawtooth'' function.  (To be precise, in that reference, $\zeta_L$ is defined without the scaling factor $\frac{2}{2^L}$.)  The function $\zeta_L$ is visually represented in Figure~\ref{fig:sawtooth} that in particular illustrates the increasingly ``oscillatory'' behavior of $\zeta_L$ for increasing values of $L$.  By Proposition~2 in \cite{Yarotsky2017}, it's feasible to implement $\zeta_L$ with a neural network of no more than $L$ hidden layers and three nodes per layer.  (See Fig.~2(c) in \cite{Yarotsky2017} for a related construction.)

\renewcommand{\thefigure}{\Alph{section}.\arabic{figure}}

\begin{figure}[ht]
	\begin{subfigure}{1\textwidth}
		\centering
		\includegraphics[width=0.9\linewidth]{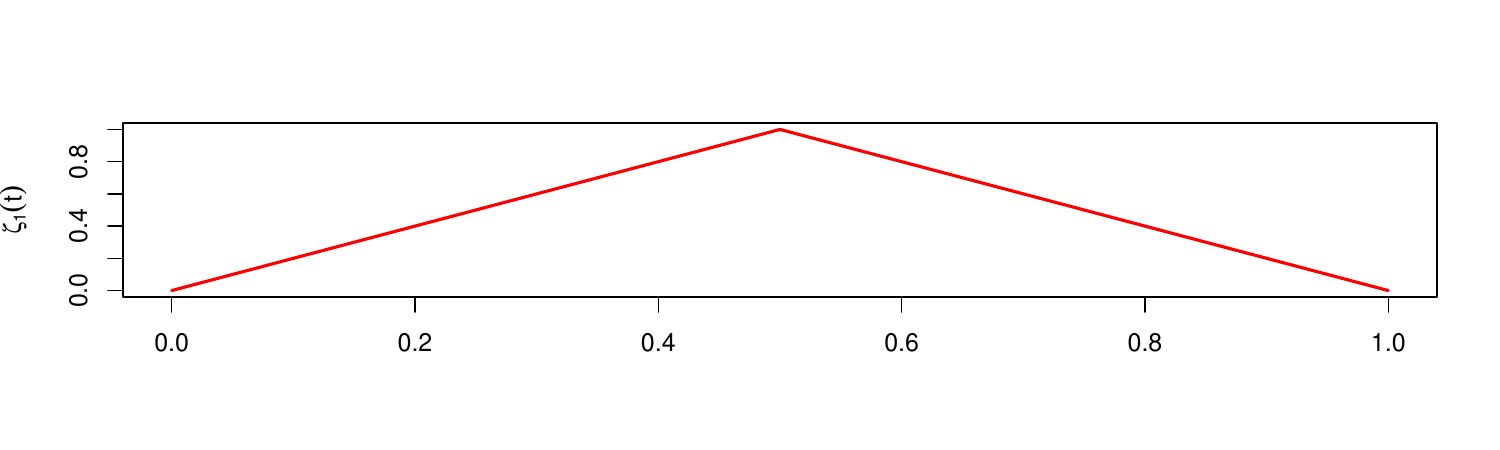}
		\vspace{-20mm}
	\end{subfigure}
	\begin{subfigure}{1\textwidth}
		\centering
		\includegraphics[width=0.9\linewidth]{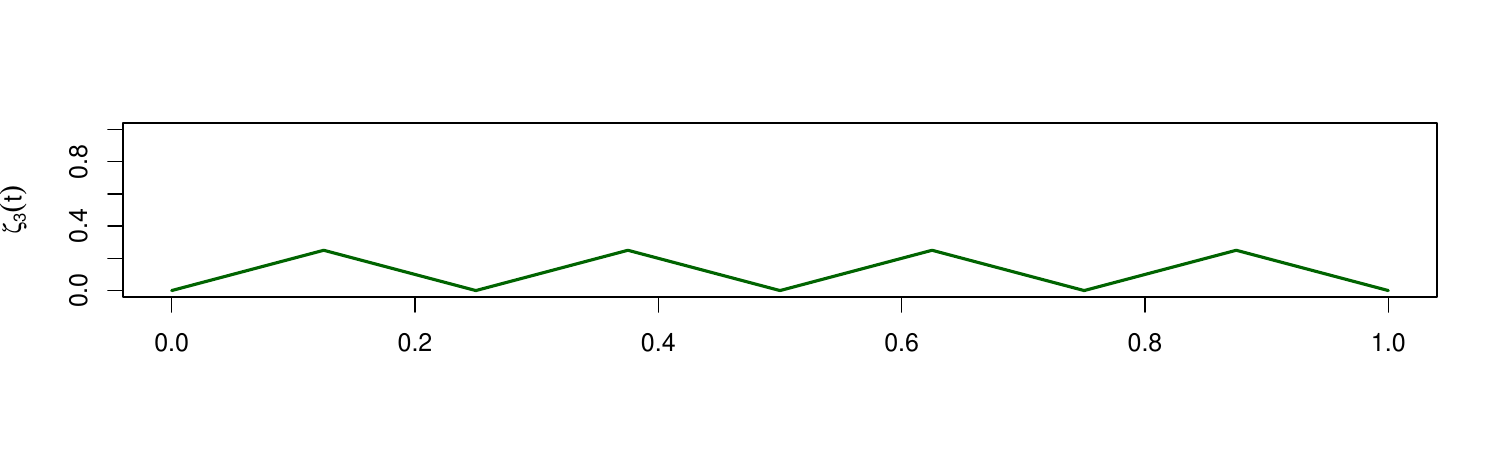}
		\vspace{-20mm}
	\end{subfigure}
	\begin{subfigure}{1\textwidth}
		\centering
		\includegraphics[width=0.9\linewidth]{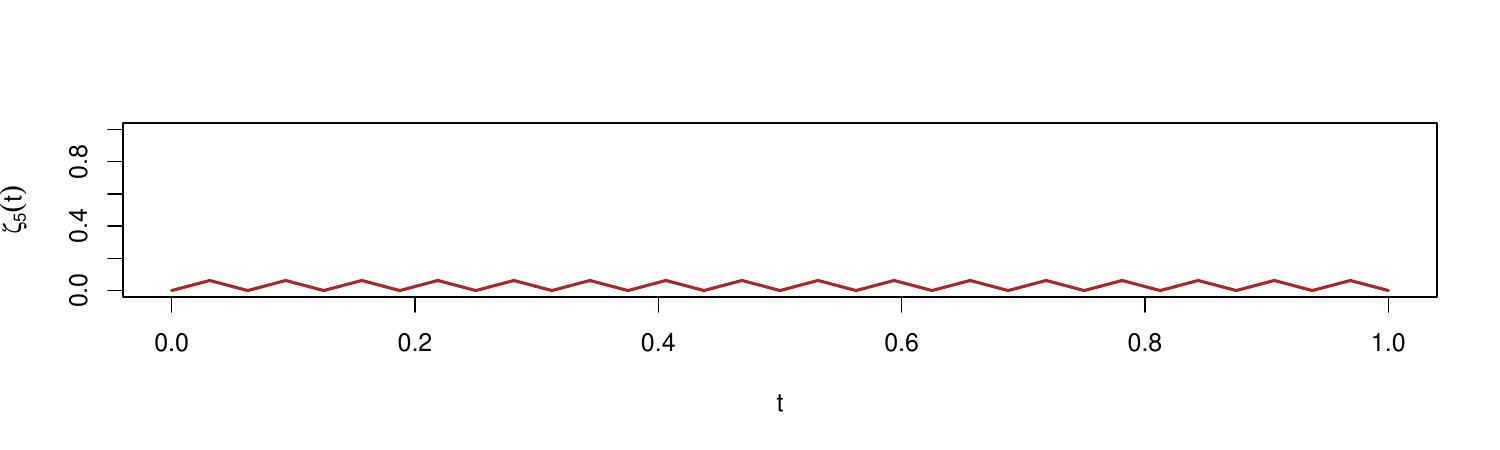}
	\end{subfigure}
	\vspace{-10mm}
	\caption{Iterated triangular functions $\zeta_L$ in Section~\ref{sec:sawtooth} for $L=1,3,5$.  As $L$ increases, the $L_2$ norm of the function $\zeta_L$ itself decays exponentially fast to zero, while the $L_2$ norm of the derivative function $\dot\zeta_L$ remains a constant.}
	\label{fig:sawtooth}
\end{figure}

The function $\zeta_L$ essentially consists of $2^L$ piecewise linear segments with each segment having a slope of constant magnitude $2$.  Consequently, the $L_2$ norm of $\zeta_L$ over the interval $[0,1]$ with respect to the Lebesgue measure is $\{ 2^L \int_0^{1/(2^L)} (2 t)^2 \rmd t \}^{1/2} \sim 1/2^L$, which vanishes as $L$ tends to infinity.  Meanwhile, let the first-order derivative of $\zeta_L$ be $\dot \zeta_L$.  Because the magnitude of $\dot \zeta_L$ is always $2$, its $L_2$
norm (with respect to the same Lebesgue measure on $[0,1]$) is also $2$.
 Thus, the $L_2$ norm of the derivative function $\dot\zeta_L$ remains a constant, even though the $L_2$ norm of the function $\zeta_L$ itself converges exponentially in $L$ to zero, yielding the observation in Section~\ref{sec:test_stat}.

Clearly, this example implies that the $L_2$ convergence between a neural network function and a target, for instance from $\wh g_n$ to $m_0$, does not necessarily imply the convergence at the same rate between their partial derivatives, for instance from $\widehat{g}_{n,j}$ to $m_{0,j}$.  Moreover, this discrepancy becomes more apparent as the depth $L$ increases, but this is precisely the regime of interest for deep neural networks.  Mathematically, this phenomenon can be understood as the absence of a general reverse Poincar\'{e} inequality.  The inequality itself conventionally bounds the function norm by the norm of the derivative of the said function (see Theorem~13.27 in \cite{Leoni2017} for an exact formulation), but the reverse direction does not always hold, as demonstrated by our simple example.

It's important to clarify that we're not asserting that \textit{no} neural network can approximate the first-order derivatives of smooth functions \citep{GuhringKutyniokPetersen2020}.  In our example above, we could easily select a neural network that consistently outputs zero, thereby perfectly approximating both the identically zero function and its derivative.  However, in practical scenarios involving observations with noise, there's no guarantee that such a derivative-approximating neural network will be chosen.  Furthermore, as previously noted, for approximating higher-order derivatives, any ReLU-based neural network will completely fail if we simply settle for the derivatives of the same order (which will always be zero almost everywhere) on the network function directly.

Having demonstrated the irregularities that the derivatives of neural networks may exhibit, we now motivate our smoothing operation.  Our first motivation is the convergence performance obtained from the smoothing operation and the generality of the operation.  As mentioned in Section~\ref{sec:our_method}, our smoothing operation proposed in \eqref{eq:smooth_master} is applicable to general machine learning estimators.  Specifically, our Lemma~\ref{lemma:derivative_bound_via_original} shows that, if we start from a generic regression estimator $\wh g$ for a generic regression function target $m$, then by using the simple smoothing operation in \eqref{eq:smooth_master}, we can easily recover a convergence rate from the \textit{derivative} estimator $\wh g_j^\us$ to the target derivative $m_j$ through the convergence rate from $\wh g_n$ to $m_0$, up to a factor $1/h$ and a bias term; thus the two rates can be made comparable as long as the bandwidth $h$ does not approach zero too fast.  Therefore, we have translated the power of neural networks for regression function estimation into an analogous result for derivative estimation through smoothing.  The bias term occurs naturally because $\wh g_j^\us$ is the \textit{smoothed} version of the naive derivative estimator $\wh g_j$, and hence its target is the \textit{smoothed} target derivative $m_j^\us$ instead of the non-smoothed target derivative $m_j$.  Later, we will apply our lemma on our deep neural network estimators, which will replace $\wh g$ here with $\wh g_n$ in \eqref{eq:def_wh_m_n_highd} or $\wc g_n$ in \eqref{eq:f_check}.
\begin{lemma}
\label{lemma:derivative_bound_via_original}
Suppose that Assumptions~\ref{ass:DGP}\ref{ass:DGP:con_1} and \ref{ass:kernel} hold.
Let $\wh g=\wh g(\wt\bff,x_j): \RR^{\rbar+1}\rightarrow\RR$ be a generic estimator that estimates a generic target $m=m(\bff,x_j):\RR^{r+1}\rightarrow\RR$ with a $L_2$ rate of $r_n$ on an event $\calA$, that is
\begin{align}
\textstyle \{ \int  ( \widehat{g}-m)^2\,\rmd\PP\}^{1/2} \le r_n~\text{on the event}~\calA.
\label{eq:derivative_bound_var}
\end{align}
Moreover, let $m_j^\us$ and $m_j$ be respectively the smoothed and non-smoothed derivatives (with respect to the last argument) of the target function $m$, and suppose that $m_j^\us$ exhibits a bias with respect to $m_j$ as follows:
\begin{align}
\textstyle \{ \int_{\Omega_h} ( m_j^\us-m_j)^2\,\rmd\PP\}^{1/2} \le \mbox{Bias}_n.
\label{eq:derivative_bound_bias}
\end{align}
Then, there exists a constant $c^\us$ (that does not depend on $g$ or $m$) such that the smoothed derivative estimator $\wh g_j^\us$ estimates the non-smoothed $m_j$ with the following rate:
\begin{align}
\textstyle \{ \int_{\Omega_h} ( \wh g_j^\us- m_j )^2\,\rmd\PP\}^{1/2} \le \frac{c^\us}{h} r_n + \mbox{Bias}_n ~\text{on the event}~\calA.
\label{eq:derivative_bound}
\end{align}
\end{lemma}

\begin{proof}
By the triangle inequality, the left-hand side of \eqref{eq:derivative_bound} can be bounded as
\begin{align*}
    \textstyle \{ \int_{\Omega_h} ( \wh g_j^\us- m_j )^2\,\rmd\PP\}^{1/2} \le \{ \int_{\Omega_h} ( \wh g_j^\us- m_j^\us )^2\,\rmd\PP\}^{1/2}  + \{ \smallint_{\Omega_h} ( \wh m_j^\us- m_j )^2\,\rmd\PP\}^{1/2} .
\end{align*}
The first term on the right-hand side above can be bounded by
Lemma~\ref{lemma:derivative_bound_via_original_specific} in conjunction with the rate in \eqref{eq:derivative_bound_var}, and the second term on the right-hand side above is simply bounded by the condition given in \eqref{eq:derivative_bound_bias}.
\end{proof}

The second motivation for our smoothing operation comes from traditional non-parametric regression.  It is well known that in the latter setting, the optimal bandwidth for the derivative estimator of either the density or regression function should be larger than the optimal bandwidth for the estimator of the density or regression function itself (see, for instance, \cite[Theorem~1]{GasserMuller1984}, \cite[Theorem~3.2]{HallMarron1987} and \cite[p.~226]{HardleMarronWand1990}), due to the former estimator possessing a larger variance at the same bandwidth.  Another example in non-parametric regression calling for over-smoothing comes from the so-called irrelevant regressors \cite[Section~2.2.4]{LiRacine2006}, which are those predictors that do not materially contribute to the regression outcome.  Irrelevant regressors are not about derivative estimation per se, but by their very definition they are similar to those coordinates $X_k$ that satisfy the null hypothesis \eqref{hypo_null} in our screening test.  In kernel regression, it is known that the optimal bandwidth for the irrelevant regressors should diverge to infinity in order to ``weed out'' these regressors more efficiently.  (Furthermore, if a regressor is not irrelevant, but contributes linearly to the regression outcome, it's conjectured that its associated optimal bandwidth should diverge to infinity as well \cite[Section~2.4.1]{LiRacine2006}.)  Thus, from the point of view of both derivative estimation and screening in traditional non-parametric regression, it is not surprising that smoothing could benefit our derivative estimation and screening test.

\section{Supporting details and preliminary proofs for Section~\ref{sec:regression_function}}
\label{app_sec:proof_sec_3}

\subsection{Bias from diversified projection and neural network approximation}
\label{sec:approx_error}

In this section we show two parallel results: first, in Lemma~\ref{lemma:approx_error} we show that the true regression function $m_0=m_0(\bff,x_j):\RR^{r+1}\rightarrow\RR$ is well approximated by $\wt g_n=\wt g_n(\wt\bff,x_j)\equiv \wt m_n(\bH^\dagger \wt\bff, x_j)$ where
\begin{align}
\label{eq:def_wt_m_n_highd}
    \textstyle \wt m_n= \argmin_{m=m(\bff,x_j)\in\calFn(r+1)} \sup_{(\bff,x_j)\in [-2b,2b]^r\times[-b,b]} |(m-m_0)(\bff,x_j)| ,
\end{align}
and second, in Lemma~\ref{lemma:approx_error_alpha} we show that analogously the true Riesz representer under the null hypothesis $\alphanull$ is well approximated by
\begin{align}
\label{talpha}
\textstyle \wt\alpha_n = \argmin_{\alpha\in\calFn(\rbar+1)} \sup_{(\wt\bff,x_j)\in[-c_b b, c_b b]^\rbar\times[-b,b]} |(\alpha - \alphanull)(\wt\bff,x_j)|.
\end{align}

We start with Lemma~\ref{lemma:approx_error}.  Let $\nuan$ be the approximation rate in the sup norm attained by $\wt m_n$ in \eqref{eq:def_wt_m_n_highd}, that is
\begin{align*}
    \textstyle \nuan = \sup_{(\bff,x_j)\in[-2b,2b]^r\times[-b,b]} | (\wt m_n-m_0)(\bff,x_j) | .
\end{align*}
As can be seen from the lemma and its proof below, the price we pay for the approximation is decomposed into the term $\deltaf$, which in fact diminishes with respect to the ambient dimension $d$ and which results from approximating $m_0$ using diversified factors through $g_0=g_0(\wt\bff,x_j)$ given in \eqref{eq:def_g0} (Lemma~\ref{lemma:factor}), and the term $\nuan$, which results from subsequently approximating $g_0$ by the neural network function $\wt g_n$ (Lemma~\ref{lemma:neural_approx_error_m}).  Our lemma is a straightforward extension of Ineq.~(F.6) in the appendix of \cite{FanGu2023factor} by including the variable of interest $X_j$.  For completeness, we will provide a sketch of the proof.
\begin{lemma}[Modification of Inequality~(F.6) in \cite{FanGu2023factor}]
\label{lemma:approx_error}
Under Assumptions~\ref{ass:DGP} to \ref{asumidio} and \ref{assum_fun_class} to \ref{ass:W},
\begin{align*}
   \textstyle \left[ \int \{ \wt g_n(\wt\bff, x_j) - m_0(\bff, x_j) \}^2 \rmd\PP \right]^{1/2} \lesssim \deltaf + \nuan.
\end{align*}
\end{lemma}

\begin{proof}
First, we have the easy bound
\begin{align*}
    & \textstyle [ \int \{ \wt g_n(\wt\bff, x_j) - m_0(\bff, x_j) \}^2 \rmd\PP ]^{1/2} \\
    & \textstyle \le [ \int \{ g_0(\wt\bff, x_j) - m_0(\bff, x_j) \}^2 \rmd\PP ]^{1/2} + [ \int \{ \wt g_n(\wt\bff, x_j) - g_0(\wt\bff, x_j) \}^2 \rmd\PP ]^{1/2} .
\end{align*}
The first and second term in the last line above are in turn handled by Lemma~\ref{lemma:factor} and Lemma~\ref{lemma:neural_approx_error_m} below respectively, whose rates when combined yield the conclusion of the lemma.
\end{proof}

\begin{lemma}
\label{lemma:factor}
Under the same assumptions as Lemma~\ref{lemma:approx_error},
\begin{align*}
    \textstyle \left[ \int \{ g_0(\wt\bff, x_j) - m_0(\bff, x_j) \}^2 \rmd\PP \right]^{1/2} \lesssim \deltaf .
\end{align*}
\end{lemma}
\begin{proof}
By the definition of diversified projection matrix, $\operatorname{rank}(\bH_{j}^{\dagger})=r$,
which implies $\bH^\dagger\bH=\mathbf{I}_{r}$.
Starting from the decomposition
\begin{align}
\label{eq:F_tilde_decomp}
\wt\bfF_i = d^{-1} \bW^\top \bX_i=d^{-1} \bW^\top \bB\bfF_i+d^{-1} \bW^\top\bmu_i = \bH \bfF_i+d^{-1} \bW^\top\bmu_i,
\end{align}
we have
\begin{align}
&m_0(\bfF_i,X_{i,j})-g_0(\wt\bfF_i,X_{i,j}) \nonumber \\
&= m_0(\bfF_i,X_{i,j})-g_0(\bH \bfF_i+d^{-1}\bW^\top\bmu_i,X_{i,j}) \nonumber \\
& = m_0(\bfF_i,X_{i,j})-m_0(\bH^\dagger \bH
 \bfF_i+d^{-1}\bH^\dagger\bW^\top\bmu_i,X_{i,j}) \nonumber \quad \mbox{ by the definition of $g_0$ in \eqref{eq:def_g0}}\\
& = m_0(\bfF_i,X_{i,j})-m_0(\bfF_i+d^{-1}\bH^\dagger \bW^\top\bmu_i,X_{i,j}) \quad \mbox{ by the definition of $\bH^\dagger$}.
\label{factor_approximation_01}
\end{align}
By Assumption~\ref{assum_fun_class}, $m_0(\bff,x_j)$ is Lipschitz in the sense that $|m_0(\bff,x_j)-m_0(\bff',x_j)|\le C_{\textup{L}}\|\bff-\bff'\|$ (for all $x_j$).  It follows from this and \eqref{factor_approximation_01} that
\begin{align}
    | m_0(\bfF_i,X_{i,j}) - g_0(\wt\bfF_i,X_{i,j}) |\le C_{\textup{L}} d^{-1}\left\Vert \bH^\dagger\bW^\top\bmu_i\right\Vert_2 \le C_{\textup{L}} d^{-1}{\left\Vert \bH^\dagger\right\Vert_{\textup{op}}} \left\Vert \bW^\top\bmu_i\right\Vert_2 .
    \label{factor_approximation_02}
\end{align}
By the derivation of (F.15) in \cite{FanGu2023factor} and Assumption~\ref{ass:W},
\begin{align}
\label{eq:H_bound}
\| \bH^\dagger \|_{\textup{op}} \le \left\{ v_{\min}( \bH^\dagger ) \right\}^{-1} \lesssim 1.
\end{align}
Next we bound $\Vert d^{-1} \bW^\top\bmu_i\Vert_2$ in \eqref{factor_approximation_02}.  It follows from the linearity of expectation and then Assumption~\ref{asumidio} that
\begin{align}
& \textstyle \EE [ \| \bW^\top\bmu_i \|^2] = \EE\left[ \sum_{k=1}^\rbar (\sum_{j=1}^d \textup{W}_{j,k} u_{i,j} )^2\right] \textstyle = \sum_{k=1}^\rbar \sum_{j=1}^d \textup{W}_{j,k}^2\EE\left[u_{i,j}^2\right] + \sum_{j\neq j^{\prime}} \textup{W}_{j,k} \textup{W}_{j^{\prime},k} \EE\left[u_{i,j}u_{i,j^{\prime}}\right] \nonumber \\
& \textstyle \le \rbar \max_{j,k} \left|\textup{W}_{j,k}\right| \sum_{j=1}^d \EE\left[u_{i,j}^2\right] + \rbar\max_{j,k} \left|\textup{W}_{j,k}\right|^2 \sum_{j\neq j^{\prime}}\left|\EE\left[u_{i,j}u_{i,j^{\prime}}\right]\right| \textstyle \lesssim  \rbar d_u.
\label{eq:W_mu_bound}
\end{align}

Finally, we conclude by taking the expectation of the square of both sides of \eqref{factor_approximation_02} and subsequently applying \eqref{eq:H_bound} and \eqref{eq:W_mu_bound}.
\end{proof}

\begin{lemma}
\label{lemma:neural_approx_error_m}
Under the same assumptions as Lemma~\ref{lemma:approx_error},
\begin{align*}
\textstyle \left[ \int \{ \wt g_n(\wt\bff, x_j) - g_0(\wt\bff, x_j) \}^2 \rmd\PP \right]^{1/2} \lesssim \nuan + \deltaf .
\end{align*}
\end{lemma}
\begin{proof}
Following essentially the same derivation on p.~10 in the appendix of \cite{FanGu2023factor},
\begin{align*}
& \textstyle \int \{ \wt g_n(\wt\bff, x_j) - g_0(\wt\bff, x_j) \}^2 \rmd\PP = \int \{ \wt g_n(\wt\bff, x_j) - g_0(\wt\bff, x_j) \}^2 \ind\{\bH^\dagger \wt\bff\in[-2b,2b]^r\}  \rmd\PP \\
& \textstyle + \int \{ \wt g_n(\wt\bff, x_j) - g_0(\wt\bff, x_j) \}^2 \ind\{\bH^\dagger \wt\bff\notin[-2b,2b]^r\} \rmd\PP .
\end{align*}
By the definition of $\wt g_n$, $g_0$ and $\nuan$, the first term on the right-hand side above is reduced to
\begin{align*}
& \textstyle \int \{ \wt g_n(\wt\bff, x_j) - g_0(\wt\bff, x_j) \}^2 \ind\{\bH^\dagger \wt\bff\in[-2b,2b]^r\} \rmd\PP \\
& \textstyle = \int \{ \wt m_n(\bH^\dagger \wt\bff, x_j) - m_0(\bH^\dagger \wt\bff, x_j) \}^2 \ind\{\bH^\dagger \wt\bff\in[-2b,2b]^r\} \rmd\PP \le \nuan^2.
\end{align*}

Meanwhile, for the second term, by Assumption~\ref{assum_fun_class}, Markov's inequality, and finally the reasoning at the top of p.~11 in the appendix of \cite{FanGu2023factor},
\begin{align*}
& \textstyle \int \{ \wt g_n(\wt\bff, x_j) - g_0(\wt\bff, x_j) \}^2 \ind\{\bH^\dagger \wt\bff\notin[-2b,2b]^r\} \rmd\PP \\
& \textstyle \le (\|\wt g_n(\wt\bff, x_j) \|_{L_\infty} + \|g_0(\wt\bff, x_j)\|_{L_\infty})^2 \PP(\bH^\dagger\wt\bfF\notin[-2b,2b]^r) \\
& \textstyle \lesssim (2M_\infty)^2 \PP(\| d^{-1} \bH^\dagger \bW^\top \bmu \|_2 \ge b) \le \EE[\| d^{-1} \bH^\dagger \bW^\top \bmu \|_2^2] \lesssim \deltaf^2.
\end{align*}
Combining the two equation displays above yields the conclusion of the lemma.
\end{proof}

\begin{lemma}[Bias for approximating the Riesz representer]
\label{lemma:approx_error_alpha}
Under Assumptions~\ref{ass:DGP} to \ref{ass:u_t_n} and \ref{assum_fun_class} to \ref{ass:u_t_n_misc},
\begin{align*}
    \textstyle \left[ \int \{ \wt\alpha_n(\wt\bff, x_j) - \alphanull(\wt\bff, x_j) \}^2 \rmd\PP \right]^{1/2} \lesssim \nuan + \deltaf .
\end{align*}
\end{lemma}

\begin{proof}
As usual, let $\wt F_k$ denote the $k$-th element of $\wt\bfF$ and moreover let $[\bW^\top]_{k\cdot}$ denote the $k$-th row of $\bW^\top$.  Using the decomposition~\eqref{eq:F_tilde_decomp},
\begin{align}
\label{eq:F_wt_decomp}
|\wt F_k| \le |\bH_{k\cdot} \bfF| + |d^{-1} [\bW^\top]_{k\cdot} \bmu| \le |\bH_{k\cdot} \bfF| + \|d^{-1} \bW^\top \bmu\|_2 .
\end{align}
By Assumption~\ref{ass:u_t_n_misc}, the first term $|\bH_{k\cdot} \bfF|$ in the last step above is uniformly bounded over $k\in\{1,\dots,\rbar\}$ by $c_b b/2$ with probability $1-\deltaf^2$.  Next, the second term has been bounded in \eqref{eq:W_mu_bound} and yields, by Markov's inequality,
\begin{align*}
\PP(\| d^{-1} \bW^\top \bmu \|_2\ge c_b b/2) \lesssim \deltaf^2 .
\end{align*}
Hence, we conclude that with probability at least $1-C \deltaf^2$, $\wt\bfF\in[-c_b b, c_b b]^\rbar$.  Then, using argument similar to that in the proof of Lemma~\ref{lemma:neural_approx_error_m}, the lemma holds because
\begin{align*}
& \textstyle \int \{ \wt\alpha_n(\wt\bff, x_j) - \alphanull(\wt\bff, x_j) \}^2 \rmd\PP = \int \{ \wt\alpha_n(\wt\bff, x_j) - \alphanull(\wt\bff, x_j) \}^2 \ind\{ \wt\bff\in[-c_b b, c_b b]^\rbar\} \rmd\PP \\
& \textstyle + \int \{ \wt\alpha_n(\wt\bff, x_j) - \alphanull(\wt\bff, x_j) \}^2 \ind\{ \wt\bff\notin[-c_b b, c_b b]^\rbar\} \rmd\PP \textstyle \lesssim \nuan^2 + \deltaf^2 . \qedhere
\end{align*}
\end{proof}

\subsection{The stability of conditional expectation in Assumption~\ref{ass:W}}
\label{sec:cond_exp_stability}
{This subsection addresses the issue that while model \eqref{eq:reg_model} yields the proper centering of the noise, namely $\EE[\epsilon|\bfF,X_j]=0$, after replacing $\bfF$ by $\wt\bfF$ we may no longer have $\EE[\epsilon|\wt\bfF,X_j]=0$; instead, what is properly centered here is the adjusted noise term $\wt\epsilon\equiv Y-\EE[Y|\wt\bfF,X_j]$.  Naturally we will need to handle the difference $|\epsilon-\wt\epsilon|=|\{Y-m_0(\bfF,X_j)\} - \{Y-\EE[Y|\wt\bfF,X_j]\}| = |\EE[Y|\bfF,X_j]-\EE[Y|\wt\bfF,X_j]|$, whose $L_2$ norm is precisely the left-hand side of \eqref{eq:cond_exp_stability}.
This difference concerns the closeness between the two conditional expectations $\EE[Y|\wt\bfF,X_j]$ and $m_0(\bfF,X_j) = \EE[Y|\bfF,X_j]$ under a perturbation of the conditioning variables.
A simple version of such a ``stability of conditional expectation'' result appears in \cite{WlodzimierzWlodzimierz1992}; however, only one-dimensional random variables were examined and no explicit convergence rate was provided.  We will provide in Lemma~\ref{lemma:cond_exp_stability} such a stability result {which is more appropriate for our context}.

For brevity, define $m^*(\bfF,\wt\bfF,X_j)=\EE[Y|\bfF,\wt\bfF,X_j]$.  Compared to the full regression function $m_0^*(\bfF,\bX)=\EE[Y|\bfF,\bX]$ defined below \eqref{hypo_alt}, $m^*(\bfF,\wt\bfF,X_j)$ is obtained by conditioning on a smoother $\sigma$-algebra (because the $\sigma$-algebra generated by $(\bfF,\wt\bfF,X_j)$ is contained in the $\sigma$-algebra generated by $(\bfF,\bX)$) and in fact $m^*(\bfF,\wt\bfF,X_j)=\EE[m_0^*(\bfF,\bX)|\bfF,\wt\bfF,X_j]$.  In turn, both $m_0(\bfF,X_j)$ and $\EE[Y|\wt\bfF,X_j]$ involved in \eqref{eq:cond_exp_stability} can be obtained by further smoothing $m^*(\bfF,\wt\bfF,X_j)$ (as the proof of Lemma~\ref{lemma:cond_exp_stability} will show)}.  Define $\bXi= \bW^\top \bmu \in \RR^\rbar$.  Let $\bfF$ have marginal density $p_{\bfF}$ and let $\bXi$ and $u_j$ have joint density $p_{\bXi, u_j}$. Further define the random vector $\buc=C_\uf^{-1}\deltaf^{-1}\frac{1}{d}(\bH^\dagger \bXi-\bH^\dagger \bz)\in\RR^r$ where $\bz\in\RR^\rbar$ and $C_\uf$ is the constant in condition~\ref{exp_stability:con_5} in Lemma~\ref{lemma:cond_exp_stability}, and the random functions $g_\un$ and $g_\ud$ as, for $a\in\RR$,
\begin{align*}
    g_\un(a) & \textstyle = \int m^*(\bfF + a \buc, \bH\bfF+\frac{1}{d} \bz, X_j)\,p_{\bfF}(\bfF + a \buc) p_{\bXi,u_j}(\bz, u_j - \bB_{j \cdot} a \buc ) \rmd\bz , \\
    g_\ud(a) & \textstyle = \int p_{\bfF}(\bfF + a \buc) p_{\bXi,u_j}(\bz, u_j - \bB_{j \cdot} a \buc ) \rmd\bz .
\end{align*}
The reason we introduce these functions is because, in terms of them, by the derivations in the proof of Lemma~\ref{lemma:cond_exp_stability},
\begin{align*}
\EE[m^*(\bfF, \wt\bfF, X_j)|\wt\bfF,X_j] = \frac{g_\un(C_\uf \deltaf)}{g_\ud(C_\uf \deltaf)} , \quad \EE[ m^*(\bfF,\wt\bfF,X_j)|\bfF,X_j] = \frac{g_\un(0)}{g_\ud(0)} ;
\end{align*}
hence, the left-hand side of \eqref{eq:cond_exp_stability} can be expressed with the help of the mean value theorem applied to the expressions above.

\begin{lemma}[Conditions for the validity of inequality~\eqref{eq:cond_exp_stability}]
\label{lemma:cond_exp_stability}
Suppose that Assumptions~\ref{ass:DGP} to \ref{asumidio}, and \ref{ass:W} excluding \eqref{eq:cond_exp_stability}, hold.  Moreover, suppose that
\begin{enumerate}[label=(\roman*)]\label{ratexx:whole}
\item\label{exp_stability:con_1}
The function $m^*=m^*(\bff,\wt\bff,x_j)$ where $\bff\in\RR^r$ and $\wt\bff\in\RR^\rbar$ is bounded, and $m^*$ admits a bounded gradient in the first argument $\bff$ which we denote by $\nabla m^*(\bff,\wt\bff,x_j)$;
\item\label{exp_stability:con_2}
$\bfF$ admits a density $p_{\bfF}$ with gradient $\nabla p_{\bfF}$;
\item\label{exp_stability:con_3}
$\bmu$ and $\bfF$ are independent;
\item\label{exp_stability:con_4}
$\bXi$ and $u_j$ admits a joint density $p_{\bXi,u_j}=p_{\bXi,u_j}(\bxi,v_j)$, and $p_{\bXi,u_j}$ admits a partial derivative with respect to $v_j$ which we denote by $\dot p_{\bXi,u_j}$;
\item\label{exp_stability:con_5}
The following bounds hold:
\begin{gather*}
    \textstyle \EE\left[ \sup_{0<a<C_\uf \deltaf} \left\{ \frac{ \int \nabla p_{\bfF}(\bfF + a \buc)^\top p_{\bXi, u_j}(\bz, u_j - \bB_{j \cdot} a \buc ) \buc \rmd\bz }{ g_\ud(a) } \right\}^2 \right] < \infty , \\
    \textstyle \EE\left[ \sup_{0<a<C_\uf \deltaf} \left\{ \frac{ \int p_{\bfF}(\bfF + a \buc) \dot p_{\bXi, u_j}(\bz, u_j - \bB_{j \cdot} a \buc ) \bB_{j \cdot} \buc \rmd\bz }{ g_\ud(a) } \right\}^2 \right] < \infty ;
\end{gather*}
\item\label{exp_stability:con_6}
The random vector $\bXi$ satisfies $\|\frac{1}{d}\bH^\dagger \bXi\|_2\le C_\uf \deltaf$ almost surely for a constant $C_\uf$.
\end{enumerate}
Then Inequality~\eqref{eq:cond_exp_stability} holds.
\end{lemma}

Some remarks about the conditions in Lemma~\ref{lemma:cond_exp_stability} are in order.  First, condition~\ref{exp_stability:con_6} is not directly used in the proof of Lemma~\ref{lemma:cond_exp_stability} per se, but it facilitates our subsequent discussion about condition~\ref{exp_stability:con_5}.  In any case, the proof of Lemma~\ref{lemma:factor} shows that $\EE \|\frac{1}{d}\bH^\dagger \bXi\|_2\lesssim \deltaf$, so condition~\ref{exp_stability:con_6} is reasonable.
 Condition~\ref{exp_stability:con_5} may appear complex at first; however, this condition simplifies vastly if we simply set $a=0$.  Take the second inequality in condition~\ref{exp_stability:con_5} as an example.  When $a=0$, the expectation on the left-hand side reduces to
\begin{align*}
    \EE\left[ \left\{ \frac{ \int \dot p_{\bXi, u_j}(\bz, u_j ) \bB_{j \cdot} \buc \rmd\bz }{ \int p_{\bXi, u_j}(\bz, u_j ) \rmd\bz } \right\}^2 \right] \stackrel{\text{condition }\ref{exp_stability:con_6}}{\lesssim} \EE\left[ \left\{ \frac{ \int \dot p_{\bXi, u_j}(\bz, u_j ) \rmd\bz }{ \int p_{\bXi, u_j}(\bz, u_j ) \rmd\bz } \right\}^2 \right] = \EE\left[ \left\{ \frac{ \dot p_{u_j}(u_j) }{ p_{u_j}(u_j) } \right\}^2 \right]
\end{align*}
where $\dot p_{u_j}$ denotes the derivative of the marginal density $p_{u_j}$ of $u_j$ and in the last step we have exchanged integration and differentiation.  We recognize the rightmost term above as simply the Fisher information for location of the idiosyncratic term $u_j$, and the boundedness of this information is a very weak condition.  Thus, condition~\ref{exp_stability:con_5} can be seen as a slightly strengthened version of the boundededness of this Fisher information under a small perturbation by the random variable $a \buc$ where $0<a<C_\uf \deltaf$.

\begin{proof}[Proof of Lemma~\ref{lemma:cond_exp_stability}]
By the smoothing property of conditional expectations (see, for instance, property~(10) on p.~348 of \cite{Resnick2005}), $\EE[m^*(\bfF,\wt\bfF,X_j)|\bfF,X_j] = \EE[\EE[Y|\bfF,\wt\bfF,X_j]|\bfF,X_j] = \EE[Y|\bfF,X_j] = m_0(\bfF,X_j)$, and similarly $\EE[m^*(\bfF,\wt\bfF,X_j)|\wt\bfF,X_j] = \EE[Y|\wt\bfF,X_j]$.  Thus,
\begin{align*}
&\EE[ \{ \EE[Y|\wt\bfF,X_j] - m_0(\bfF,X_j) \}^2 ] = \EE[ \{ \EE[m^*(\bfF,\wt\bfF,X_j)|\wt\bfF,X_j] - \EE[m^*(\bfF,\wt\bfF,X_j)|\bfF,X_j] \}^2 ] .
\end{align*}
We now give explicit expressions for $\EE[m^*(\bfF, \wt\bfF, X_j)|\wt\bfF,X_j]$ and $\EE[m^*(\bfF,\wt\bfF,X_j)|\bfF,X_j]$.  We start from the first of these, namely
\begin{align}
\label{eq:cond_exp_Ftilde}
\EE[m^*(\bfF, \wt\bfF, X_j)|\wt\bfF,X_j] = \frac{\int m^*(\bff, \wt\bfF, X_j) p_{\bfF,\wt\bfF,X_j}(\bff,\wt\bfF,X_j)\rmd\bff}{p_{\wt\bfF,X_j}(\wt\bfF,X_j)} ,
\end{align}
where $p_{\wt\bfF,X_j}$ and $p_{\bfF,\wt\bfF,X_j}$ denote respectively the joint densities of $\wt\bfF$ and $X_j$, and of $\bfF,\wt\bfF$ and $X_j$.
Recalling \eqref{eq:F_tilde_decomp} which gives $\wt\bfF=\bH\bfF+\frac{1}{d}\bXi$, and letting $\bB_{j \cdot}$ denote the $j$-th row of $\bB$ so $X_j=\bB_{j \cdot}\bfF+u_j$, we note that for $\wt\ba\in\RR^\rbar$ and $b\in\RR$, that $(\wt\bfF,X_j)=(\bH\bfF+\frac{1}{d}\bXi,\bB_{j \cdot}\bfF+u_j)=(\wt\ba^\top,b)^\top$ is equivalent to $\bXi=d(\wt\ba-\bH\bfF)$ and $u_j=b-\bB_{j \cdot} \bfF = b - \bB_{j \cdot} \bH^\dagger (\wt\bfF-\frac{1}{d} \bXi)$.  Then, by the usual convolution argument, the joint density of $p_{\wt\bfF,X_j}$ is given by
\begin{align*}
    & \textstyle p_{\wt\bfF,X_j}(\wt\ba,b) = \frac{\rmd^2}{\rmd\wt\ba\rmd b} \PP(\wt\bfF\le\wt\ba,X_j\le b) = \frac{\rmd^2}{\rmd\wt\ba\rmd b} \PP(\bH \bfF+\frac{1}{d}\bXi \le \wt\ba, \bB_{j \cdot}\bfF + u_j\le b) \\
    & \textstyle = \frac{\rmd^2}{\rmd\wt\ba\rmd b} \PP(\bXi \le d (\wt\ba - \bH \bfF), u_j\le b-\bB_{j \cdot}\bfF) = \frac{\rmd^2}{\rmd\wt\ba\rmd b} \int \int^{d (\wt\ba - \bH \bff)} \int^{b-\bB_{j \cdot}\bff} p_{\bfF}(\bff) p_{\bXi,u_j}(\bm{w},v) \rmd v \rmd\bm{w} \rmd\bff \\
    & \textstyle = d^\rbar \int p_{\bfF}(\bff) p_{\bXi,u_j}(d(\wt\ba-\bH\bff),b-\bB_{j \cdot}\bff) \rmd\bff
\end{align*}
where the last step follows by the Leibniz integral rule.  Now, $p_{\wt\bfF,X_j}(\wt\bfF,X_j)$ can serve as the denominator of $\EE[m^*(\bfF, \wt\bfF, X_j)|\wt\bfF,X_j]$ in \eqref{eq:cond_exp_Ftilde}. Analogously, in the numerator of \eqref{eq:cond_exp_Ftilde}, for $\ba\in\RR^r$,
\begin{align*}
    \textstyle p_{\bfF,\wt\bfF,X_j}(\ba,\wt\ba,b) = d^\rbar  p_{\bfF}(\ba) p_{\bXi,u_j}(d(\wt\ba-\bH\ba),b-\bB_{j \cdot}\ba) .
\end{align*}
Consequently,
\begin{align*}
&\EE[m^*(\bfF, \wt\bfF, X_j)|\wt\bfF,X_j] \\
&= \dfrac{ \int m^*(\bff, \wt\bfF, X_j)\,p_{\bfF}(\bff)\,p_{\bXi,u_j}(d(\wt\bfF-\bH\bff),X_j-\bB_{j \cdot}\bff) \rmd\bff }{ \int p_{\bfF}(\bff)\,p_{\bXi,u_j}(d(\wt\bfF-\bH\bff), X_j-\bB_{j \cdot}\bff) \rmd\bff } \\
&= \dfrac{ \left( \splitdfrac{\int m^*(\bfF+\frac{1}{d} \bH^\dagger \bXi-\frac{1}{d} \bH^\dagger \bz, \bH\bfF+\frac{1}{d} \bz, X_j)\,p_{\bfF}(\bfF+\frac{1}{d} \bH^\dagger \bXi-\frac{1}{d} \bH^\dagger \bz)}{\times p_{\bXi,u_j}(\bz, u_j - \bB_{j \cdot} (\frac{1}{d} \bH^\dagger \bXi - \frac{1}{d} \bH^\dagger \bz )) \rmd\bz} \right) }{ \int p_{\bfF}(\bfF+\frac{1}{d} \bH^\dagger \bXi-\frac{1}{d} \bH^\dagger \bz)\,p_{\bXi,u_j}(\bz, u_j  - \bB_{j \cdot} (\frac{1}{d} \bH^\dagger \bXi - \frac{1}{d} \bH^\dagger \bz) ) \rmd\bz } = \frac{g_\un(C_\uf \deltaf)}{g_\ud(C_\uf \deltaf)} .
\end{align*}
where in the transition to the second to last step we have conducted a change of the integration variable from $\bff$ to $\bz=d(\wt\bfF-\bH\bff)$. Meanwhile, for the second conditional expectation $\EE[ m^*(\bfF,\wt\bfF,X_j)|\bfF,X_j]$, we can analogously obtain
\begin{align*}
\EE[ m^*(\bfF,\wt\bfF,X_j)|\bfF,X_j] = \dfrac{ \int m^*(\bfF,\bH\bfF+\frac{1}{d} \bz,X_j)\,p_{\bXi,u_j}(\bz,u_j) \rmd\bz }{ p_{u_j}(u_j) } = \frac{g_\un(0)}{g_\ud(0)}
\end{align*}
Next, denote by $\dot g_\un$ and $\dot g_\ud$ the derivatives of $g_\un$ and $g_\ud$ respectively.  By the mean value theorem, for some random $\wt a$ strictly between $0$ and $C_\uf \deltaf$,
\begin{align*}
    &\textstyle | \EE[m^*(\bfF, \wt\bfF, X_j)|\wt\bfF,X_j] - \EE[ m^*(\bfF,\wt\bfF,X_j)|\bfF,X_j] | = \left| \frac{g_\un(C_\uf \deltaf)}{g_\ud(C_\uf \deltaf)} - \frac{g_\un(0)}{g_\ud(0)} \right| \\
    &\textstyle = \left| \left\{ \frac{\dot g_\un(\wt a) }{ g_\ud(\wt a) } - \frac{g_\un(\wt a)}{g_\ud(\wt a)} \frac{\dot g_\ud(\wt a)}{ g_\ud(\wt a) } \right\} \right| C_\uf \deltaf \stackrel{\text{condition }\ref{exp_stability:con_1}}{\lesssim} \left\{ \left| \frac{\dot g_\un(\wt a) }{ g_\ud(\wt a) } \right| + \left| \frac{\dot g_\ud(\wt a)}{ g_\ud(\wt a) } \right| \right\} C_\uf \deltaf.
\end{align*}
Finally, Inequality~\eqref{eq:cond_exp_stability} follows by applying condition~\ref{exp_stability:con_5} to the last term above.
\end{proof}

\subsection{Useful bounds on key stochastic terms}
\label{sec:thm:m_hat_est_master_stochastic}
We first present a couple of results that will be used repeatedly for bounding some key stochastic terms throughout our proofs, for instance in the proof of Theorem~\ref{thm:m_hat_est_master} in Section~\ref{sec:proof_thm:m_hat_est_master}.  Here, the key element is to utilize sharp weighted empirical process techniques in, for instance, \cite{GineKoltchinskiiJonWellner2003} to obtain, for each indexing function (for instance $\theta$ on the left-hand sides of \eqref{eq:m_n_hat_ratio_rate_1_master} and \eqref{eq:m_n_hat_ratio_rate_2_master}), a bound on the empirical sum of the indexing function that is related to the norm of that particular indexing function (for instance, $\| \theta(Y_i,\bfF_i,\bX_i) \|_{L_2}$ on the right-hand sides of \eqref{eq:m_n_hat_ratio_rate_1_master} and \eqref{eq:m_n_hat_ratio_rate_2_master}).  Bounds of this kind become sharper when the norm of the indexing function itself is small, and are key to obtaining good rates for $\wh g_n$ and our subsequent derivative estimators.

Define the quantity $\nusn$ that will feature prominently in the our subsequent Lemma~\ref{lem:square_deviation} and Proposition~\ref{prop:m_n_hat_ratio_rate_master}, where we recall again from Def.~\ref{def:VC_subgraph} that $V_{\calFn(\rbar+1)}$ is the VC-index of the subgraphs associated with the neural network class $\calFn(\rbar+1)$:
\begin{align}
\label{eq:nu_n_generic}
\textstyle \nusn = \left\{ \frac{V_{\calFn(\rbar+1)} \log(n)}{n} \right\}^{1/2}.
\end{align}
\begin{lemma}[Lemma~2 in \cite{FanGu2023factor}]
\label{lem:square_deviation}
Let $\calG$ be a class of uniformly bounded functions on $R^{1+r+d}$, for which we can choose a constant (bounded) envelope function $G$.  Assume that the class $\calG$ is of VC-type (Def.~\ref{def:VC}) with index $v=2V_{\calFn(\rbar+1)}$ (so in particular the covering number for $\calG$ admits the bound \eqref{eq:covering_number_bound_master} with $v=2V_{\calFn(\rbar+1)}$).  Then, there exist some universal constants $c_1$, $c_2$, $c_3$ such that for all $\sqrt{s'}\ge c_1 \sqrt{n} \nusn/\{\zeta (1-\zeta)\}$ with $0<\zeta<1$, on an event $\calA_{1,\zeta,s'}$ satisfying
\begin{align}
\label{eq:square_deviation_prob}
\PP(\calA_{1,\zeta,s'})\ge 1-c_2\exp(-c_3 \zeta^2(1-\zeta) s')
\end{align}
we have
\begin{align}
\label{eq:square_deviation}
    \textstyle \forall g\in\calG, \quad \left| \frac{1}{n} \sum_{i=1}^n \{ g^2(Y_i,\bfF_i,\bX_i) - \EE[g^2(Y,\bfF,\bX)] \} \right| \le \zeta \left\{  \int g^2(Y_i,\bfF_i,\bX_i) \rmd\PP + \frac{s'}{n} \right\} .
\end{align}
\end{lemma}

\begin{proof}
As hinted in the lemma statement, this is essentially Lemma~2 in \cite{FanGu2023factor} with minor notational differences.  Alternatively, one could rely on Lemma~2 in \cite{GineKoltchinskiiJonWellner2003}.  We omit the proof.
\end{proof}

Recall the adjusted noise $\wt\epsilon=Y-\EE[Y|\wt\bfF,X_j]$ from Section~\ref{sec:cond_exp_stability}.  Then, define $\wt\epsilon_i=Y_i-\EE[Y_i|\wt\bfF_i,X_{i,j}]$, $i=1,\dots,n$, which form i.i.d.\,copies of $\wt\epsilon$.

\begin{proposition}
\label{prop:m_n_hat_ratio_rate_master}
Let $\calG$ be a class of uniformly bounded functions on $R^{\rbar+1}$, for which we can choose a constant (bounded) envelope function $G$.  Assume that the class $\calG$ is of VC-type with index $v=2V_{\calFn(\rbar+1)}$. Then, on an event $\calA_{2,s}$ with $\PP(\calA_{2,s})\ge 1-12 e^{-s}$, for a constant $c_4$,
\begin{align}
\label{eq:m_n_hat_ratio_rate_1_master}
    \textstyle \forall \theta\in\calG, \quad \textstyle |\frac{1}{n} \sum_{i=1}^n \wt\epsilon_i \theta(\wt\bfF_i,X_{i,j}) | \le c_4  \left\{ \| \theta(\wt\bfF,X_j) \|_{L_2} + \nusn \right\} \left\{ \nusn + \sqrt{\frac{s}{n}} \right\} .
\end{align}
Moreover, instead let $\calG$ be a class of uniformly bounded functions on $R^{1+r+d}$, again of VC-type with index $v=2V_{\calFn(\rbar+1)}$.  Then on an event $\calA_{3,s}$ with $\PP(\calA_{3,s})\ge 1-12 e^{-s}$, for a constant $c_5$,
\begin{align}
\textstyle \forall \theta\in\calG, \quad \textstyle \left|\frac{1}{n} \sum_{i=1}^n \left\{ \theta(Y_i,\bfF_i,\bX_i) - \EE[\theta(Y_i,\bfF_i,\bX_i)] \right\} \right| \le c_5 \left\{ \| \theta(Y,\bfF,\bX) \|_{L_2} + \nusn \right\} \left\{ \nusn + \sqrt{\frac{s}{n}} \right\} .
\label{eq:m_n_hat_ratio_rate_2_master}
\end{align}
\end{proposition}

\begin{proof}
We first prove \eqref{eq:m_n_hat_ratio_rate_1_master}.  Our main strategy is to apply the proof of Theorem~4 in \cite{GineKoltchinskiiJonWellner2003}.  For simplicity we assume that the uniform bound on the magnitudes of functions in $\calG$ is one, so we could simply choose the constant envelope $G=1$.

We let $s=s_n=1$, $l=l_n=\lceil\log(\nusn^{-1})\rceil$ (where ``$\lceil\cdot\rceil$'' is the ceiling operator that returns the smallest integer no smaller than the argument), $r=r_n=e^{-l_n}\in[e^{-1} \nusn, \nusn]$, and $\rho_j=r_n e^j$.  Define the function classes $\calG(a)=\{\theta: \theta\in\calG, \|\theta\|_{L_2}\le a\}$, $\calG(a,b]=\calG(b)\setminus\calG(a)$, and $\calG_e(a)=\calG(a/e,a]$.

Our general proof strategy based on the aforementioned Theorem~4 in \cite{GineKoltchinskiiJonWellner2003} is to control the stochastic behavior of the empirical process $\sum_{i=1}^n \wt\epsilon_i \theta(\wt\bfF_i,X_{i,j})$ within each ``shell'' $\theta\in\calG_e(\rho_j)$ for the shells numbered $j=1,\dots,l$.  \\\textit{\textbf{As our first step}}, we control the tail probability with respect to the expectation of the (supremum of the) empirical process within each shell.  Recall the sub-Gaussian norm $C_\epsilon<0$ of $\epsilon$ from Assumption~\ref{assum_fun_class}.  Also note that the bound $M_\infty$ on the full conditional expectation $m_0^*(\bfF,\bX)$ implies the same bound on any of its smoothed versions including $m_0(\bfF,X_j)=\EE[m_0^*(\bfF,\bX)|\bfF,X_j]$ and $\EE[Y|\wt\bfF,X_j]=\EE[m_0^*(\bfF,\bX)|\wt\bfF,X_j]$ due to the monotonicity property of conditional expectations.  {Thus, $\wt\epsilon-\epsilon=m_0(\bfF,X_j)-\EE[Y|\wt\bfF,X_j]$ is bounded, and so $\wt\epsilon$ is sub-Gaussian as well and we assume its sub-Gaussian norm is a constant $\wt C_\epsilon$ (with $\wt C_\epsilon<\infty$).}  By first Theorem~4 and then Eq.~(13) in \cite{Adamczak2008}, we know that
\begin{align*}
    &\textstyle \PP\left( \sup_{\theta\in\calG_e(\rho_j)} \left| \sum_{i=1}^n \wt\epsilon_i \theta(\wt\bfF_i,X_{i,j}) \right| \ge 2 \EE \sup_{\theta\in\calG_e(\rho_j)} \left| \sum_{i=1}^n \wt\epsilon_i \theta(\wt\bfF_i,X_{i,j}) \right| + s \right) \\
    &\le \exp\left( -\dfrac{s^2}{4 n \sup_{\theta\in\calG_e(\rho_j)} \|\wt\epsilon \theta(\wt\bfF,X_j)\|_{L_2}^2 } \right) + 3 \exp\left( -\dfrac{s^2}{C^2 \| \max_{i\in\{1,\dots,n\}]} \sup_{\theta\in\calG_e(\rho_j)} |\wt\epsilon_i \theta(\wt\bfF_i,X_{i,j})| \|_{\psi_2}^2 } \right) \\
    &\le \exp\left( -\dfrac{s^2}{4 n \|\wt\epsilon\|_{L_2}^2 \rho_j^2 } \right) + 3 \exp\left( -\dfrac{s^2}{C^2 \max_{i\in\{1,\dots,n\}} \| \sup_{\theta\in\calG_e(\rho_j)}\wt\epsilon_i \theta(\wt\bfF_i,X_{i,j}) \|_{\psi_2}^2 \log(n) } \right) \\
    &\le \exp\left( -\dfrac{s^2}{4 n \|\wt\epsilon\|_{L_2}^2  \rho_j^2 } \right) + 3 \exp\left( -\dfrac{s^2}{C^2 \wt C_\epsilon^2 \log(n) } \right).
\end{align*}

Then, absorbing $\wt C_\epsilon$ into $C$ and setting $s = \{2 \|\wt\epsilon\|_{L_2} \rho_j \sqrt{n} + C \log^{1/2}(n) \} \sqrt{u}$ yields
\begin{align}
    &\textstyle \PP\bigg( \sup_{\theta\in\calG_e(\rho_j)} \left| \frac{1}{n} \sum_{i=1}^n \wt\epsilon_i \theta(\wt\bfF_i,X_{i,j}) \right| \nonumber \\
    & \textstyle \le 2 \EE \sup_{\theta\in\calG_e(\rho_j)} \left| \frac{1}{n} \sum_{i=1}^n \wt\epsilon_i \theta(\wt\bfF_i,X_{i,j}) \right| + \frac{2}{\sqrt{n}} \|\wt\epsilon\|_{L_2} \rho_j \sqrt{u} + C \frac{\log^{1/2}(n)}{n} \sqrt{u} \bigg) \ge 1 - 4 e^{-u}.
    \label{eq:rel_dev_master}
\end{align}
		
Then, conforming to the notation in the proof of Theorem~4 in \cite{GineKoltchinskiiJonWellner2003}, define the events
\begin{align*}
    \calE_j^+(s) = \Big\{& \textstyle \sup_{\theta\in\calG_e(\rho_j)} \left| \frac{1}{n} \sum_{i=1}^n \wt\epsilon_i \theta(\wt\bfF_i,X_{i,j}) \right| \le 2 \EE \sup_{\theta\in\calG_e(\rho_j)} \left| \frac{1}{n} \sum_{i=1}^n \wt\epsilon_i \theta(\wt\bfF_i,X_{i,j}) \right| \\
    & \textstyle + \frac{2 \sqrt{s+2\log(l-j+1)} }{\sqrt{n}} \|\wt\epsilon\|_{L_2} \rho_j + C \frac{\log^{1/2}(n)}{n} \sqrt{s+2\log(l-j+1)} \Big\} .
\end{align*}
Then $\PP(\cap_{j=1}^l \calE_j^+(s)) \ge 1 - 8 e^{-s}$.  \\

\textit{\textbf{As our second step}}, with the tail probability on hand, we shall bound the expectation $\EE \sup_{\theta\in\calG_e(\rho_j)} \left| \sum_{i=1}^n \wt\epsilon_i \theta(\wt\bfF_i,X_{i,j}) \right|$ within each shell.
Further introduce the notation that $\omega_n(a)=n^{-1/2} \EE\sup_{\theta\in\calG_e(\rho_j)} | \sum_{i=1}^n \wt\epsilon_i \theta(\wt\bfF_i,X_{i,j}) |$ if $a\in(\rho_{j-1},\rho_j]$.  We now find a function $\omega:[0,1]\rightarrow\RR^+$ that satisfies the conditions in Theorem~4 in \cite{GineKoltchinskiiJonWellner2003}, that is,
\begin{enumerate*}[label=(\roman*)]
    \item $\omega$ is non-decreasing;
    \item $\omega(a)/a$ is non-increasing;
    \item $\omega_n(a)\le \omega(a)$, $\forall a\in[r_n,s_n]$;
    \item \label{eq:w_cond_iv_master} $\sup_{a\in(0,1]} a\sqrt{\log\log(1/a)}/\omega(a) \le K<\infty$.
\end{enumerate*}
We set out to find such a $\omega$.

Note that we can regard $\frac{1}{\sqrt{n}} \sum_{i=1}^n (\wt\epsilon_i/\wt C_\epsilon) \theta(\wt\bfF_i,X_{i,j})$, (indexed by) $\theta\in\calG_e(\rho_j)$ as the empirical process $\frac{1}{\sqrt{n}} \sum_{i=1}^n \theta(\wt\bfF_i,X_{i,j})$, $\theta\in\calG_e(\rho_j)$ symmetrized by the random variables $\wt\epsilon_i/\wt C_\epsilon$, $i\in\{1,\dots,n\}$.
{Although the $\wt\epsilon_i/\wt C_\epsilon$'s are not Rademacher random variables, by our earlier argument, when conditioning on the $\wt\bfF_i$'s and $X_{i,j}$'s, the process $\frac{1}{\sqrt{n}} \sum_{i=1}^n (\wt\epsilon_i/\wt C_\epsilon) \theta(\wt\bfF_i,X_{i,j})$ is nevertheless centered and sub-Gaussian with respect to the (random) distance $d$ on $\calG_e(\rho_j)$ defined as $d^2(\theta_1,\theta_2)=\frac{1}{n} \sum_{i=1}^n (\theta_1-\theta_2)^2(\wt\bfF_i,X_{i,j})=\|\theta_1-\theta_2\|_{L_2(\PPn)}^2$.}  Thus, the proof of Theorem~3.5.4 in \cite{GineNickl2016}, in particular the part that invokes Theorem~3.5.1, applies.  Specifically, let $Q$ be a generic measure on $(\wt\bfF,X_j)$, and recall that $G= 1$ is the envelope for the class $\calG_e(\rho_j)$ (so simply $\|G\|_{L_2(Q)}=1$, though at times we will retain the dependence on $\|G\|_{L_2(Q)}$ to conform to the notations in \cite{vVW1996} and \cite{GineNickl2016}), let $N(\calG_e(\rho_j), L_2(Q), \tau )$ be the covering number for the class $\calG_e(\rho_j)$ by balls of radius $\tau$ in the $L_2(Q)$ norm, and let the constants $A_1$, $A_2$ be as in Theorem~3.5.4 in \cite{GineNickl2016}. We then obtain from that theorem that, if $a\in(\rho_{j-1},\rho_j]$,
\begin{align*}
    & \textstyle \frac{1}{\wt C_\epsilon} \omega_n(a) = \EE \sup_{\theta\in\calG_e(\rho_j)} \left| \frac{1}{\sqrt{n}} \sum_{i=1}^n (\wt\epsilon_i/\wt C_\epsilon) \theta(\wt\bfF_i,X_{i,j}) \right| \\
    & \textstyle \le A_1 \int_0^{\rho_j} \sup_Q \sqrt{ \log\{ 2 N(\calG_e(\rho_j), L_2(Q), \tau \|G\|_{L_2(Q)} ) \} } \rmd\tau \\
    & \textstyle + A_2 \left\{ \int_0^{\rho_j} \sup_Q \sqrt{ \log\{ 2 N(\calG_e(\rho_j), L_2(Q), \tau \|G\|_{L_2(Q)} ) \} } \rmd\tau \right\}^2 / \sqrt{n} \rho_j^2 .
\end{align*}
By the covering number condition imposed in the proposition, we conclude that for a large enough constant $C_2$,
\begin{align*}
    \omega(u) = C_2 \left\{ V_{\calFn(\rbar+1)}^{1/2} \log^{1/2}(2/u) u + n^{-1/2} V_{\calFn(\rbar+1)} \log(2/u) \right\} .
\end{align*}
		
\textit{\textbf{Finally, we put everything together.}}  Now, on the event $\cap_{j=1}^l \calE_j^+(s)$, $\forall 1\le j \le l$, $\forall \theta\in \calG_e(\rho_j)$,
\begin{align*}
    & \textstyle \left| \frac{1}{\sqrt{n}} \sum_{i=1}^n \wt\epsilon_i \theta(\wt\bfF_i,X_{i,j}) \right| \le 2 \omega_n({\|\theta\|_{L_2}}) + 2 \sqrt{s+2\log(l-j+1)} \|\wt\epsilon\|_{L_2} \rho_j \\
& \textstyle \quad + C \frac{\log^{1/2}(n)}{\sqrt{n}} \sqrt{s+2\log(l-j+1)} \le 2 \omega(\rho_j) + (2 \|\wt\epsilon\|_{L_2} + C) \sqrt{s+2\log(l-j+1)} \rho_j .
\end{align*}

By our condition~\ref{eq:w_cond_iv_master} on $\omega$ and the derivations in the proof of Theorem~4 in \cite{GineKoltchinskiiJonWellner2003}, for some positive constant $C_3$,
\begin{align*}
    \textstyle \max_{1\le j\le l}\dfrac{ \rho_j \sqrt{\log(l-j+1)} }{ \omega(\rho_j) } \le C_3 .
\end{align*}
Hence, on the same event (that is, the event $\cap_{j=1}^l \calE_j^+(s)$), $\forall j$, $\forall \theta\in \calG_e(\rho_j)$,
\begin{align*}
    \left| \frac{1}{\sqrt{n}} \sum_{i=1}^n \wt\epsilon_i \theta(\wt\bfF_i,X_{i,j}) \right| &\le \{ 2 + (2 \|\wt\epsilon\|_{L_2} + C) \sqrt{2} C_3 \} \omega(\rho_j)  + (2 \|\wt\epsilon\|_{L_2} + C) \rho_j \sqrt{s} \\
    &\le C \left[ \left\{ V_{\calFn(\rbar+1)}^{1/2} \log^{1/2}(n)+ \sqrt{s} \right\} \|\theta\|_{L_2} + n^{-1/2} V_{\calFn(\rbar+1)} \log(n) \right] .
\end{align*}
For the class $\calG(r_n)$, we can just invoke \eqref{eq:rel_dev_master} but with $\calG_e(\rho_j)$ replaced by $\calG(r_n)$, to conclude that
\begin{align*}
    \textstyle \mathbb{P}\left[ \sup_{\theta\in\calG(r_n)} |\frac{1}{n} \sum_{i=1}^n \wt\epsilon_i \theta(\wt\bfF_i,X_{i,j})| \le C \left\{ \nusn + \sqrt{\frac{s}{n}} \right\} \nusn \right] \ge 1 - 4e^{-s} .
\end{align*}
Combining the two cases above, we obtain the part of the conclusion of the proposition regarding \eqref{eq:m_n_hat_ratio_rate_1_master}.  The other part of the proposition regarding \eqref{eq:m_n_hat_ratio_rate_2_master} follows by minor modification.  (Here the process on the left-hand side of \eqref{eq:m_n_hat_ratio_rate_2_master} is properly centered for each $\theta$, so we can proceed with the usual symmetrization by Rademacher random variables.)
\end{proof}

\subsection{The stochastic term in the proof of Theorem~\ref{thm:Riesz_est}}
\label{sec:thm:Riesz_est_stochastic}

We recall $\wt\alpha_n$ from \eqref{talpha} and that the $(\cdot)_j^s$ as the smoothed $j$-th partial derivative of the argument function as in \eqref{eq:smooth_master}.
Define, for $\alpha\in\calFn(\rbar+1)$, the function $\theta(\cdot;\alpha):\RR^{\rbar+1}\rightarrow\RR$ indexed by $\alpha$ as
\begin{align}
\label{eq:alpha_theta}
\theta(\wt\bff,x_j;\alpha) = (\alpha-\wt\alpha_n)_j^\us(\wt\bff,x_j)\ind\{x_j\in\calB_h\} - \{(\alpha-\wt\alpha_n)\alphanull\}(\wt\bff,x_j) .
\end{align}
Note that, under the null hypothesis or not, for each $\alpha$, the function $\theta(\cdot;\alpha)$ is always centered under $\PP$.  Thus, it is not necessary to center $\PPn \theta(\cdot;\alpha)$ at $\PP \theta(\cdot;\alpha)=0$ in \eqref{prop:alpha_hat_ratio_rate_2_master}.

We show the following result that will be useful for bounding the major stochastic term in the proof of Theorem~\ref{thm:Riesz_est}.
\begin{proposition}
\label{prop:g_n_hat_rate}
{Assume that Assumptions~\ref{ass:DGP} to \ref{ass:u_t_n} hold, and in addition the bandwidth $h$ satisfies
\begin{align}
    \label{eq:g_h_condition}
    h^2 V_{\calFn(\rbar+1)} \log(n) \ge 1 .
\end{align}
}
Then on an event with probability at least $1-12e^{-s}$, for a constant $c_6$,
\begin{align}
    \textstyle \forall \alpha\in\calFn(\rbar+1), \quad |\PPn \theta(\cdot;\alpha)| \le c_6 \left[ \left\{ \| \theta(\cdot;\alpha) \|_{L_2} + \nusn \right\} \left\{ \nusn + \sqrt{\frac{s}{n}} \right\} + \frac{1}{h} \nusn^2 \right] .
    \label{prop:alpha_hat_ratio_rate_2_master}
\end{align}
\end{proposition}

\begin{proof}
We mimic the proof of our Proposition~\ref{prop:m_n_hat_ratio_rate_master}.  (Alternatively, we could adapt the proof of Theorem~4 in \cite{GineKoltchinskiiJonWellner2003}, which will improve some multiplicative constant.)  Define the function class $\calGn = \{ \theta(\cdot;\alpha): \alpha\in\calFn(\rbar+1) \}$.  By the form of the function $\theta$, we can take $G = c_\calG / h$ for a constant $c_\calG$ large enough as the envelope for the class $\calGn$.
	
Similar to the proof of Proposition~\ref{prop:m_n_hat_ratio_rate_master} but with some minor modifications, we let $s=s_n=\|G\|_{L_2}=c_\calG / h$, $l=l_n=\lceil\log((c_\calG / h) \nusn^{-1})\rceil$, $r=r_n=s_n e^{-l_n}\in[e^{-1} \nusn, \nusn]$, and $\rho_j=r_n e^j$.  Define the function classes $\calG(a)=\{\theta\in\calGn, \|\theta\|_{L_2}\le a\}$, $\calG(a,b]=\calG(b)\setminus\calG(a)$, and $\calG_e(a)=\calG(a/e,a]$.
	
By Theorem~4 in \cite{Adamczak2008} with $\eta=1$ (which is in fact allowed), $\delta=1$ and $\alpha=2$, we know that
\begin{align*}
    &\textstyle \PP\left( \sup_{\theta\in\calG_e(\rho_j)} \left| n \PPn \theta \right| \ge 2 \EE \sup_{\theta\in\calG_e(\rho_j)} \left| n \PPn \theta \right| + s \right) \\
    &\le \exp\left( -\dfrac{s^2}{4 n \sup_{\theta\in\calG_e(\rho_j)} \|\theta(\wt\bfF,X_j)-\EE \theta(\wt\bfF,X_j)\|_{L_2}^2 } \right) \\
& \quad + 3 \exp\left( -\dfrac{s^2}{C^2 \| {\max_{i\in\{1,\dots,n\}}} \sup_{\theta\in\calG_e(\rho_j)} |\theta(\wt\bfF_i,X_{i,j})-\EE \theta(\wt\bfF,X_{i,j})| \|_{\psi_2}^2 } \right) \\
    &\le \exp\left( -\dfrac{s^2}{4 n \rho_j^2 } \right) + 3 \exp\left( -\dfrac{s^2}{C^2 c_\calG^2 / h^2 } \right) .
\end{align*}
Then, for $u>0$, setting $s = \{2 \rho_j \sqrt{n} + C c_\calG / h \} \sqrt{u}$ yields
\begin{align}
    \textstyle \PP\left( \sup_{\theta\in\calG_e(\rho_j)} \left| \PPn \theta \right| \le 2 \EE \sup_{\theta\in\calG_e(\rho_j)} \left| \PPn \theta \right| + \dfrac{2}{\sqrt{n}} \rho_j \sqrt{u} + C \dfrac{ c_\calG }{ h n } \sqrt{u} \right) \ge 1 - 4 e^{-u}.
    \label{eq:alpha_dev_master}
\end{align}
	
Then, conforming to the notation in the proof of Theorem~4 in \cite{GineKoltchinskiiJonWellner2003}, and essentially by setting $u=s+2\log(l-j+1)$ in the event in the probability above, we define the events
\begin{align*}
    \calE_j^+(s) = \Big\{& \textstyle \sup_{\theta\in\calG_e(\rho_j)} \left| \PPn \theta \right| \le 2 \EE \sup_{\theta\in\calG_e(\rho_j)} \left| \PPn \theta \right| \\
    & \textstyle + \frac{2 \sqrt{s+2\log(l-j+1)} }{\sqrt{n}} \rho_j + C \frac{ c_\calG }{ h n } \sqrt{s+2\log(l-j+1)} \Big\} .
\end{align*}
Then $\PP(\cap_{j=1}^l \calE_j^+(s)) \ge 1 - 8 e^{-s}$.  Further introduce the notation that $\omega_n(a)=n^{1/2} \EE\sup_{\theta\in\calG_e(\rho_j)} | (\PPn-P) \theta |$ if $a\in(\rho_{j-1},\rho_j]$.  We now find a function $\omega:[0,s_n]\rightarrow\RR^+$ that satisfies the conditions in Theorem~4 in \cite{GineKoltchinskiiJonWellner2003}, that is, \begin{enumerate*}[label=(\roman*)]
    \item $\omega$ is non-decreasing;
    \item $\omega(a)/a$ is non-increasing;
    \item $\omega_n(a)\le \omega(a)$, $\forall a\in[r_n,s_n]$;
    \item \label{eq:w_cond_iv_master_g} $\sup_{a\in(0,1]} a\sqrt{\log\log(1/a)}/\omega(a) \le K<\infty$.
\end{enumerate*}
We set out to find such a $\omega$.
	
Similar to the proof of Proposition~\ref{prop:m_n_hat_ratio_rate_master}, let $N(\calG_e(\rho_j), L_2(Q), \tau )$ be the covering number for the class $\calG_e(\rho_j)$ by balls of radius $\tau$ in the $L_2(Q)$ norm for a generic measure $Q$ on $(Y,\bfF,\bX)$, and let the constants $A_1$, $A_2$ be as in Theorem~3.5.4 in \cite{GineNickl2016}.  Then, by that theorem, if $a\in(\rho_{j-1},\rho_j]$,
\begin{align*}
    & \omega_n(a) \le \textstyle A_1 \dfrac{ c_\calG }{ h } \int_0^{\rho_j/\|G\|_{L_2(Q)}} \sup_Q \sqrt{ \log\{ 2 N(\calG_e(\rho_j), L_2(Q), \tau\|G\|_{L_2(Q)} ) \} } \rmd\tau \\
    & \textstyle \quad + A_2 \dfrac{ c_\calG }{ h } \left\{ \int_0^{\rho_j/\|G\|_{L_2(Q)}} \sup_Q \sqrt{ \log\{ 2 N(\calG_e(\rho_j), L_2(Q), \tau\|G\|_{L_2(Q)} ) \} } \rmd\tau \right\}^2 / \sqrt{n} \rho_j^2 .
\end{align*}
By Eq.~\eqref{eq:derivative_bound_via_original}, $N(\calGn, L_2(Q), \tau ) \le N(\calF_n(\rbar+1), L_2(Q), c_- h \tau )$ for a small enough constant $c_->0$ that is not dependent on $Q$.  In combination with Theorem~2.6.7 in \cite{vVW1996}, we can take, for a large enough constant $C_1$ and by adjusting $c_-$ if necessary, for $\calQ$ the set of all generic measures $Q$ on $(Y,\bfF,\bX)$,
\begin{align}
    \textstyle \sup_{Q\in\calQ} \log\{ 2 N(\calG_e(\rho_j), L_2(Q), \tau \|G\|_{L_2(Q)} ) \} \le C_1 V_{\calFn(\rbar+1)} \log(1/c_- h \tau) .
    \label{eq:alpha_covering_number}
\end{align}
Thus, for a large enough constant $C_2$, we can set
\begin{align*}
    \omega(u) = C_2 \left\{ V_{\calFn(\rbar+1)}^{1/2} \log^{1/2}\left(\dfrac{1}{h^2 u}\right) u + \frac{1}{n^{1/2} h} V_{\calFn(\rbar+1)} \log\left(\dfrac{1}{h^2 u}\right) \right\} .
\end{align*}

Now, on the event $\cap_{j=1}^l \calE_j^+(s)$, $\forall j$, $\forall \theta\in \calG_e(\rho_j)$, by condition~\eqref{eq:g_h_condition},
\begin{align*}
    & \sqrt{n} \left| \PPn \theta \right| \le 2 \omega_n(\|\theta\|_{L_2}) + 2 \sqrt{s+2\log(l-j+1)} \rho_j + C \dfrac{ c_\calG }{ h \sqrt{n}} \sqrt{s+2\log(l-j+1)} \\
    &\le 2 \omega(\rho_j) + (2 + C c_\calG) \sqrt{s+2\log(l-j+1)} \rho_j .
\end{align*}

{
From condition~\eqref{eq:g_h_condition}, we have
$1/(h \sqrt{n}) \lesssim r_n \lesssim \rho_j $.}
By our condition~\ref{eq:w_cond_iv_master} on $\omega$ and the derivations in the proof of Theorem~4 in \cite{GineKoltchinskiiJonWellner2003},
\begin{align*}
    \textstyle \max_{1\le j\le l}\dfrac{ \rho_j \sqrt{\log(l-j+1)} }{ \omega(\rho_j) } \le C_3 .
\end{align*}
Hence, on the same event (that is, the event $\cap_{j=1}^l \calE_j^+(s)$), $\forall j$, $\forall \theta\in \calG_e(\rho_j)$,
\begin{align*}
    \textstyle \left| \PPn \theta \right| & \textstyle \le \{ 2 + (2 + C c_\calG) \sqrt{2} C_3 \} \frac{1}{\sqrt{n}} \omega(\rho_j)  + (2 + C c_\calG) \rho_j \sqrt{\frac{s}{n}} \\
    & \textstyle \le C \left\{ \left( \nusn + \sqrt{\frac{s}{n}} \right) \|\theta\|_{L_2} + \frac{1}{h} \nusn^2 \right\} .
\end{align*}
In addition, for the class $\calG(r_n)$, we can just invoke \eqref{eq:alpha_dev_master} but with $\calG_e(\rho_j)$ replaced by $\calG(r_n)$, and simplify with \eqref{eq:g_h_condition}, to conclude that
\begin{align*}
    \textstyle \mathbb{P}\left[ \sup_{\theta\in\calG(r_n)} \sqrt{n} \left| \PPn \theta \right| \le C \left\{ \frac{1}{h} \nusn^2 + \nusn \sqrt{\frac{s}{n}} \right\} \right] \ge 1 - 4 e^{-s} .
\end{align*}
Combining the two cases above, we obtain the conclusion of the proposition.
\end{proof}

\section{Proofs of main theorems and additional results for Section~\ref{sec:regression_function}}
\label{app_sec:proof_sec_4}

\subsection{Proof of Theorem~\ref{thm:m_hat_est_master}}
\label{sec:proof_thm:m_hat_est_master}

\begin{proof}
We first introduce a shorthand notation that we will employ throughout: for a possibly random function $m=m(\bff,\bx)$,
\begin{align*}
    \textstyle \PP m = \int m(\bff,\bx) \rmd\PP = \int m(\bff,\bx) \rmd\PP.
\end{align*}
Since $\wt\bfF=\wt\bfF(\bX)$, for a possibly random function $g=g(\wt\bff,x_j)$, the analogous expression
\begin{align*}
    \textstyle \PP g = \int g(\wt\bff,x_j) \rmd\PP = \int g(\wt\bff,x_j) \rmd\PP
\end{align*}
makes sense as well.  We note that $\PP m$ and $\PP g$ could still be random due to the potential inherent randomness of $m$ and $g$.

\textit{We first prove \eqref{eq:rn1}}.  Recall $\wt g_n$ from Section~\ref{sec:approx_error}.
 To start, we have
\begin{align}
&\PP(\wh g_n - m_0 )^2 = \PP(\wh g_n - m_0)^2 - \PPn(\wh g_n - m_0)^2 \nonumber \\
& + \PPn(\wh g_n - m_0)^2 - \PPn\ell(\cdot;\wh g_n) + \PPn\ell(\cdot;\wh g_n) - \PPn\ell(\cdot;\wt g_n) + \PPn\ell(\cdot;\wt g_n) \nonumber \\
&\le \PP(\wh g_n - m_0)^2 - \PPn(\wh g_n - m_0)^2 + \left\{ \PPn(\wh g_n - m_0)^2 - \PPn\ell(\cdot;\wh g_n) \right\} + \PPn\ell(\cdot;\wt g_n),
\label{eq:m_n_hat_decomp_s}
\end{align}
where the inequality follows by the definition of $\wh g_n$  (in particular that it minimizes the loss $\PPn\ell(\cdot;g)$ within $g\in\calFn(\rbar+1)$).  Next, for the term in the curly bracket in the last step of \eqref{eq:m_n_hat_decomp_s},
\begin{align*}
& \PPn(\wh g_n - m_0)^2 - \PPn\ell(\cdot;\wh g_n) = \PPn(\wh g_n - m_0)^2 - \PPn(\wh g_n - m_0)^2 \\
& \textstyle \quad + \frac{2}{n} \sum_{i=1}^n \epsilon_i (\wh g_n(\wt\bfF_i,X_{i,j}) - m_0(\bfF_i,X_{i,j}))- \frac{1}{n} \sum_{i=1}^n \epsilon_i^2 \\
& \textstyle = \frac{2}{n} \sum_{i=1}^n \epsilon_i (\wh g_n(\wt\bfF_i,X_{i,j}) - m_0(\bfF_i,X_{i,j})) - \frac{1}{n} \sum_{i=1}^n \epsilon_i^2
\end{align*}
and, for the last term in the last step of \eqref{eq:m_n_hat_decomp_s},
\begin{align*}
\textstyle \PPn\ell(\cdot;\wt g_n) = \PPn(\wt g_n - m_0)^2 - \frac{2}{n} \sum_{i=1}^n \epsilon_i (\wt g_n - m_0) + \frac{1}{n} \sum_{i=1}^n \epsilon_i^2 .
\end{align*}
Plugging the two equalities above into \eqref{eq:m_n_hat_decomp_s},
\begin{align}
& \PP(\wh g_n - m_0)^2 \le \underbrace{\PP(\wh g_n - m_0)^2 - \PPn(\wh g_n - m_0)^2}_{I_{n,1}} \nonumber \\
&\textstyle \quad + \frac{2}{n} \sum_{i=1}^n \epsilon_i (\wh g_n(\wt\bfF_i,X_{i,j}) - m_0(\bfF_i,X_{i,j})) - \frac{2}{n} \sum_{i=1}^n \epsilon_i (\wt g_n - m_0) + \PPn(\wt g_n - m_0)^2 \nonumber \\
&\textstyle  = I_{n,1} + \frac{2}{n} \sum_{i=1}^n \epsilon_i (\wh g_n - \wt g_n)(\wt\bfF_i,X_{i,j}) + \PPn(\wt g_n - m_0)^2 .
\label{eq:m_n_hat_decomp_10_s}
\end{align}
We next analyze $I_{n,1}$.
Note that the class of functions $\{g-m_0=g(\wt\bff,x_j)-m_0(\bff,x_j): g\in\calFn(\rbar+1)\}$ satisfies the conditions of Lemma~\ref{lem:square_deviation} (in particular the covering number bound) by Theorem~2.6.7 in \cite{vVW1996}.  Thus by Lemma~\ref{lem:square_deviation}, on an event $\calA_{1,\zeta,s'}$ satisfying \eqref{eq:square_deviation_prob}, we have
\begin{align}
\label{eq:square_bound_tail_prob}
\textstyle \forall g\in\calFn(\rbar+1), \quad | \PP(g - m_0)^2 - \PPn(g - m_0)^2 | \le \zeta \left\{ \PP(g - m_0)^2 + \frac{s'}{n} \right\} .
\end{align}
Applying this result with $g=\wh g_n$ on \eqref{eq:m_n_hat_decomp_10_s}, on the event $\calA_{1,\zeta=\frac{1}{2},s'}$ we obtain
\begin{align}
&\textstyle \PP(\wh g_n - m_0)^2 \le \frac{s'}{n} + 4 \underbrace{ \textstyle \frac{1}{n} \sum_{i=1}^n \epsilon_i (\wh g_n - \wt g_n)(\wt\bfF_i,X_{i,j}) }_{I_{n,2}} + 2 \PPn(\wt g_n - m_0)^2 .
\label{eq:m_n_hat_decomp_20_s_master}
\end{align}
Next, we investigate the rate of the stochastic term $I_{n,2}$ on the right-hand side of \eqref{eq:m_n_hat_decomp_20_s_master}.  Recalling the $\wt\epsilon_i$'s introduced in Section~\ref{sec:thm:m_hat_est_master_stochastic}, we further decompose
\begin{align*}
\textstyle I_{n,2} = \frac{1}{n} \sum_{i=1}^n \wt\epsilon_i (\wh g_n - \wt g_n)(\wt\bfF_i,X_{i,j}) + \frac{1}{n} \sum_{i=1}^n (\epsilon_i-\wt\epsilon_i) (\wh g_n - \wt g_n)(\wt\bfF_i,X_{i,j}) \equiv I_{n,2,1} + I_{n,2,2} .
\end{align*}
The reason we need this decomposition is that in $I_{n,2,1}$, the $\wt\epsilon_i$'s are conditionally properly centered, that is $\EE[\wt\epsilon_i|\wt\bfF_i,X_{i,j}]=0$, while the term $I_{n,2,2}$ will utilize the ``stability of conditional expectation'' condition in \eqref{eq:cond_exp_stability} to treat the difference $\epsilon_i-\wt\epsilon_i$.

Now, for the term $I_{n,2,1}$, note that the class of functions $\{(g-\wt g_n)=(g-\wt g_n)(\wt\bff,x_j): g\in\calFn(\rbar+1)\}$ satisfies the covering number bound in Proposition~\ref{prop:m_n_hat_ratio_rate_master} for \eqref{eq:m_n_hat_ratio_rate_1_master}; we then apply \eqref{eq:m_n_hat_ratio_rate_1_master} in Proposition~\ref{prop:m_n_hat_ratio_rate_master} with $\theta=\wh g_n - \wt g_n$ to conclude that, on an event $\calA_{2,s}$ satisfying $\PP(\calA_{2,s})\ge 1 - 12 e^{-s}$,
\begin{align}\label{error}
    \textstyle |I_{n,2,1}| = |\frac{1}{n} \sum_{i=1}^n \wt\epsilon_i (\wh g_n - \wt g_n)(\wt\bfF_i,X_{i,j})| \le c_4 \left\{ \| \wh g_n - \wt g_n \|_{L_2} + \nusn\right\} \left\{ \nusn + \sqrt{\frac{s}{n}} \right\} .
\end{align}
For the term $I_{n,2,2}$, first note that $|\epsilon_i-\wt\epsilon_i|=|\{Y_i-m_0(\bfF_i,X_{i,j})\} - \{Y_i-\EE[Y_i|\wt\bfF_i,X_{i,j}]\}| = |m_0(\bfF_i,X_{i,j})-\EE[Y|\wt\bfF,X_j]|$, which is uniformly bounded by $2M_\infty$ by the comments in the proof or Proposition~\ref{prop:m_n_hat_ratio_rate_master}.  As already mentioned, the class of functions $\{(g-\wt g_n)=(g-\wt g_n)(\wt\bff,x_j): g\in\calFn(\rbar+1)\}$ satisfies the covering number bound in Proposition~\ref{prop:m_n_hat_ratio_rate_master} for \eqref{eq:m_n_hat_ratio_rate_2_master}; multiply this class by $\epsilon-\wt\epsilon=\EE[Y|\wt\bfF,X_j]-m_0(\bfF_i,X_{i,j})$ considered as a single, uniformly bounded function changes at most the constant in the covering number bound (see Theorem~3 in \cite{Andrews1994Handbook}).  Hence, we can apply \eqref{eq:m_n_hat_ratio_rate_2_master} in Proposition~\ref{prop:m_n_hat_ratio_rate_master} to conclude that, on an event $\calA_{3,s}$ satisfying $\PP(\calA_{3,s})\ge 1 - 12 e^{-s}$,
\begin{align*}
    \textstyle |I_{n,2,2} - \int I_{n,2,2} \rmd\PP | & \textstyle \le c_4 \left\{ \| (\epsilon-\wt\epsilon)(\wh g_n - \wt g_n)(\wt\bfF,X_j)\|_{L_2} + \nusn\right\} \left\{ \nusn + \sqrt{\frac{s}{n}} \right\} \\
    & \textstyle \lesssim \left\{ \| \wh g_n - \wt g_n\|_{L_2} + \nusn\right\} \left\{ \nusn + \sqrt{\frac{s}{n}} \right\}
\end{align*}
where in the transition to the last line we have again applied the boundedness of $\epsilon-\wt\epsilon$.
For the centering term $\int I_{n,2,2} \rmd\PP$, by the Cauchy-Schwarz inequality,
\begin{align}
    \textstyle |\int I_{n,2,2} \rmd\PP| \le \left\{ \EE[(\epsilon_i-\wt\epsilon_i)^2] \right\}^{1/2} \left\{ \PP(\wh g_n - \wt g_n)^2 \right\}^{1/2} \lesssim \deltaf \| \wh g_n - \wt g_n\|_{L_2} .
    \label{eq:I_{n,2,2}_bound}
\end{align}
Collecting the rates for $I_{n,2,1}$ and $I_{n,2,2}$, we conclude that on $\calA_{2,s}\cap \calA_{3,s}$,
\begin{align*}
    \textstyle |I_{n,2}| = |I_{n,2,1} + I_{n,2,2}| \lesssim \left\{ \| \wh g_n - \wt g_n\|_{L_2} + \nusn\right\} \left\{ \nusn + \sqrt{\frac{s}{n}} \right\} + \deltaf \| \wh g_n - \wt g_n\|_{L_2} .
\end{align*}
Then, applying this result in \eqref{eq:m_n_hat_decomp_20_s_master}, we obtain that on the event $\calA_{1,\zeta=\frac{1}{2},s'}\cap\calA_{2,s}\cap\calA_{3,s}$ we have (note the change in notation from $\PP(\wh g_n - m_0)^2$ to $\|\wh g_n - m_0\|_{L_2}^2$ below)
\begin{align}
& \textstyle \|\wh g_n - m_0\|_{L_2}^2 \textstyle \le \frac{s'}{n} + C \left\{ \| \wh g_n - \wt g_n\|_{L_2} + \nusn\right\} \left\{ \nusn + \sqrt{\frac{s}{n}} \right\} + C \deltaf \| \wh g_n - \wt g_n\|_{L_2} + 2 \PPn(\wt g_n - m_0)^2 \nonumber \\
& \textstyle \lesssim \| \wh g_n - \wt g_n\|_{L_2} \{ \nusn + \deltaf + \sqrt{\frac{s}{n}} \} + \nusn \{ \nusn + \sqrt{\frac{s}{n}} \} + \frac{s'}{n} + 2 \PPn(\wt g_n - m_0)^2 \nonumber \\
& \textstyle \lesssim \{ \| \wh g_n - m_0 \|_{L_2} + \| \wt g_n - m_0 \|_{L_2} \} \{ \nusn + \deltaf + \sqrt{\frac{s}{n}} \} + \nusn \{ \nusn + \sqrt{\frac{s}{n}} \} + \frac{s'}{n} + 2 \PPn(\wt g_n - m_0)^2 \nonumber \\
& \textstyle \lesssim \| \wh g_n - m_0 \|_{L_2} \{ \nusn + \deltaf + \sqrt{\frac{s}{n}} \} + \{\nuan+\deltaf\} \{ \nusn + \deltaf + \sqrt{\frac{s}{n}} \} + \nusn \{ \nusn + \sqrt{\frac{s}{n}} \} + \frac{s'}{n} \nonumber \\
& \textstyle \quad + 2 \underbrace{\PPn(\wt g_n - m_0)^2}_{I_{n,3}}
\label{eq:m_n_hat_decomp_30_s_master}
\end{align}
where in the transition to the last step we have invoked Lemma~\ref{lemma:approx_error}.
For the term $I_{n,3}$, by \eqref{eq:square_bound_tail_prob} and the remark above it, and then Lemma~\ref{lemma:approx_error} again, on the event $\calA_{1,\zeta=\frac{1}{2},s'}$ we have
\begin{align*}
\textstyle I_{n,3} = \PPn(\wt g_n - m_0)^2 \le \frac{3}{2} \PP(\wt g_n - m_0)^2 + \frac{1}{2} \frac{s'}{n} \lesssim \nuan^2+\deltaf^2+\frac{s'}{n}.
\end{align*}
Then, invoking the above in \eqref{eq:m_n_hat_decomp_30_s_master} and then using $ab \le 2^{-1}(a^2+b^2)$ several times, we obtain that on the event $\calA_{1,\zeta=\frac{1}{2},s'}\cap\calA_{2,s}\cap\calA_{3,s}$,
\begin{align*}
& \textstyle \|\wh g_n - m_0\|_{L_2}^2 \textstyle \lesssim \| \wh g_n - m_0 \|_{L_2} \{ \nusn + \deltaf + \sqrt{\frac{s}{n}} \} \\
& \textstyle \quad + \{\nuan+\deltaf\} \{ \nusn + \deltaf + \sqrt{\frac{s}{n}} \} + \nusn \{ \nusn + \sqrt{\frac{s}{n}} \} + \nuan^2 + \deltaf^2 + \frac{s'}{n} \\
& \textstyle \quad \lesssim \| \wh g_n - m_0 \|_{L_2} \{ \nusn + \deltaf + \sqrt{\frac{s}{n}} \} + \{ \nusn^2 + \nuan^2 + \deltaf^2 + \frac{s}{n} + \frac{s'}{n} \} .
\end{align*}
Finally, solving for $\|\wh g_n - m_0\|_{L_2}$ yields that, on $\calA_{1,\zeta=\frac{1}{2},s'}\cap\calA_{2,s}\cap\calA_{3,s}$,
\begin{align}
\label{eq:rn1_generic}
\|\wh g_n - m_0\|_{L_2} \textstyle \lesssim \nusn + \nuan + \deltaf + \sqrt{\frac{s}{n}} + \sqrt{\frac{s'}{n}} .
\end{align}

\medskip

\textit{Next, we impose Assumption~\ref{ass:NN_scaling} and derive the concrete resultant rates for $\nusn$ in \eqref{eq:nu_n_generic}} (which then becomes the rate $\nun$ in \eqref{eq:nu_n_concrete}) and $\nuan$.  For the former, namely $\nusn$, essentially we need to obtain an upper bound for the VC-index $V_{\calFn(\rbar+1)}$ of the class $\calFn(\rbar+1)$; see our Def.~\ref{def:VC_subgraph}.  Note that there are a couple of closely-related though not always identical concepts regarding the VC-index.  Our definition turns out to be identical to the \textit{pseudo-dimension} defined in, for instance, Def.~11.2 in \cite{AnthonyBartlett1999} and Def.~2 in \cite{BartlettHarveyLiawMehrabian2019} for \textit{real}-valued functions.  A closely related concept, termed $\textup{VCdim}(\calG)$ for a class $\calG$ of $\{0,1\}$-valued functions, also appears; for instance, see Theorem~14.1 in \cite{AnthonyBartlett1999}.  Fortunately, we have the following simple connection: let $\calFn^\dagger$ be our class of function $\calFn(\rbar+1)$ augmented with one extra input unit and one extra computation unit, and let $\calH^\dagger$ be the resulting set of $\{0,1\}$-valued functions computed by $\calFn^\dagger$, all as described in Theorem~14.1 in \cite{AnthonyBartlett1999}.  Then the VC-index of $\calFn(\rbar+1)$ will be bounded by $\textup{VCdim}(\calH^\dagger)$, that is $V_{\calFn(\rbar+1)}\le\textup{VCdim}(\calH^\dagger)$, and the right-hand side is in turn bounded in \cite{AnthonyBartlett1999}.  Specifically, let $U$ be the total number of weights for a function in $\calFn(\rbar+1)$, so following the derivation on p.~2 in appendix~B of \cite{KohlerLanger2021}, $U\sim L k_{\textup{0}}^2$ (note that the input dimension $\rbar+1$ does not appear because by assumption $\rbar$ is a constant).  Then, by Theorem~7 in \cite{BartlettHarveyLiawMehrabian2019} and then Assumption~\ref{ass:NN_scaling}, $V_{\calFn(\rbar+1)}\le\textup{VCdim}(\calH^\dagger)\lesssim U L {\log(U)} \sim n^{\frac{1}{2\kappa+1}} \log^{2\frac{4\kappa-1}{2\kappa+1}+1}(n)$.  To fix the multiplicative factor in this bound on $V_{\calFn(\rbar+1)}$, we shall explicitly set
\begin{align}
\label{eq:V_calF_bound}
    V_{\calFn(\rbar+1)} \le \cp n^{\frac{1}{2\kappa+1}} \log^{2\frac{4\kappa-1}{2\kappa+1}+1}(n) = \pn \log^{-1}(n)
\end{align}
for a large enough constant $\cp$ and for $\pn$ from \eqref{eq:nu_n_concrete}.

Now we are ready to fix $s$ and $s'$ in \eqref{eq:rn1_generic}.  We set $s=\pn$ so $\sqrt{s/n}$ precisely becomes the concrete rate $\nun$ in \eqref{eq:nu_n_concrete} (up to a factor involving $\cp$), and $\nun$ in turn bounds $\nusn$ from \eqref{eq:nu_n_generic} if we apply in \eqref{eq:nu_n_generic} the bound \eqref{eq:V_calF_bound}.  Similarly, we choose $s'$ to be a large enough multiple of $\pn$.  Together with \eqref{eq:rn1_generic}, we then conclude that, on the intersection $\calA' \equiv \calA_{1,\zeta=\frac{1}{2},s'}\cap\calA_{2,s}\cap\calA_{3,s}$ whose probability is at least $1 - C e^{-\pn}$,
\begin{align}
\label{eq:rn1_generic_2}
\|\wh g_n - m_0\|_{L_2} \textstyle \lesssim \nun + \nuan + \deltaf .
\end{align}

Finally, we derive in \eqref{eq:rn1_generic_2} the rate of $\nuan$ in order to reach \eqref{eq:rn1}.  This rate is implicitly implied in Corollary~1 in \cite{FanGu2023factor}, and depends on the earlier Proposition~3.3 in \cite{fan2022noise}, but we give a brief derivation here.  Define $\wt L$ and $\wt k_{\textup{0}}$ such that $\wt L \log(\wt L)=L$ and $\wt k_{\textup{0}} \log(\wt k_{\textup{0}})=k_{\textup{0}}$; we introduce these quantities because the neural network architecture scaling in Proposition~3.3 in \cite{fan2022noise} is specified by such quantities.  Under Assumption~\ref{ass:NN_scaling}, it's easy to compute that $\wt L \wt k_{\textup{0}} \sim n^{\frac{1}{4\kappa+2}} \log^{\frac{4\kappa-1}{2\kappa+1}-2}(n)$,
which further yields, by Proposition~3.3 in \cite{fan2022noise},
\begin{align}
\label{eq:nuan_rate_concrete}
\nuan \lesssim (\wt L \wt k_{\textup{0}})^{-2\kappa} \sim \nun .
\end{align}
Finally, applying the bound above in \eqref{eq:rn1_generic_2} yields that \eqref{eq:rn1} holds on the intersection $\calA_{1,\zeta=\frac{1}{2},s'}\cap\calA_{2,s}\cap\calA_{3,s}$.

\bigskip

Finally, given \eqref{eq:rn1}, a straightforward application of Lemma~\ref{lemma:derivative_bound_via_original} with $\wh g=\wh g_n$ and $m=m_0$ will yield \eqref{eq:rn1_deriv}.  This completes the proof of the theorem.
\end{proof}

\subsection{Proof of Theorem~\ref{thm:Riesz_est}}
\label{sec:proof_thm:Riesz_est}

\begin{proof}
Introduce the shorthand notation $\PP_{n|\calI_h}$ for the empirical measure $\PPn$ but only including the samples with sample indices in $\calI_h$.  Then, the loss function $\Rnullhat(\cdot)$ introduced in Section~\ref{sec:Riesz_est} can be written as
\begin{align*}
    \textstyle \Rnullhat(\alpha) = \PPn \alpha^2 - 2 \PP_{n|\calI_h} \alpha_j^\us .
\end{align*}
Also recall $\wt\alpha_n$ defined in \eqref{talpha}.  Then, we have
\begin{align}
    &\PP(\wh\alpha_n - \alphanull)^2 = \PP(\wh\alpha_n - \alphanull)^2 - \PPn(\wh\alpha_n - \alphanull)^2 + \PPn(\wh\alpha_n - \alphanull)^2 \nonumber \\
    & \quad - \Rnullhat(\wh\alpha_n) + \Rnullhat(\wh\alpha_n) - \Rnullhat(\wt\alpha_n) + \Rnullhat(\wt\alpha_n) - \PPn(\wt\alpha_n - \alphanull)^2 + \PPn(\wt\alpha_n - \alphanull)^2 \nonumber \\
    &\le \PP(\wh\alpha_n - \alphanull)^2 - \PPn(\wh\alpha_n - \alphanull)^2 + \underbrace{\PPn(\wh\alpha_n - \alphanull)^2 - \Rnullhat(\wh\alpha_n)}_{I_{n,2}} \nonumber \\
    & \quad +\underbrace{ \Rnullhat(\wt\alpha_n) - \PPn(\wt\alpha_n - \alphanull)^2}_{I_{n,3}} + \PPn(\wt\alpha_n - \alphanull)^2,
    \label{eq:alpha_n_hat_decomp}
\end{align}
where the inequality is due to $\wh\alpha_n$ minimizing the loss $\Rnullhat(\cdot)$.
Now, for handling the terms $I_{n,2}$ and $I_{n,3}$ in the above, note that for any $\alpha\in\calFn$,
\begin{align*}
	\PPn(\alpha - \alphanull)^2 - \Rnullhat(\alpha) &= \PPn \alpha^2 - 2 \PPn(\alpha \alphanull) + \PPn \alpha_n^{ 2\textup{null}} + 2 \PP_{n|\calI_h} \alpha_j^\us - \PPn \alpha^2 \\
	&= - 2 \PPn(\alpha \alphanull) + \PPn \alpha_n^{ 2\textup{null}} + 2 \PP_{n|\calI_h} \alpha_j^\us .
\end{align*}
Then, together using the function $\theta(\cdot;\alpha)$ introduced in \eqref{eq:alpha_theta}, the sum of the terms $I_{n,2}$ and $I_{n,3}$ becomes
\begin{align}
    & \PPn(\wh\alpha_n - \alphanull)^2 - \Rnullhat(\wh\alpha_n) + \Rnullhat(\wt\alpha_n) - \PPn(\wt\alpha_n - \alphanull)^2 \nonumber  \\
    &= - 2 \PPn(\wh\alpha_n \alphanull) + 2 \PPn(\wt\alpha_n \alphanull) + 2 \PP_{n|\calI_h} \wh\alpha_{n,j}^\us - 2 \PP_{n|\calI_h} \wt\alpha_{n,j}^\us = 2 \PPn \theta(\cdot;\wh\alpha_n) .
\label{eq:alpha_n_hat_decomp_10}
\end{align}

Plugging \eqref{eq:alpha_n_hat_decomp_10} into \eqref{eq:alpha_n_hat_decomp}, we obtain
\begin{align}
    \PP(\wh\alpha_n - \alphanull)^2 &\le \underbrace{ \PP(\wh\alpha_n - \alphanull)^2 - \PPn(\wh\alpha_n - \alphanull)^2 }_{I_{n,1}} + 2 \PPn \theta(\cdot;\wh\alpha_n) + \PPn(\wt\alpha_n - \alphanull)^2 .
    \label{eq:alpha_n_hat_decomp_20}
\end{align}
For the term $I_{n,1}$ above, first note that the class of functions $\{\alpha-\alphanull: \alpha\in\calFn(\rbar+1)\}$ satisfies the conditions of Lemma~\ref{lem:square_deviation}.  Then, similar to \eqref{eq:square_bound_tail_prob}, by Lemma~\ref{lem:square_deviation}, on an event $\calA_{1,\zeta,s'}$ satisfying \eqref{eq:square_deviation_prob}, we have
\begin{align}
    \textstyle \forall \alpha\in\calFn(\rbar+1), \quad | \PP(\alpha - \alphanull)^2 - \PPn(\alpha - \alphanull)^2 | \le \zeta \left\{ \PP(\alpha - \alphanull)^2 + \frac{s'}{n} \right\} .
    \label{eq:alpha_square_dev}
\end{align}
Applying this result with $\alpha=\wh\alpha_n$ for $I_{n,1}$ on \eqref{eq:alpha_n_hat_decomp_20}, on the event $\calA_{1,\zeta=\frac{1}{2},s'}$ we obtain
\begin{align*}
    \textstyle \PP(\wh\alpha_n - \alphanull)^2 & \le \textstyle \frac{1}{2} \PP(\wh\alpha_n - \alphanull)^2 + \frac{1}{2} \frac{s'}{n} + 2 \PPn \theta(\cdot;\wh\alpha_n) + \PPn(\wt\alpha_n - \alphanull)^2
\end{align*}
or equivalently, rearranging and centering we have,
\begin{align}
    & \textstyle \PP(\wh\alpha_n - \alphanull)^2 \le \frac{s'}{n} + 4 \PPn \theta(\cdot;\wh\alpha_n) + 2 \PPn(\wt\alpha_n - \alphanull)^2 \nonumber \\
    & \textstyle = \frac{s'}{n} + 4 \underbrace{\PPn \theta(\cdot;\wh\alpha_n)}_{I_{n,4}}
    + 2 \underbrace{\PPn(\wt\alpha_n - \alphanull)^2}_{I_{n,5}} .
     \label{eq:alpha_n_hat_decomp_30}
\end{align}

Next, we consider the term $I_{n,4}$ in \eqref{eq:alpha_n_hat_decomp_30}.  First, we note that
\begin{align*}
    \| \theta(\cdot;\alpha) \|_{L_2}^2 &= \EE \left[ \left\{ (\alpha-\wt\alpha_n)_j^\us(\wt\bfF,X_j) \ind\{X_j\in\calB_h\} - (\alpha-\wt\alpha_n)\alphanull\}(\wt\bfF,X_j) \right\}^2 \right] \\
    &\le 2 \EE\left[ (\alpha-\wt\alpha_n)_j^\us(\wt\bfF,X_j)^2 \ind\{X_j\in\calB_h\} + \{ (\alpha-\wt\alpha_n)\alphanull\}^2(\wt\bfF,X_j) \right] \lesssim h^{-2} \| \alpha-\wt\alpha_n \|_{L_2}^2
\end{align*}
where the last inequality follows from \eqref{eq:derivative_bound_via_original}.  Then, applying Proposition~\ref{prop:g_n_hat_rate} (whose condition \eqref{eq:g_h_condition} is satisfied due to Assumptions~\ref{ass:NN_scaling}, \ref{ass:alpha_h} and the paragraph below \eqref{eq:rn1_generic}) and the result above on the term $(\PPn-\PP) \theta(\cdot;\wh\alpha_n)$ in \eqref{eq:alpha_n_hat_decomp_30}, we obtain that, on an event $\calA_{2,s}$ satisfying $\PP(\calA_{2,s})\ge 1 - 12 e^{-s}$,
\begin{align*}
    \textstyle |I_{n,4}| & \textstyle \le c_4 \left\{ \| \theta(\cdot;\wh\alpha_n) \|_{L_2} + \nusn\right\} \left\{ \nusn + \sqrt{\frac{s}{n}} \right\} \lesssim \left\{ \frac{1}{h} \| \wh\alpha_n - \alphanull \|_{L_2} + \nusn\right\} \left\{ \nusn + \sqrt{\frac{s}{n}} \right\} + \frac{1}{h} \nusn^2 .
\end{align*}

Next, we consider the term $I_{n,5}$ in \eqref{eq:alpha_n_hat_decomp_30}, which represents a population bias induced by approximating $\alphanull$ through the neural network function $\wt\alpha_n$.  By \eqref{eq:alpha_square_dev} and the remarks above, on the event $\calA_{1,\zeta=\frac{1}{2},s'}$ we have
\begin{align*}
    \textstyle |I_{n,5}| = | \PPn(\wt\alpha_n - \alphanull)^2 | \le \frac{3}{2} \PP(\wt\alpha_n - \alphanull)^2 + \frac{1}{2}\frac{s'}{n} \lesssim \nuan^2 + \deltaf^2 + \frac{s'}{n}
\end{align*}
where the last step follows by Lemma~\ref{lemma:approx_error_alpha}.  Then, collecting the rates for $I_{n,4}$ and $I_{n,5}$ in \eqref{eq:alpha_n_hat_decomp_30}, we obtain, on the intersection $\calA_{1,\zeta=\frac{1}{2},s'}\cap\calA_{2,s}$,
\begin{align*}
    \textstyle \| \wh\alpha_n - \alphanull\|_{L_2}^2 & \textstyle \lesssim h^{-1} ( \nusn + \sqrt{\frac{s}{n}} ) \| \wh\alpha_n - \alphanull\|_{L_2} + \nusn \left\{ \nusn + \sqrt{\frac{s}{n}} \right\} + \frac{1}{h} \nusn^2 + \nuan^2 + \deltaf^2 + \frac{s'}{n}.
\end{align*}
Then, solving for $\| \wh\alpha_n - \alphanull\|_{L_2}$ yields that
\begin{align*}
& \textstyle \| \wh\alpha_n - \alphanull\|_{L_2} \lesssim  h^{-1} ( \nusn + \sqrt{\frac{s}{n}} ) + \nuan+\deltaf + \sqrt{\frac{s'}{n}} .
\end{align*}
This is our counterpart to \eqref{eq:rn1_generic} in the current proof of Theorem~\ref{thm:Riesz_est}.  Then, we apply Assumption~\ref{ass:NN_scaling} as in the paragraph below \eqref{eq:rn1_generic}, and subsequently choose $s$ and $s'$, and bound $\nusn$, as in the paragraph below \eqref{eq:V_calF_bound}, and bound $\nuan$ as in \eqref{eq:nuan_rate_concrete}.  This will yield the rate, holding on an event with probability at least $1 - C e^{-\pn}$, on the left-hand side of \eqref{eq:r_n2_null} except for the cap at $M$.  The additional cap is straightforward due the boundedness assumption on $\alphanull$.  This completes the proof of Theorem~\ref{thm:Riesz_est} regarding \eqref{eq:r_n2_null}.

\medskip

Next, we bound $\|\alphanull-\alpha_{t,n}^*\|_{L_2}$ which will subsequently lead to \eqref{eq:r_n2}.  Recall the population loss $\Rnull(\alpha)= \EE \alpha^2(\wt\bfF,X_j) - 2 \int_{\Omega_h} \alpha_j^\us(\wt\bff,x_j) \rmd\PP$ (i.e., the first two terms in the last line of \eqref{eq:Rnull}) introduced in the proof of Lemma~\ref{lem:Riesz_analytic}.  We have, for all $t\in\RR$,
\begin{align}
    &\PP(\alphanull-\alpha_{t,n}^*)^2 = \PP(\alphanull-\alpha_{t,n}^*)^2 - \Rnull(\alphanull) + \Rnull(\alphanull) - \Rnull(\alpha_{t,n}^*) + \Rnull(\alpha_{t,n}^*) \nonumber \\
    & \le \PP(\alphanull-\alpha_{t,n}^*)^2 - \Rnull(\alphanull) + \Rnull(\alpha_{t,n}^*) \nonumber \\
    & = \|\alphanull\|_{L_2}^2 - 2 \langle \alphanull, \alpha_{t,n}^*\rangle + \|\alpha_{t,n}^*\|_{L_2}^2 - \|\alphanull\|_{L_2}^2 + 2 \langle \alphanull, \alphanull\rangle + \|\alpha_{t,n}^*\|_{L_2}^2 - 2 \langle \alpha_{t,n}^*, \alphanull\rangle \nonumber \\
    & = - 2 \langle \alphanull-\alpha_{t,n}^*, \alpha_{t,n}^*\rangle + 2 \langle \alphanull-\alpha_{t,n}^*, \alphanull\rangle \nonumber \\
    & \textstyle = - 2 \int_{\Omega_h} e^{t \,g_{0,j}^\us(\wt\bff,x_j) } (\alphanull-\alpha_{t,n}^*)_j^\us(\wt\bff,x_j) \rmd\PP + 2 \int_{\Omega_h} (\alphanull-\alpha_{t,n}^*)_j^\us(\wt\bff,x_j) \rmd\PP \nonumber \\
    & \textstyle = - 2 \int_{\Omega_h} [ e^{t \,g_{0,j}^\us(\wt\bff,x_j) } - 1 ] (\alphanull-\alpha_{t,n}^*)_j^\us(\wt\bff,x_j) \rmd\PP \nonumber \\
    & \textstyle \le 2 \| Z_{t,j,h} \|_{L_2} \| (\alphanull-\alpha_{t,n}^*)_j^\us \|_{L_2} \le \frac{2c^\us}{h} \| Z_{t,j,h} \|_{L_2} \| \alphanull-\alpha_{t,n}^* \|_{L_2} ,
    \label{eq:alpha_null_bias}
\end{align}
where the inequality is due to $\alphanull$ minimizing the loss $\Rnull(\cdot)$ (see the discussion in the proof of Lemma~\ref{lem:Riesz_analytic}), the transition to the third to last line follows by the property of $\alpha_{t,n}^*$ as the Riesz representer in \eqref{eq:Riesz_H0_master} and the corresponding property of $\alphanull$, the transition to the last line follows by the Cauchy-Schwarz inequality, and the last step follows by Lemma~\ref{lemma:derivative_bound_via_original_specific}.  Then, combining \eqref{eq:alpha_null_bias} with the boundedness assumption on $\alphanull$ and $\alpha_{t,n}^*$ yields \eqref{eq:r_alpha}.

Finally, applying the triangle inequality on \eqref{eq:r_n2_null} and \eqref{eq:r_alpha} yields the remaining conclusion, namely \eqref{eq:r_n2}, of Theorem~\ref{thm:Riesz_est}.
\end{proof}

\subsection{Proof of Proposition~\ref{prop:tweaking}}
\label{sec:proof_prop:tweaking}
\begin{proof}
\textit{We first prove the first half of the proposition on $\deltathat$.}  By definition,
\begin{align*}
    \textstyle \deltathat = \argmin_{\deltat\in\RR} \PPn \ell(\cdot;\wh g_n + \deltat \wh\alpha_n) .
\end{align*}
Denote the function $\wh\delta_n=\wh g_n - m_0$.  Then
\begin{align*}
& \textstyle \PPn \ell(\cdot;\wh g_n + \deltat \wh\alpha_n) = \dfrac{1}{n} \sum_{i=1}^n \left\{ Y_i - \wh g_n(\wt\bfF_i,X_{i,j}) - \deltat \wh\alpha_n(\wt\bfF_i,X_{i,j}) \right\}^2 \\
&= \dfrac{1}{n} \sum_{i=1}^n \left\{ Y_i - \wh g_n(\wt\bfF_i,X_{i,j}) \right\}^2 - 2 \deltat \dfrac{1}{n} \sum_{i=1}^n \left\{ Y_i - \wh g_n(\wt\bfF_i,X_{i,j}) \right\} \wh\alpha_n(\wt\bfF_i,X_{i,j}) + \deltat^2 \dfrac{1}{n} \sum_{i=1}^n \wh\alpha_n^2(\wt\bfF_i,X_{i,j}) .
\end{align*}
The last line above viewed as a quadratic function in $\deltat$ reaches its minimum at $\deltat=\deltathat$ for $\deltathat$ given in \eqref{eq:k_check}.  We can also re-write \eqref{eq:k_check} as
\begin{align}
    \label{eq:k_hat_explicit}
    \deltathat &= \dfrac{1}{\frac{1}{n} \sum_{i=1}^n \wh\alpha_n^2(\wt\bfF_i,X_{i,j})} \dfrac{1}{n} \sum_{i=1}^n \left\{ \epsilon_i - \wh\delta_n(\wt\bfF_i,X_{i,j}) \right\} \wh\alpha_n(\wt\bfF_i,X_{i,j}) \\
    & = \dfrac{ 1 }{ \underbrace{ \textstyle \frac{1}{n} \sum_{i=1}^n \wh\alpha_n^2(\wt\bfF_i,X_{i,j}) }_{J_{n,1}} } \underbrace{ \dfrac{1}{n} \sum_{i=1}^n \epsilon_i \wh\alpha_n(\wt\bfF_i,X_{i,j})}_{J_{n,2}} - \dfrac{ 1 }{\frac{1}{n} \sum_{i=1}^n \wh\alpha_n^2(\wt\bfF_i,X_{i,j}) } \underbrace{\dfrac{1}{n} \sum_{i=1}^n \wh\delta_n(\wt\bfF_i,X_{i,j}) \wh\alpha_n(\wt\bfF_i,X_{i,j})}_{J_{n,3}} . \nonumber
\end{align}
For the term $J_{n,1}$ (in the denominator) above, similar to \eqref{eq:alpha_square_dev} (with the choice $s'$ being a large enough multiple of $\pn$), we conclude that on an event $\calA_{1,1}$ with $\PP(\calA_{1,1})\ge 1-C e^{-\pn}$,
\begin{align*}
\textstyle \frac{1}{n} \sum_{i=1}^n \wh\alpha_n^2(\wt\bfF_i,X_{i,j}) \ge \frac{1}{2} \int \wh\alpha_n^2(\wt\bff,x_j) \rmd\PP - C \nun^2.
\end{align*}
Next,
\begin{align*}
& \textstyle \int \wh\alpha_n^2(\wt\bff,x_j) \rmd\PP \ge \int (\alphanull)^2(\wt\bff,x_j) \rmd\PP + \int (\wh\alpha_n^2-(\alphanull)^2)(\wt\bff,x_j) \rmd\PP \\
& \textstyle = \int (\alphanull)^2(\wt\bff,x_j) \rmd\PP + \int (\wh\alpha_n+\alphanull)(\wh\alpha_n-\alphanull)(\wt\bff,x_j) \rmd\PP \\
& \textstyle \ge \int (\alphanull)^2(\wt\bff,x_j) \rmd\PP - C M_\infty \{ \int (\wh\alpha_n-\alphanull)^2(\wt\bff,x_j) \rmd\PP \}^{1/2} \ge \int (\alphanull)^2(\wt\bff,x_j) \rmd\PP - C(h^{-1}\nun+\deltaf)
\end{align*}
where the last step holds on an event $\calA_{1,2}$ with $\PP(\calA_{1,2})\ge 1-C e^{-\pn}$ by Theorem~\ref{thm:Riesz_est}.
Combining the above two inequalities, by the conditions imposed in Proposition~\ref{prop:tweaking}
with Assumption \ref{ratexx}\ref{ratexx:con_1}, on the event $\calA_{1,1}\cap\calA_{1,2}$ we have that $J_{n,1}$ is bounded from below from zero by a constant (independent of $n$).

Next, for the term $J_{n,2}$, we can follow the treatment for the terms $I_{n,2} = I_{n,2,1} + \{ I_{n,2,2} - \int I_{n,2,2} \rmd\PP \} + \int I_{n,2,2} \rmd\PP$ in the proof of Theorem~\ref{thm:m_hat_est_master} in Section~\ref{sec:proof_thm:m_hat_est_master}, together with the choice of $s$ in that proof, to conclude that on an event $\calA_2$ with $\PP(\calA_2)\ge 1-C e^{-\pn}$,
\begin{align*}
    |J_{n,2}| \lesssim \{ \|\wh\alpha_n\|_{L_2} + \nun\}\nun + \{ \|\wh\alpha_n\|_{L_2} + \nun\}\nun + \deltaf\|\wh\alpha_n\|_{L_2} \lesssim r_{m,n}
\end{align*}
where the last step follows by the boundedness of $\wh\alpha_n$.

For the term $J_{n,3}$ above, first note that the class of functions $\{(g-m_0)\alpha: g, \alpha\in\calFn(\rbar+1)\}$ satisfies the conditions of Proposition~\ref{prop:m_n_hat_ratio_rate_master} for \eqref{eq:m_n_hat_ratio_rate_2_master} (with minor adjustment of constants) by Theorem~3 in \cite{Andrews1994Handbook} and our assumptions.  Then, by \eqref{eq:m_n_hat_ratio_rate_2_master} in Proposition~\ref{prop:m_n_hat_ratio_rate_master}, together with the choice of $s$ in the proof of Theorem~\ref{thm:m_hat_est_master} in Section~\ref{sec:proof_thm:m_hat_est_master}, we conclude that an event $\calA_3$ with $\PP(\calA_3)\ge 1-C e^{-\pn}$,
\begin{gather*}
    \textstyle |J_{n,3} - \int J_{n,3} \rmd\PP| \lesssim \{ \|\wh\delta_n \wh\alpha_n\|_{L_2} + \nun\} \nun \lesssim \{ M \|\wh\delta_n\|_{L_2} + \nun\}\nun \lesssim r_{m,n}, \\
    \textstyle |\int J_{n,3} \rmd\PP| \lesssim \{ (\PP\wh\alpha_n^2) (\PP\wh\delta_n^2) \}^{1/2} \lesssim M r_{m,n} \lesssim r_{m,n}
\end{gather*}
both hold.  Combining these results, we conclude that $|\deltathat| \lesssim r_{m,n} $ on an event with probability at least $1 - C e^{-\pn}$.

\medskip

\textit{Next, we prove the second half of the proposition on the minimization condition \eqref{eq:min_con}.}
Denote the functions $\wh\delta_{\alpha,n}=\wh\alpha_n - \alphanull$, $\delta_{\alpha,n,t}=\alphanull - \alpha_{t,n}^*$ and $\wh\delta_{\alpha,n,t}=\wh\alpha_n - \alpha_{t,n}^* = \wh\delta_{\alpha,n} + \delta_{\alpha,n,t}$.  It follows from the definition of $\deltathat$ that $\PPn\ell(\cdot; \wc g_n) - \PPn\ell(\cdot; \wc g_n \pm \rinf \wh\alpha_n )\leq 0$.
 We then have, for all $t\in\RR$,
\begin{align*}
    &\PPn\ell(\cdot; \wc g_n) - \PPn\ell(\cdot; \wc g_n \pm \rinf \alpha_{t,n}^* ) \\
    & = {\PPn\ell(\cdot; \wc g_n) - \PPn\ell(\cdot; \wc g_n \pm \rinf \wh\alpha_n )} + \PPn\ell(\cdot; \wc g_n \pm \rinf \wh\alpha_n ) - \PPn\ell(\cdot; \wc g_n \pm \rinf \alpha_{t,n}^* ) \\
    &\le \PPn\ell(\cdot; \wc g_n \pm \rinf \wh\alpha_n ) - \PPn\ell(\cdot; \wc g_n \pm \rinf \alpha_{t,n}^* )\\& = \PPn\ell(\cdot; \wh g_n + \deltathat \wh\alpha_n \pm \rinf \wh\alpha_n ) - \PPn\ell(\cdot; \wh g_n + \deltathat \wh\alpha_n \pm \rinf \alpha_{t,n}^* ) \\
    & = \textstyle \PPn\ell(\cdot; \wh g_n + (\deltathat \pm \rinf ) \wh\alpha_n )
    - \PPn\ell(\cdot; \wh g_n + (\deltathat \pm \rinf ) \wh\alpha_n \mp {\rinf \wh\delta_{\alpha,n,t}} ) \equiv J_{n,t},
\end{align*}
where the inequality follows by our earlier remark.

Next,
\begin{align*}
    & \textstyle J_{n,t} = \frac{1}{n} \sum_{i=1}^n [ \{ \epsilon_i -\wh\delta_n(\wt\bfF_i,X_{i,j}) - (\deltathat \pm \rinf ) \wh\alpha_n(\wt\bfF_i,X_{i,j}) \}^2 \\
    & \textstyle - \{ \epsilon_i -\wh\delta_n(\wt\bfF_i,X_{i,j}) - (\deltathat \pm \rinf ) \wh\alpha_n(\wt\bfF_i,X_{i,j}) \pm \rinf \wh\delta_{\alpha,n,t}(\wt\bfF_i,X_{i,j}) \}^2 ] \\
    & \textstyle = \frac{1}{n} \sum_{i=1}^n \left[ \mp 2 \left\{ \epsilon_i -\wh\delta_n(\wt\bfF_i,X_{i,j}) - (\deltathat \pm \rinf ) \wh\alpha_n(\wt\bfF_i,X_{i,j}) \right\} \rinf \wh\delta_{\alpha,n,t}(\wt\bfF_i,X_{i,j}) - \rinf^2 \wh\delta_{\alpha,n,t}^2(\wt\bfF_i,X_{i,j}) \right] \\
    & \textstyle = \pm (\deltathat \pm \rinf ) \frac{2 \rinf}{n} \sum_{i=1}^n \wh\alpha_n(\wt\bfF_i,X_{i,j}) \wh\delta_{\alpha,n,t}(\wt\bfF_i,X_{i,j}) \mp \frac{2 \rinf}{n} \sum_{i=1}^n \epsilon_i \wh\delta_{\alpha,n,t}(\wt\bfF_i,X_{i,j}) \\
    & \textstyle\quad \pm \frac{2 \rinf}{n} \sum_{i=1}^n \wh\delta_n(\wt\bfF_i,X_{i,j}) \wh\delta_{\alpha,n,t}(\wt\bfF_i,X_{i,j})  - \frac{ \rinf^2 }{n} \sum_{i=1}^n \wh\delta_{\alpha,n,t}^2(\wt\bfF_i,X_{i,j}) \\
    & \textstyle = \pm (\deltathat \pm \rinf ) \frac{2 \rinf}{n} \sum_{i=1}^n \alpha_{t,n}^*(\wt\bfF_i,X_{i,j}) \wh\delta_{\alpha,n,t}(\wt\bfF_i,X_{i,j}) \mp \frac{2 \rinf}{n} \sum_{i=1}^n \epsilon_i \wh\delta_{\alpha,n,t}(\wt\bfF_i,X_{i,j})  \\
    & \textstyle \quad \pm \frac{2 \rinf}{n} \sum_{i=1}^n \wh\delta_n(\wt\bfF_i,X_{i,j}) \wh\delta_{\alpha,n,t}(\wt\bfF_i,X_{i,j}) + (\pm 2 \deltathat + \rinf ) \frac{\rinf}{n} \sum_{i=1}^n \wh\delta_{\alpha,n,t}^2(\wt\bfF_i,X_{i,j}) \\
    & \equiv \pm 2 \rinf (\deltathat \pm \rinf ) J_{n,t,1} \mp 2 \rinf J_{n,t,2} + 2 \rinf J_{n,t,3} + \rinf (\pm 2 \deltathat + \rinf ) J_{n,t,4}
\end{align*}
where we have introduced
\begin{gather*}
\textstyle J_{n,t,1} = \frac{1}{n} \sum_{i=1}^n \alpha_{t,n}^*(\wt\bfF_i,X_{i,j}) \wh\delta_{\alpha,n,t}(\wt\bfF_i,X_{i,j}), \quad J_{n,t,2} = \frac{1}{n} \sum_{i=1}^n \epsilon_i \wh\delta_{\alpha,n,t}(\wt\bfF_i,X_{i,j}), \\
\textstyle J_{n,t,3} = \frac{1}{n} \sum_{i=1}^n \wh\delta_n(\wt\bfF_i,X_{i,j}) \wh\delta_{\alpha,n,t}(\wt\bfF_i,X_{i,j}), \quad J_{n,t,4} = \frac{1}{n} \sum_{i=1}^n \wh\delta_{\alpha,n,t}^2(\wt\bfF_i,X_{i,j}) .
\end{gather*}
We only demonstrate the treatment for $J_{n,t,1}$ in detail; the treatments for the other $J_{n,k}$'s are in principle similar and have been employed in the proofs of Theorems~\ref{thm:m_hat_est_master}, \ref{thm:Riesz_est} and Proposition~\ref{prop:tweaking}.  For the term $J_{n,t,1}$, we further define $J_{n,t,1,1} = \PPn \alpha_{t,n}^* \wh\delta_{\alpha,n}$ and $J_{n,t,1,2} = \PPn \alpha_{t,n}^* \delta_{\alpha,n,t}$ so $J_{n,t,1}=J_{n,t,1,1}+J_{n,t,1,2}$.  Now, we apply \eqref{eq:m_n_hat_ratio_rate_2_master} in Proposition~\ref{prop:m_n_hat_ratio_rate_master} to conclude that with probability at least $1-C e^{-\pn}$,
\begin{align*}
    & \textstyle \forall t\in\calT_\delta,  |J_{n,t,1,1}| = |(\PPn-\PP) \alpha_{t,n}^* \wh\delta_{\alpha,n} + \PP \alpha_{t,n}^* \wh\delta_{\alpha,n}| \lesssim \{ \|\alpha_{t,n}^* \wh\delta_{\alpha,n}\|_{L_2} + \nun \} \nun + \{ \PP \alpha_{t,n}^{*2} \PP \wh\delta_{\alpha,n}^2 \}^{1/2} \\
    & \textstyle \lesssim \|\wh\delta_{\alpha,n}\|_{L_2} \lesssim \rnull .
\end{align*}
To treat $J_{n,t,1,2}$, we apply the same proposition to conclude that with probability at least $1-C e^{-\pn}$,
\begin{align*}
    & \textstyle \forall t\in\calT_\delta, \quad |J_{n,t,1,2}| = |(\PPn-\PP) \alpha_{t,n}^* \delta_{\alpha,n,t} + \PP \alpha_{t,n}^* \delta_{\alpha,n,t}| \lesssim \left\{ \|\alpha_{t,n}^* \delta_{\alpha,n,t}\|_{L_2} + \nun \right\} \nun + \{ \PP \alpha_{t,n}^{*2} \PP \delta_{\alpha,n,t}^2 \}^{1/2} \\
    & \textstyle \lesssim \nun+ r_{\alpha,t,\ub} .
\end{align*}
Thus, with probability at least $1-C e^{-\pn}$,
\begin{align*}
\forall t\in\calT_\delta, \quad |J_{n,t,1}| \lesssim \rnull + r_{\alpha,t,\ub}.
\end{align*}
For the other $J_{n,t,k}$ terms, we can similarly conclude that, with probability at least $1-C e^{-\pn}$,
\begin{align*}
 & \textstyle \forall t\in\calT_\delta, \quad |J_{n,t,2}| \lesssim r_{m,n} (\rnull + r_{\alpha,t,\ub}), \quad |J_{n,t,3}| \lesssim r_{m,n} (\rnull + r_{\alpha,t,\ub}), \quad |J_{n,t,4}| \lesssim \rnull (\rnull \nun + r_{\alpha,t,\ub}) .
\end{align*}
Combining all the $J_{n,t,k}$ terms and further simplifying, we conclude that with probability at least $1-C e^{-\pn}$,
\begin{align*}
 & \textstyle \forall t\in\calT_\delta, \quad |J_{n,t}| \lesssim \rinf r_{m,n} (\rnull + r_{\alpha,t,\ub}) .
\end{align*}
This concludes the proof of the proposition.
\end{proof}

\subsection{Proof of Theorem~\ref{thm:m_check_est_master}}
\label{sec:proof_thm:f_check_est}
\begin{proof}
As in the proof of Theorem~\ref{thm:m_hat_est_master}, given \eqref{eq:m_check_rate_master}, the proof of \eqref{eq:m_check_rate_master_deriv} is straightforward, so we will only prove \eqref{eq:m_check_rate_master}.  By Proposition~\ref{prop:tweaking} and Assumption~\ref{ratexx}, on the event $\calA_{\ut,1}$ (see Proposition~\ref{prop:tweaking}) we have $|\deltathat|\le 1$, and hence further by \eqref{eq:f_check}, on the same event,
\begin{align*}
    \wc g_n \in \calFnbar(\rbar+1)\equiv\{ g + \deltat \alpha: g, \alpha\in\calFn(\rbar+1), |\deltat|\le 1\} .
\end{align*}
We now control the complexity of the class $\calFn(\rbar+1)$.
Let $Q$ be a generic measure on $(\wt\bfF,X_j)$, let $\calN$ be a cover of $\calFn(\rbar+1)$ by balls of radius $\tau/\sqrt{9}$ in the $L_2(Q)$ norm and moreover $\calN$ consist of functions uniformly bounded in magnitude by $M$, and let $\calE$ be a cover of the interval $[-1,1]$ by intervals of length at most $\tau/(\sqrt{9}M)$; we choose $\calE$ to have cardinality at most $\lceil2\sqrt{9}M/\tau\rceil$.

Suppose $g, g', \alpha, \alpha'$ be functions of $(\wt\bfF,X_j)$, that are uniformly bounded in magnitude by $M$, and moreover satisfy $\|g-g'\|_{L_2(Q)}\le \tau/\sqrt{9}$ and $\|\alpha-\alpha'\|_{L_2(Q)}\le \tau/\sqrt{9}$, and let $\deltat, \deltat'\in[-1,1]$ be such that $|\deltat-\deltat'|\le \tau/(\sqrt{9}M)$.  Let $\theta= g + \deltat\alpha$ and $\theta'= g' + \deltat'\alpha'$.  Then
\begin{align*}
    &\textstyle \|\theta'-\theta\|_{L_2(Q)}^2 \le 3\|g-g'\|_{L_2(Q)}^2 + 3 \deltat^2 \|\alpha-\alpha'\|_{L_2(Q)}^2 + 3 (\deltat-\deltat')^2 \|\alpha'\|_{L_2(Q)}^2 \\
    &\textstyle  \le 3\|g-g'\|_{L_2(Q)}^2 + 3 \deltat^2 \|\alpha-\alpha'\|_{L_2(Q)}^2 + 3 (\deltat-\deltat')^2 M^2 \le 3\frac{\tau^2}{9} + 3\frac{\tau^2}{9} + 3\frac{\tau^2}{9} = \tau^2.
\end{align*}
Thus $N(\calFnbar(\rbar+1), L_2(Q), \tau) \le N(\calFn(\rbar+1), L_2(Q), \tau/\sqrt{9})^2 \lceil2\sqrt{9}M/\tau\rceil$ where we follow the notation for the covering numbers, for instance $N(\calFnbar(\rbar+1), L_2(Q), \tau)$, in \eqref{eq:covering_number_bound_master}.  Thus, the logarithm of the covering number $N(\calFnbar(\rbar+1), L_2(Q), \tau)$ for the class $\calFnbar(\rbar+1)$ is bounded up to a multiplicative constant by the logarithm of the covering number for the class $\calFn(\rbar+1)$.  Then, the proof of \eqref{eq:m_check_rate_master} in Theorem~\ref{thm:m_check_est_master} follows that of \eqref{eq:rn1} in Theorem~\ref{thm:m_hat_est_master} with minor modifications.
\end{proof}

\subsection{Proof of the size under the null: Theorem~\ref{thm:main_null_master} }
\label{sec:Proof_Thm_thm:main_null_master}

Let $r=r(\cdot;b)$ be an unspecified, possibly random function that may change for each occurrence, but that always satisfies $\sup_{t\in\calT_\delta}|r(t;b)|\lesssim b$.
By this convention, the function $r=r(t;\Co(1))$ in \eqref{eq:H0_iid} satisfies $\sup_{t\in\calT_\delta}|r(t;\Co(1))|=\Co(1)$.

\begin{proof}[Proof of Theorem~\ref{thm:main_null_master}]

We will first prove the following master approximate i.i.d.\,expansion for the statistic $\wc\eta_t^\us(\wc g_n)$: under the null hypothesis $H_0$ in \eqref{hypo_null}, with probability at least $1 - C e^{-\pn}$, uniformly over all $t\in\calT_\delta$ we have
\begin{align}
\textstyle \sqrt{n} \wc\eta_t^\us(\wc g_n) = \frac{t}{\sqrt{n}} \sum_{i=1}^n \wt\epsilon_i \alphanull(\wt\bfF_i, X_{i,j})
+ r(t;\Co(1)) .
\label{eq:H0_iid}
\end{align}

Recall the notations $\wc\eta_t^\us$ from \eqref{eta_t_s_wc}, $\wt\eta_t^\us$ from \eqref{labeltildeg}, and the random variable $Z_{t,j,h}$ from \eqref{eq:Z_t}.  Introduce, for a generic function $m:\RR^{r+1}\rightarrow\RR$, the functional $\eta_t^\us$ as
\begin{align*}
\textstyle \eta_t^\us(m) = \int_{\Omega_h} [ \exp\{ t\, m_j^\us(\bff, x_j) \} -1 ] \rmd\PP .
\end{align*}
Under either $H_0$ and $H_1$, we have the master decomposition
\begin{align}
    &\wc\eta_t^\us(\wc g_n) = \left\{ \wc\eta_t^\us(\wc g_n) - \wt\eta_t^\us(\wc g_n) \right\} + \left\{ \wt\eta_t^\us(\wc g_n) - \wt\eta_t^\us(g_0) \right\} + \left\{ \wt\eta_t^\us(g_0) - \eta_t^\us(m_0) \right\} \nonumber \\
    & \textstyle \quad + \{ \eta_t^\us(m_0) - \int_{\Omega_h} \left[ \exp\{t\,m_{0,j}(\bff, x_j)\} - 1 \right] \rmd\PP \} + \int_{\Omega_h} \left[ \exp\{t\,m_{0,j}(\bff, x_j)\} - 1 \right] \rmd\PP \nonumber \\
    & \equiv I_{t,n,1} + I_{t,n,2} + \underbrace{I_{t,n,3} + I_{t,n,4} + I_{t,n,5}}_{=\EE Z_{t,j,h}\stackrel{\text{under $H_0$}}{=}0}
    \label{eq:test_stat_master_decomp_master}
\end{align}
where we note that $\wt\eta_t^\us(g_0)=\EE Z_{t,j,h}$ (by a remark in the proof of Lemma~\ref{lem:Riesz_analytic}, $\Psi$ in \eqref{labeltildeg} can be discarded if $g=g_0$).  In this decomposition,
\begin{itemize}[leftmargin=*]
    \item[i)]
    $I_{t,n,1}$ is a centered empirical process evaluated at the random function $\wc g_n$.  Under the null, $I_{t,n,1}$ admits a fast rate converging to zero because here the smoothed derivative estimator $\wc g_{n,j}^\us$ is expected to be close to zero.  {Under the alternative}, $I_{t,n,1}$ still converges to zero due to the averaging effect, though the rate will be comparatively slower.
    \item[ii)]
    $I_{t,n,2}$ is the dominant term in the decomposition under the null, and in fact will admit an approximate i.i.d.\,expansion that will in turn become the leading, i.i.d.\,sum in \eqref{eq:H0_iid}.
    \item[iii)]
    The term $I_{t,n,3}$ represents the approximation error arising from replacing the latent factor $\bfF$ by the diversified factor $\wt\bfF$, $I_{t,n,4}$ is a bias term arising from smoothing the true regression function derivative $m_{0,j}$, and finally, $I_{t,n,5}$ can be regarded as the original signal strength favoring the alternative.  Together, the (magnitude of the) sum $I_{t,n,3}+I_{t,n,4}+I_{t,n,5}$ can be regarded as containing the signal favoring the alternative.  Here, under the null, the sum becomes $\wt\eta_t^\us(g_0)=\EE Z_{t,j,h}$ which is zero (see the remark below \eqref{eq:Z_t}).
\end{itemize}
Next, the treatments of the terms $I_{t,n,1}$ and $I_{t,n,2}$ will be covered in Sections~\ref{sec:In1_master} and \ref{sec:In2_master} respectively.  The transition from \eqref{eq:H0_normality_L2} to \eqref{eq:H0_normality_Ln} will be covered in Section~\ref{sec:H0_var_est} by showing that the estimated variance $\| \wh\epsilon \wh\alpha_n \|_{L_2(\PPn)}^2$ is a consistent estimator of the unknown asymptotic variance $\| \epsilon \alphanull\|_{L_2}^2$.  We finalize our proof of Theorem~\ref{thm:main_null_master} afterward.

\subsubsection{Proof for the first term in \eqref{eq:test_stat_master_decomp_master}}
\label{sec:In1_master}
Now we handle the term $I_{t,n,1}$ in the decomposition \eqref{eq:test_stat_master_decomp_master}, which as stated below \eqref{eq:test_stat_master_decomp_master} is a centered empirical process term. Recall the definition of  $\calT_\delta= [-T, -\delta] \cup [\delta, T]$ where $T$ is a constant.

\begin{proposition}
\label{prop:I_tn1_rate}
Suppose that the conditions of Theorem~\ref{thm:main_null_master} hold.  Then, under either $H_0$ or $H_1$, with probability at least $1 - C e^{-\pn}$, for a constant $c_{4,1}$, and uniformly over $t\in\calT_\delta$,
\begin{align*}
\textstyle |I_{t,n,1}| \le c_{4,1} \{ \frac{|t|}{h} r_{m,n} + \frac{1}{h} \nun + \|Z_{t,j,h}\|_{L_2} \} \nun .
\end{align*}
\end{proposition}
\begin{proof}

Define the function $\theta(\cdot; g,t)=\theta(\wt\bff, x_j; g,t)=\{\exp[t\,\Psi\{ g_j^\us(\wt\bff, x_j)\} ]-1\}\ind\{x_j\in\calB_h\}:\RR^{\rbar+1}\rightarrow\RR$ (essentially the summand in \eqref{eta_t_s_wc}) indexed by $g\in\calFnbar(\rbar+1)$ and $t$.  Then, $I_{t,n,1} = (\PPn-\PP) \theta(\cdot;\wc g_n, t)$.  Now, in the proof of Theorem~\ref{thm:m_check_est_master} in Section~\ref{sec:proof_thm:f_check_est}, it was shown that with probability at least $1 - C e^{-\pn}$, $\wc g_n$ belongs to the functional class $\calFnbar(\bar{r}+1)$ defined in that proof.  We shall focus on this event.

As in the proof of Theorem~\ref{thm:m_check_est_master}, let $g, g'$ be functions of $(\wt\bfF,X_j)$, and moreover let $t, t'\in\calT_\delta$.  Then, by Assumption~\ref{ass:truncation},
\begin{align*}
    &|\theta(\wt\bff, x_j; g, t)-\theta(\wt\bff, x_j; g', t')| = |\exp[t\,\Psi\{ g_j^\us(\wt\bff, x_j)\}] - \exp[t'\,\Psi\{ g^{'\us}_j(\wt\bff,x_j)\} ]| \ind\{x_j\in\calB_h\} \\
    &\lesssim t |(g-g')_j^\us(\wt\bff, x_j)|\ind\{x_j\in\calB_h\} + |t-t'| \le T |(g-g')_j^\us(\wt\bff, x_j)|\ind\{x_j\in\calB_h\} + |t-t'|.
\end{align*}
Since $t\in\calT_\delta$ and $\calT_\delta$ is compact, the class of functions $\{\theta(\cdot; g,t): g\in\calFnbar(\rbar+1), t\in\calT_\delta\}$ admits a uniform covering number similar to that in \eqref{eq:alpha_covering_number}.  Thus, the situation is largely similar to that in Proposition~\ref{prop:g_n_hat_rate}, and we obtain that on an event with probability at least $1-12 e^{-s}$, for a constant $c_6'$,
\begin{align*}
\textstyle \forall g\in\calFnbar(\rbar+1), \forall t\in\calT_\delta, \quad |(\PPn-\PP) \theta(\cdot;g,t)| \le c_6' [ \left\{ \| \theta \|_{L_2} + \nusn \right\} \left\{ \nusn + \sqrt{\frac{s}{n}} \right\} + \frac{1}{h^2}\nusn^2 ].
\end{align*}
Choosing $s$ as in the proof in Section~\ref{sec:proof_thm:m_hat_est_master} of Theorem~\ref{thm:m_hat_est_master} in the above, we obtain that, with probability at least $1-C e^{-\pn}$, by changing the constant $c_6'$ if necessary,
\begin{align}
\textstyle \forall t\in\calT_\delta, \quad |I_{t,n,1}| \le c_6' \{ \| \theta(\cdot;\wc g_n, t) \|_{L_2} + \frac{1}{h} \nun \} \nun .
\label{eq:I_{t,n,1}_intermediate}
\end{align}
Now, with the aid of \eqref{eq:derivative_bound_via_original} and the definition of $Z_{t,j,h}$ from \eqref{eq:Z_t},
\begin{align*}
&\textstyle \| \theta(\cdot;\wc g_n, t) \|_{L_2}^2 = \int_{\Omega_h} \{ \exp(t\,\Psi\{\wc g_{n,j}^\us(\wt\bff, x_j)\}) -1 \}^2 \rmd\PP \\
&\textstyle  \le 2 \int_{\Omega_h} \{ \exp(t\,\Psi\{\wc g_{n,j}^\us(\wt\bff, x_j)\}) - \exp(t\,\Psi\{g_{0,j}^\us(\wt\bff, x_j)\}) \}^2 \rmd\PP + 2 \int_{\Omega_h} \{ \exp(t\,\Psi\{g_{0,j}^\us(\wt\bff, \rmd x_j)\} ) -1 \}^2 \rmd\PP \\
& \textstyle \lesssim t^2 \int_{\Omega_h} \{ (\wc g_n-g_0)_j^\us \}^2 \rmd\PP + \|Z_{t,j,h}\|_{L_2}^2 \lesssim \frac{t^2}{h^2} \int (\wc g_n-g_0)^2 \rmd\PP + \|Z_{t,j,h}\|_{L_2}^2 \\
& \textstyle \lesssim \frac{t^2}{h^2} \int \{ (\wc g_n-m_0)^2 + (g_0-m_0)^2 \} \rmd\PP + \|Z_{t,j,h}\|_{L_2}^2 \lesssim \frac{t^2}{h^2} r_{m,n}^2 + \|Z_{t,j,h}\|_{L_2}^2 ,
\end{align*}
where the last step follows by Theorem~\ref{thm:m_check_est_master} and Lemma~\ref{lemma:factor}.
Combining the above with \eqref{eq:I_{t,n,1}_intermediate} then yield the conclusion of the proposition.
\end{proof}

\subsubsection{Proof for the second term in \eqref{eq:test_stat_master_decomp_master}}
\label{sec:In2_master}

Recall that $\alpha_{t,n}^*=\alpha_{t,n}^*(\wt\bff,x_j): \RR^{\rbar+1}\rightarrow\RR$ is the Riesz representer that satisfies \eqref{eq:Riesz_H0_master}.
\begin{proposition}
    \label{prop:I_tn2_rate}
    Suppose that the conditions of Theorem~\ref{thm:main_null_master} hold.  Then, under either $H_0$ or $H_1$, with probability at least $1 - C e^{-\pn}$, for a common function $r=r(\cdot;b)$ as introduced at the beginning of Section~\ref{sec:Proof_Thm_thm:main_null_master}, uniformly over $t\in\calT_\delta$ we have
    \begin{align}
     \textstyle I_{t,n,2} = t \frac{1}{n}\sum_{i=1}^n \wt\epsilon_i \alpha_{t,n}^*(\wt\bfF_i,X_{i,j})  + |t|\, r(t; r_{m,n} r_{\alpha,t,n} + \deltaf ) + (|t|+t^2) r(t; h^{-2} r_{m,n}^2 ) .
     \label{eq:I_tn2_rate}
\end{align}
\end{proposition}

\begin{proof}
To streamline our presentation, we divide the proof into two parts.  In the first part, Lemma~\ref{lem:In2_1_master}, the proposition is proven under a set of technical conditions, while in the second part, Lemma~\ref{lem:In2_2_master}, the technical conditions themselves are verified.  The general direction of the proof of Lemma~\ref{lem:In2_1_master} follows that of \cite{ChenLiaoSun2014, ChenLiaoWang2022inference}, but here we work with the adjusted regression function estimator $\wc g_n$ instead of the preliminary estimator $\wh g_n$, and also take into account the accuracy loss of using $\wt\bfF$ in place of $\bfF$.  {More importantly, by using the adjusted estimator $\wc g_n$, we avoid artificially imposing the ``undersmoothing/bias vanishing abnormally fast'' condition as in \cite{ChenLiaoWang2022inference}; more specifically, see how their ``undersmoothing'' condition, Assumption~4.4, is invoked in the proof of their Lemma~C.1(ii).  If we were to impose a similar condition, then we would need $\nuan=\nu_{\ua,n,\textup{undersmooth}}$ from \eqref{eq:nuan_rate_concrete},  the bound on the bias, to be $\Co(n^{-1/2})$; this would lead to a scaling of $\wt L \wt k_{\textup{0}}$ that would further result in $\nusn=\nu_{\us,n,\textup{undersmooth}}$ from \eqref{eq:nu_n_generic}, the accompanying standard deviation for the stochastic variation, to be at least on the order of {$n^{-\frac{2\kappa-1}{4\kappa}}$} (up to log factors).  Clearly, the rates for $\nu_{\ua,n,\textup{undersmooth}}$ and $\nu_{\us,n,\textup{undersmooth}}$ in the bias/variance decomposition are no longer balanced, and in particular the rate for $\nu_{\us,n,\textup{undersmooth}}$ is inferior to our balanced rate $\nun$ in \eqref{eq:nu_n_concrete}; in fact, when $\kappa<1/2$, it is no longer true that $\nu_{\us,n,\textup{undersmooth}}=\Co(1)$, so the regression function estimators $\wh g_n$ and $\wc g_n$ are no longer consistent.
}\\

\begin{lemma}
\label{lem:In2_1_master}
Suppose that Assumption~\ref{ratexx}\ref{ratexx:con_2} holds, and that on an event $\calA$, uniformly over $t\in\calT_\delta$ the following conditions hold:
\begin{align}
\label{eq:min_con_master}
\PPn\ell(\cdot;\wc g_n)-\PPn\ell(\cdot;\wc g_n \pm \rinf \alpha_{t,n}^*) &\le r(t; \rinf r_{m,n} r_{\alpha,t,n} ) \\
\intertext{{(note that \eqref{eq:min_con_master} is essentially the minimization condition \eqref{eq:min_con} at $t$ for the tolerance $b_n$ on the order of $r_{m,n} (\rnull + r_{\alpha,t,\ub})$),}}
\label{c1_master}
\PP\{ \ell(\cdot;\wc g_n)-\ell(\cdot;\wc g_n \pm \rinf \alpha_{t,n}^*) \} \hspace{3cm} & \nonumber \\
\pm 2\rinf \langle \wc g_n-m_0, \alpha_{t,n}^* \rangle {\mp 2 \rinf \EE [\epsilon \alpha_{t,n}^*(\wt\bfF,X_j)}] &= r(t;\rinf^2) , \\
\label{c2_master}
(\PPn-\PP)\left\{ \ell(\cdot;\wc g_n)-\ell(\cdot;\wc g_n \pm \rinf \alpha_{t,n}^*) \mp 2 \rinf \epsilon \alpha_{t,n}^* \right\} &= r(t;\rinf r_{m,n} \nun) , \\
\label{c5_master}
\wt\eta_t^\us(\wc g_n)-{\wt\eta_t^\us(g_0)}-\frac{\partial \wt\eta_t^\us(g_0)}{\partial g}[\wc g_n-g_0] &= (|t|+t^2) r(t; h^{-2} r_{m,n}^2 ) , \\
\label{c6_master}
\frac{\partial\wt\eta_t^\us(g_0)}{\partial g}[\wc g_n-g_0] &= t \langle \wc g_n-g_0, \alpha_{t,n}^* \rangle .
\end{align}
Then, on the intersection of the event $\calA$ and another event with probability at least $1 - C e^{-\pn}$, uniformly over $t\in\calT_\delta$, \eqref{eq:I_tn2_rate} in Proposition~\ref{prop:I_tn2_rate} hold.
\end{lemma}

\begin{proof}
By condition \eqref{c2_master},
\begin{align*}
&\PPn\ell(\cdot;\wc g_n)-\PPn\ell(\cdot;\wc g_n \pm \rinf \alpha_{t,n}^*) = \PP\{ \ell(\cdot;\wc g_n)-\ell(\cdot;\wc g_n \pm \rinf \alpha_{t,n}^*) \} \pm 2 \rinf (\PPn-\PP)\{ \epsilon \alpha_{t,n}^* \} \\
&\quad + (\PPn-\PP)\left\{ \ell(\cdot;\wc g_n) - \ell(\cdot;\wc g_n \pm \rinf \alpha_{t,n}^*) \mp 2 \rinf \epsilon \alpha_{t,n}^*\right\} \\
&= \PP\{ \ell(\cdot;\wc g_n)-\ell(\cdot;\wc g_n \pm \rinf \alpha_{t,n}^*) \} \pm 2 \rinf (\PPn-\PP)\{ \epsilon \alpha_{t,n}^* \} + r(t;\rinf r_{m,n} \nun) .
\end{align*}
Combining the above with condition~\eqref{c1_master} and Assumption~\ref{ratexx}\ref{ratexx:con_2} yields
\begin{align*}
\PPn\ell(\cdot;\wc g_n)-\PPn\ell(\cdot;\wc g_n \pm \rinf \alpha_{t,n}^*)= \mp 2\rinf \langle \wc g_n-m_0,\alpha_{t,n}^* \rangle \pm 2 \rinf \PPn\{ \epsilon \alpha_{t,n}^* \} + r(t;\rinf r_{m,n} \nun) .
\end{align*}
By the above and the minimization condition~\eqref{eq:min_con_master}, we further obtain
\begin{align}
\label{close_master}
\langle \wc g_n-m_0,\alpha_{t,n}^* \rangle - \PPn\{ \epsilon \alpha_{t,n}^* \} =r(t; r_{m,n} r_{\alpha,t,n}).
\end{align}

Next, we show that $t \langle \wc g_n-m_0,\alpha_{t,n}^* \rangle$ can approximate $\eta_t^\us(\wc g_n)-\eta_t^\us(g_0)$ well.  First, by conditions~\eqref{c5_master} and \eqref{c6_master},
\begin{align*}
I_{t,n,2} &= \wt\eta_t^\us(\wc g_n) - \wt\eta_t^\us(g_0) = \frac{\partial \wt\eta_t^\us(g_0)}{\partial g}[\wc g_n-g_0] + \left\{ \wt\eta_t^\us(\wc g_n)-\wt\eta_t^\us(g_0)-\frac{\partial \wt\eta_t^\us(g_0)}{\partial g}[\wc g_n-g_0] \right\} \\
&= t \langle \wc g_n-g_0, \alpha_{t,n}^* \rangle + (|t|+t^2) r(t; h^{-2} r_{m,n}^2 ) .
\end{align*}
Combining the above with \eqref{close_master}, and using the Cauchy-Schwarz inequality followed by Lemma~\ref{lemma:factor}, we obtain
\begin{align*}
\textstyle I_{t,n,2}  & \textstyle = t \langle \wc g_n-g_0, \alpha_{t,n}^* \rangle + (|t|+t^2) r(t; h^{-2} r_{m,n}^2 ) \\
& \textstyle = t\, r(t; r_{m,n} r_{\alpha,t,n} ) + (|t|+t^2) r(t; h^{-2} r_{m,n}^2 ) + t \langle m_0-g_0, \alpha_{t,n}^* \rangle + t \frac{1}{n}\sum_{i=1}^n \epsilon_i \alpha_{t,n}^*(\wt\bfF_i,X_{i,j}) \\
& \textstyle = t\, r(t; r_{m,n} r_{\alpha,t,n} + \deltaf ) + (|t|+t^2) r(t; h^{-2} r_{m,n}^2 ) + t \frac{1}{n}\sum_{i=1}^n \epsilon_i \alpha_{t,n}^*(\wt\bfF_i,X_{i,j}) .
\end{align*}
Compared to \eqref{eq:I_tn2_rate}, we only need to replace the $\epsilon_i$'s by the $\wt\epsilon_i$'s in the last line above.  This can be done similarly to how we dealt with the term $I_{n,2,2}$ in the proof of Theorem~\ref{thm:m_hat_est_master} in Section~\ref{sec:proof_thm:m_hat_est_master}, and is why the intersection of $\calA$ with another event with probability at least $1 - C e^{-\pn}$ in the statement of Lemma~\ref{lem:In2_1_master} is needed, at the cost of another $t\, r(t; \deltaf )$ term.  Eventually we arrive at
\begin{align*}
\textstyle I_{t,n,2} - t \frac{1}{n}\sum_{i=1}^n \wt\epsilon_i \alpha_{t,n}^*(\wt\bfF_i,X_{i,j}) & \textstyle = t\, r(t; r_{m,n} r_{\alpha,t,n} + \deltaf ) + (|t|+t^2) r(t; h^{-2} r_{m,n}^2 ) ,
\end{align*}
which is the conclusion of the lemma.
\end{proof}

\begin{lemma}
\label{lem:In2_2_master}
Suppose that the conditions of Theorem~\ref{thm:main_null_master} hold.  Then, conditions~\eqref{eq:min_con_master}, \eqref{c1_master}, \eqref{c2_master}, \eqref{c5_master} and \eqref{c6_master} hold with probability at least $1 - C e^{-\pn}$.
\end{lemma}

\begin{proof}
As stated in Lemma~\ref{lem:In2_1_master},
Ineq.~\eqref{eq:min_con_master} is essentially the minimization condition \eqref{eq:min_con}, which has already been established in Proposition~\ref{prop:tweaking}.  Next, we verify condition~\eqref{c1_master}.  We have
\begin{align}
& \ell(Y,\wt\bfF,X_j;\wc g_n) - \ell(Y,\wt\bfF,X_j;\wc g_n \pm \rinf \alpha_{t,n}^*) =\{ Y-\wc g_n(\wt\bfF,X_j) \}^2 - \{ Y-\wc g_n(\wt\bfF,X_j) \mp \rinf \alpha_{t,n}^*(\wt\bfF,X_j) \}^2 \nonumber \\
&= [ \epsilon - \{ \wc g_n(\wt\bfF,X_j) - m_0(\bfF,X_j) \} ]^2 - [ \epsilon - \{ \wc g_n(\wt\bfF,X_j) - m_0(\bfF,X_j) \} \mp \rinf \alpha_{t,n}^*(\wt\bfF,X_j) ]^2 \nonumber \\
&= \pm 2 \rinf [\epsilon - \{ \wc g_n(\wt\bfF,X_j) - m_0(\bfF,X_j) \} ] \alpha_{t,n}^*(\wt\bfF,X_j) - \rinf^2 \alpha_{t,n}^{*2}(\wt\bfF,X_j) \nonumber \\
&= \mp 2 \rinf \{ \wc g_n(\wt\bfF,X_j) - m_0(\bfF,X_j) \} \alpha_{t,n}^*(\wt\bfF,X_j) - \rinf^2 \alpha_{t,n}^{*2}(\wt\bfF,X_j) \pm 2 \rinf \epsilon \alpha_{t,n}^*(\wt\bfF,X_j) .
\label{eq:ell_diff}
\end{align}
Then, \eqref{c1_master} follows by taking the integral of \eqref{eq:ell_diff} against the measure $\PP$.

Next, we verify condition~\eqref{c2_master}.  Starting from Eq.~\eqref{eq:ell_diff},
\begin{align*}
&(\PPn-\PP)\left\{ \ell(\cdot;\wc g_n)-\ell(\cdot;\wc g_n \pm \rinf \alpha_{t,n}^*){\mp} 2 \rinf \epsilon \alpha_{t,n}^* \right\}
\\&= {\mp} 2 \rinf (\PPn-\PP) \{ (\wc g_n-m_0) \alpha_{t,n}^* \} - \rinf^2 (\PPn-\PP) \alpha_{t,n}^{*2} = r(t;\rinf r_{m,n} \nun) ,
\end{align*}
where the last step holds with probability at least $1 - C e^{-\pn}$, and follows by (at this point) the usual application of \eqref{eq:m_n_hat_ratio_rate_2_master} in Proposition~\ref{prop:m_n_hat_ratio_rate_master}, with $r_{m,n}$ from Theorem~\ref{thm:m_check_est_master}.

Next, we verify condition~\eqref{c5_master}.  In condition~\eqref{c5_master}, the second-order remainder term is, for some $\tau'$ strictly between $0$ and $1$, and $v=\wc g_n-g_0$,
\begin{align*}
& \frac{1}{2} \left[ \frac{\partial^2}{\partial\tau^2} \int_{\Omega_h} [ e^{t\,\Psi\{(g_0+\tau v)_j^\us(\wt\bff, x_j)\}} - 1 ] \rmd\PP \right] \Big\vert_{\tau=\tau'} \\
& = \frac{1}{2} \left[ \textstyle \frac{\partial}{\partial\tau} \displaystyle \left[ \int_{\Omega_h} t \dot\Psi\{(g_0+\tau v)_j^\us(\wt\bff, x_j)\} v_j^\us(\wt\bff, x_j) e^{t\,\Psi\{(g_0+\tau v)_j^\us(\wt\bff, x_j)\}} \rmd\PP \right] \right] \Big\vert_{\tau=\tau'} \\
& = \frac{1}{2} \bigg[ \int_{\Omega_h} \Big[ t \ddot\Psi\{(g_0+\tau v)_j^\us(\wt\bff, x_j)\} (v_j^\us)^2(\wt\bff, x_j) e^{t\,\Psi\{(g_0+\tau v)_j^\us(\wt\bff, x_j)\}} \\
& \textstyle \quad + t^2 \dot\Psi^2\{(g_0+\tau v)_j^\us(\wt\bff, x_j)\} (v_j^\us)^2(\wt\bff, x_j) e^{t\,\Psi\{(g_0+\tau v)_j^\us(\wt\bff, x_j)\}} \Big] \rmd\PP \bigg]\Big\vert_{\tau=\tau'} ,
\end{align*}
which by Assumption~\ref{ass:truncation} and the boundedness of $\calT_\delta$ is bounded in magnitude and up to a multiplicative factor (that is not dependent on $t$) by $(|t|+t^2) \int_{\Omega_h} v_j^\us(\wt\bff, x_j)^2 \rmd\PP$.  Thus, condition~\eqref{c5_master} is verified because, by the first half of \eqref{eq:derivative_bound_via_original} in Lemma~\ref{lemma:derivative_bound_via_original_specific} applied to this result,
\begin{align*}
& \textstyle \|\wt\eta_t^\us( \wc g_n )-\wt\eta_t^\us(g_0)-\frac{\partial\wt\eta_t^\us(g_0)}{\partial g}[\wc g_n-g_0]\|_{L_2} \lesssim (|t|+t^2) \|(\wc g_n-g_0)_j^\us\|_{L_2}^2 \\
& \textstyle \lesssim (|t|+t^2)\frac{1}{h^2} \|\wc g_n-g_0\|_{L_2}^2\lesssim (|t|+t^2)\frac{1}{h^2} r_{m,n}^2 .
\end{align*}

Finally, we verify condition~\eqref{c6_master}.
We just need to recall the first part of Lemma~\ref{lem:Riesz_analytic}, from which condition~\eqref{c6_master} follows as a special case.  This completes the proof of Lemma~\ref{lem:In2_2_master}.
\end{proof}
Then, as mentioned earlier, combining Lemmata~\ref{lem:In2_1_master} and \ref{lem:In2_2_master} completes the proof of Proposition~\ref{prop:I_tn2_rate}.
\end{proof}

\subsubsection{Variance estimation}
\label{sec:H0_var_est}

\begin{proposition}
\label{prop:H0_var_est}
Suppose that the conditions of Theorem~\ref{thm:main_null_master} hold.  Then, under either $H_0$ or $H_1$, with probability at least $1 - C e^{-\pn}$,
\begin{align*}
\textstyle \| \wh\epsilon \wh\alpha_n \|_{L_2(\PPn)}^2 - \| \epsilon \alphanull\|_{L_2}^2 = \Co(1) .
\end{align*}
Hence, under $H_0$, we are free to replace $\| \epsilon \alphanull\|_{L_2}^2$ in \eqref{eq:H0_normality_L2} by $\| \wh\epsilon \wh\alpha_n \|_{L_2(\PPn)}^2$ to arrive at \eqref{eq:H0_normality_Ln}.
\end{proposition}

\begin{proof}
We start from the overall decomposition
\begin{align*}
& \textstyle \| \wh\epsilon \wh\alpha_n \|_{L_2(\PPn)}^2 - \| \epsilon \alphanull\|_{L_2}^2 =\frac{1}{n} \sum_{i=1}^n \{m_0(\bfF_i,X_{i,j})+\epsilon_i-\wc g_n(\wt\bfF_i,X_{i,j})\}^2 \wh\alpha_n(\wt\bfF_i,X_{i,j})^2 - \| \epsilon \alphanull\|_{L_2}^2 \\
& \textstyle = \underbrace{ \textstyle \frac{1}{n} \sum_{i=1}^n \{\wc g_n(\wt\bfF_i,X_{i,j})-m_0(\bfF_i,X_{i,j})\}^2 \wh\alpha_n(\wt\bfF_i,X_{i,j})^2}_{J_{n,1}} \\
& \textstyle \quad - 2 \underbrace{ \textstyle  \frac{1}{n} \sum_{i=1}^n \epsilon_i \{\wc g_n(\wt\bfF_i,X_{i,j})-m_0(\bfF_i,X_{i,j})\} \wh\alpha_n(\wt\bfF_i,X_{i,j})^2}_{J_{n,2}} \\
& \textstyle \quad +\underbrace{ \textstyle \frac{1}{n} \sum_{i=1}^n [ \epsilon_i^2 \wh\alpha_n(\wt\bfF_i,X_{i,j})^2 - \PP\{\epsilon^2 \wh\alpha_n^2\} ] }_{J_{n,3}} +\underbrace{ \textstyle \PP\{\epsilon^2 \wh\alpha_n^2\} - \| \epsilon \alphanull\|_{L_2}^2 }_{J_{n,4}}\\
& = J_{n,1} - 2 J_{n,2} + J_{n,3} + J_{n,4} .
\end{align*}
The term $J_{n,1}$ can be treated similarly to how we dealt with the term $I_{n,1}$ in the proof of Theorem~\ref{thm:m_hat_est_master} in Section~\ref{sec:proof_thm:m_hat_est_master}, specifically by invoking Lemma~\ref{lem:square_deviation} and the choice of $s'$ in the proof of Theorem~\ref{thm:m_hat_est_master}.  We conclude that with probability at least $1 - C e^{-\pn}$,
\begin{align*}
|J_{n,1}| & \le |(\PPn-\PP) [\{\wc g_n-m_0\}^2 \wh\alpha_n^2 ]| + \PP [\{\wc g_n-m_0\}^2 \wh\alpha_n^2 ] \\
&\lesssim \PP [\{\wc g_n-m_0\}^2 \wh\alpha_n^2 ] + \nun^2 + \PP [\{\wc g_n-m_0\}^2 \wh\alpha_n^2 ] \lesssim r_{m,n}^2 M^2 \lesssim r_{m,n}^2 .
\end{align*}

The term $J_{n,2}$ can be treated similarly to how we dealt with the term $I_{n,2}$ in the proof of Theorem~\ref{thm:m_hat_est_master}.  We conclude that with probability at least $1 - C e^{-\pn}$,
\begin{align*}
|J_{n,2}| & \lesssim \left\{ \| (\wc g_n-m_0) \wh\alpha_n^2 \|_{L_2} + \nun \right\} \nun + \deltaf \| (\wc g_n-m_0) \wh\alpha_n^2 \|_{L_2} \lesssim r_{m,n}^2 .
\end{align*}

The term $J_{n,3}$ is a centered empirical process $(\PPn-\PP)\{\epsilon^2\alpha^2\}$ evaluated at the random function $\alpha=\wh\alpha_n$.  This time neither Lemma~\ref{lem:square_deviation} nor Proposition~\ref{prop:m_n_hat_ratio_rate_master} directly applies because the factor $\epsilon^2$ involved is sub-exponential instead of sub-Gaussian.  Nevertheless one can still apply Theorem~1 in \cite{Adamczak2008} to conclude that with probability at least $1 - C e^{-\pn}$, $|J_{n,3}|\lesssim \nun$.

Finally, for the term $J_{n,4}$, first using the sub-Gaussianity of $\epsilon$ and then the Cauchy-Schwarz inequality,
\begin{align*}
    &| \PP\{\epsilon^2 \wh\alpha_n^2\} - \| \epsilon \alphanull\|_{L_2}^2 | \lesssim | \PP\{ \wh\alpha_n^2 - (\alphanull)^2\}| = | \PP\{ (\wh\alpha_n+ \alphanull) (\wh\alpha_n - \alphanull) \}| \\
    & \lesssim [\PP\{ (\wh\alpha_n+ \alphanull)^2\} \PP\{ (\wh\alpha_n - \alphanull)^2 \}]^{1/2} \lesssim \rnull
\end{align*}
where the last step holds with probability at least $1 - C e^{-\pn}$ by Theorem~\ref{thm:Riesz_est}.

Collecting terms, we conclude that on an event with probability at least $1 - C e^{-\pn}$, $|\| \wh\epsilon \wh\alpha_n \|_{L_2(\PPn)}^2 - \| \epsilon \alphanull\|_{L_2}^2| \lesssim r_{m,n}^2 + \nun + \rnull$ which is easily $\Co(1)$ under condition~\ref{thm:main_null_master:con_1} in Theorem~\ref{thm:main_null_master}.  This completes the proof of the proposition.
\end{proof}

Now we are ready to put everything together to complete the proof of the component of Theorem~\ref{thm:main_null_master} regarding the fixed-$t$ test \eqref{test:fixed_t}.  Under the null, as argued earlier in the proof of Theorem~\ref{thm:main_null_master}, the sum $I_{t,n,3}+I_{t,n,4}+I_{t,n,5}=0$ in \eqref{eq:test_stat_master_decomp_master}.  Next, (under the null) the leading term in \eqref{eq:I_tn2_rate} in Proposition~\ref{prop:I_tn2_rate} becomes $t \frac{1}{n}\sum_{i=1}^n \wt\epsilon_i \alphanull(\wt\bfF_i,X_{i,j})$, while the term $I_{t,n,1}$ and the remainder terms in \eqref{eq:I_tn2_rate} are bounded by $\Co(n^{-1/2})$ with probability at least $1 - C e^{-\pn}$ by Propositions~\ref{prop:I_tn1_rate} and \ref{prop:I_tn2_rate} and condition~\ref{thm:main_null_master:con_1} of Theorem~\ref{thm:main_null_master}.  This concludes the i.i.d.\,representation in \eqref{eq:H0_iid}.  Next, \eqref{eq:H0_normality_L2} follows from \eqref{eq:H0_iid} by an application of {a central limit theorem} followed by a straightforward replacement of $\| \wt\epsilon \alphanull\|_{L_2}$ by $\| \epsilon \alphanull\|_{L_2}$ based on \eqref{eq:cond_exp_stability}.  Finally, the transition from \eqref{eq:H0_normality_L2} to \eqref{eq:H0_normality_Ln} is covered by Proposition~\ref{prop:H0_var_est} in Section~\ref{sec:H0_var_est}.

Finally, we prove the component of Theorem~\ref{thm:main_null_master} regarding the square test involving $\wh\chi^2$ in \eqref{square}; the proof for the sup test involving $\wh Z$ in \eqref{sup} is analogous.
By the earlier \eqref{eq:H0_iid}, with probability at least $1 - C e^{-\pn}$, uniformly over all $t\in\calT_\delta$ we have
\begin{align*}
\textstyle \frac{n}{t^2} \{\wc\eta_t^\us(\wc g_n)\}^2 = \{ \frac{1}{\sqrt{n}} \sum_{i=1}^n \wt\epsilon_i \alphanull(\wt\bfF_i, X_{i,j}) \}^2 + \{ \frac{1}{\sqrt{n}} \sum_{i=1}^n \wt\epsilon_i \alphanull(\wt\bfF_i, X_{i,j}) \} r(t;\Co(1)) + r(t;\Co(1))^2 .
\end{align*}
Therefore, by the central limit theorem and the Slutsky's theorem,
\begin{align*}
\textstyle \int_{\calT_\delta} \frac{n}{t^2} \{\wc\eta_t^\us(\wc g_n)\}^2 w(t) \rmd t
\textstyle \rightarrow_{d} \left\{ \int_{\calT_\delta} w(t) \rmd t \right\} \| \wt\epsilon \alphanull\|_{L_2}^2 \chi_1^2
\end{align*}
which, after replacing $\wt\epsilon$ by $\epsilon$ as in the proof of the fixed-$t$ case, is equivalent to the last equation in \eqref{eq:H0_normality_L2}.  Again by Proposition~\ref{prop:H0_var_est}, we are further free to replace $\| \epsilon \alphanull\|_{L_2}$ in \eqref{eq:H0_normality_L2} by $\|\wh\epsilon \wh\alpha_n\|_{L_2(\PPn)}$.

\end{proof}

\subsection{Power of the fixed-$t$ test under the alternative}
\label{sec:alter_supp}
In this subsection, our principal Theorem~\ref{thm:main_alternative}
guarantees that the power of our fixed-$t$ test in \eqref{test:fixed_t} approaches one under all local alternatives satisfying conditions~\ref{thm:main_alt:con_1} and \ref{thm:main_alt:con_2} in the theorem that place mild conditions on how fast $|\EE Z_{t,j,h}|$, the signal strength favoring the alternative, can decay.  Similar results for the ensemble tests \eqref{sup} and \eqref{square} can be shown using analogous techniques and we omit the details here.

\begin{theorem}
\label{thm:main_alternative}
Suppose that the conditions of Theorem~\ref{thm:main_null_master} hold.  Furthermore, suppose that at some $t\in\calT_\delta$, for large enough constants $C_{\textup{alt},1}$ and $C_{\textup{alt},2}$,
\begin{enumerate*}[label=(\roman*)]
\item\label{thm:main_alt:con_1} $C_{\textup{alt},1} \sqrt{\log(n)/n} \le |\EE Z_{t,j,h}|$;
\item\label{thm:main_alt:con_2}
$C_{\textup{alt},2} \{ \deltaf + \frac{1}{h^2} (\nun^2+\deltaf^2) + \frac{1}{h} (\nun+\deltaf) \|Z_{t,j,h}\|_{L_2} \} \le |\EE Z_{t,j,h}|$.
\end{enumerate*}
Then, under the alternative hypothesis $H_1$ in \eqref{hypo_alt}, with probability at least $1 - C e^{-\pn} - \frac{1}{n}$ our fixed-$t$ test in \eqref{test:fixed_t} will reject; thus the power of the test tends to one.
\end{theorem}

Condition~\ref{thm:main_alt:con_1} in the theorem is satisfied as long as $|\EE Z_{t,j,h}|$ {is a large enough multiple of $\sqrt{\log(n)/n}$}.  In condition~\ref{thm:main_alt:con_2}, the signal strength $Z_{t,j,h}$ appears on both sides; however, $\|Z_{t,j,h}\|_{L_2}$ on the left-hand side, which admittedly is no smaller than $|\EE Z_{t,j,h}|$ on the right-hand side, is scaled by $\nun+\deltaf$ which we could expect to converge to zero.  Thus, condition~\ref{thm:main_alt:con_2} is mild as well.  It is worth noting that condition~\ref{thm:main_alt:con_1} is itself implied by, for large enough constants $C_{\textup{alt},1}'$ and $C_{\textup{alt},3}$,
\begin{enumerate}[label=(\alph*)]
\item\label{thm:main_alt:con_3} $C_{\textup{alt},1}' \sqrt{\log(n)/n} \le |\int_{\Omega_h} \left[ \exp\{t\,m_{0,j}(\bff, x_j)\} - 1 \right] \rmd\PP|$;
\item\label{thm:main_alt:con_4} $C_{\textup{alt},3} \{\frac{1}{h} \deltaf + r_{\textup{b},m,j} \} \le |\int_{\Omega_h} \left[ \exp\{t\,m_{0,j}(\bff, x_j)\} - 1 \right] \rmd\PP|$ for $r_{\textup{b},m,j}$ introduced above Theorem~\ref{thm:m_hat_est_master}.
\end{enumerate}
Here, condition~\ref{thm:main_alt:con_4} states that the effect of diversification, represented by $\deltaf$, and the effect of smoothing the true derivative, represented by $r_{\textup{b},m,j}$, should be weak compared to the original signal.  We also prove the remarks above on conditions~\ref{thm:main_alt:con_3} and \ref{thm:main_alt:con_4} following the proof of Theorem~\ref{thm:main_alternative}.

\begin{proof}[Proof of Theorem~\ref{thm:main_alternative}]
The original decomposition~\eqref{eq:test_stat_master_decomp_master} in the proof of Theorem~\ref{thm:main_null_master} remains, but under the alternative the treatment of the various constituent terms could be different.
 Now, {$|\EE Z_{t,j,h}|$ is the} signal strength favoring the alternative; thus, in contrast, $I_{t,n,1}+I_{t,n,2}$ becomes the ``noise'' under the alternative.  The sum of the left-hand sides of conditions~\ref{thm:main_alt:con_1} and \ref{thm:main_alt:con_2} in the statement of Theorem~\ref{thm:main_alternative} is, up to a multiplicative factor and  with probability at least $1 - C e^{-\pn}$, a bound on the magnitude of the noise $I_{t,n,1}+I_{t,n,2}$.  The bound can be straightforwardly established by Propositions~\ref{prop:I_tn1_rate} and \ref{prop:I_tn2_rate} which hold under the alternative as well.  We only remark that
\begin{itemize}[leftmargin=*]
    \item[i)]
    The necessity of condition~\ref{thm:main_alt:con_1} is due to the fact that the leading i.i.d.\,term within $I_{t,n,2}$ now also becomes a part of the noise that could be bounded through Hoeffding's inequality by the left-hand side of condition~\ref{thm:main_alt:con_1} with probability at least $1-\frac{1}{n}$;
    \item[ii)]
    The left-hand side of condition~\ref{thm:main_alt:con_2} is similar to the left-hand side in condition~\ref{thm:main_null_master:con_1} in Theorem~\ref{thm:main_null_master} and is a bound on the term $I_{t,n,1}$ and the remainder terms in \eqref{eq:I_tn2_rate}, but note that here, naturally, instead of requiring the bound to be on the order of $\Co(n^{-1/2})$ under the null, under the alternative we require the bound to become insignificant as compared to the signal $|\EE Z_{t,j,h}|$, which is exactly what condition~\ref{thm:main_alt:con_2} requires.
\end{itemize}
Then, if the signal $|I_{t,n,3}+I_{t,n,4}+I_{t,n,5}|=|\EE Z_{t,j,h}|$ is strong enough, as required by conditions~\ref{thm:main_alt:con_1} and \ref{thm:main_alt:con_2}, with probability at least $1 - C e^{-\pn} - \frac{1}{n}$ our test in \eqref{test:fixed_t} will reject, leading to the conclusion of Theorem~\ref{thm:main_alternative}.

Next, we establish how condition~\ref{thm:main_alt:con_1} is implied by conditions~\ref{thm:main_alt:con_3} and \ref{thm:main_alt:con_4}.  To establish condition~\ref{thm:main_alt:con_1}, it suffices that $|I_{t,n,5}|=|\int_{\Omega_h} \left[ \exp\{t\,m_{0,j}(\bff, x_j)\} - 1 \right] \rmd\PP|$ is strong enough, which is implied by condition~\ref{thm:main_alt:con_3}, and moreover $|I_{t,n,5}|$ dominates $I_{t,n,3}$ and $I_{t,n,4}$.  We first consider the term $I_{t,n,3}$, which arises due to the diversified projection.  By a remark in the proof of Lemma~\ref{lem:Riesz_analytic}, $\Psi$ in \eqref{labeltildeg} can be discarded if $g=g_0$, so
\begin{align*}
I_{t,n,3} \textstyle = \int_{\Omega_h} [ \exp\{t\,g_{0,j}^\us(\wt\bff, x_j)\} - \exp\{t\,m_{0,j}^\us(\bff, x_j)\} ] \rmd\PP .
\end{align*}
Then, by Taylor expansion and the boundedness of $m_{0,j}^\us$ and $g_{0,j}^\us$ derived in the same remark, followed by the Cauchy-Schwarz inequality, Lemma~\ref{lemma:derivative_bound_via_original_specific} and then Lemma~\ref{lemma:factor},
\begin{align*}
& \textstyle |I_{t,n,3}| \le C t \int_{\Omega_h} |g_{0,j}^\us(\wt\bff, x_j)-m_{0,j}^\us(\bff, x_j)| \rmd\PP \le C t [ \int_{\Omega_h} \{g_{0,j}^\us(\wt\bff, x_j)-m_{0,j}^\us(\bff, x_j)\}^2 \rmd\PP ]^{1/2} \le C \frac{t}{h} \deltaf,
\end{align*}
where the constant $C$ is not dependent on $t$.

We next consider the term $I_{t,n,4}$.  By arguments similar to those used in establishing the bound for $I_{t,n,3}$ above, but instead using the quantity $r_{\textup{b},m,j}$ introduced above Theorem~\ref{thm:m_hat_est_master} following the Cauchy-Schwarz inequality, we obtain
\begin{align*}
    I_{t,n,4} \le C t r_{\textup{b},m,j}
\end{align*}
(where the constant $C$ again is not dependent on $t$).

Therefore, under condition~\ref{thm:main_alt:con_4}, $|I_{t,n,5}|$ indeed dominates $I_{t,n,3}$ and $I_{t,n,4}$.  This concludes the proof of the remarks following Theorem~\ref{thm:main_alternative}.
\end{proof}

\section{Low-dimensional regime as a special case}
\label{sec:lowd_supp}

\subsection{Conditional screening in low-dimension}
\label{sec:lowd}

Our low-dimensional regime refers to the scenario where the dimension $d$ is small.  This regime can be handled as a special case of our high-dimensional regime through straightforward modifications.  First, (when $d$ is small) in our regression model~\eqref{eq:reg_model} we can simply set $\bfF=\bX_{-j}$ where $\bX_{-j}\in\RR^{d-1}$ denotes the vector $\bX$ but without the $j$-th coordinate.  Then, since we observe all coordinates of $\bX$, the ``latent'' $\bfF$ is in fact observed. Consequently, in model~\eqref{eq:reg_model}, $m_0(\bfF,X_j)=m_0(\bX)=\EE[Y|\bX]$, and $r=d-1$.  Then, if we wish to screen the $j$-th coordinate $X_j$, our conditional screening hypotheses in \eqref{hypo_null} and \eqref{hypo_alt} become accordingly
\begin{gather*}
  H_0: \text{for all}~\bx\in\RR^d,~m_{0,j}(\bx) = 0~\text{identically},~\text{against} \\
  H_A: \text{for some}~\bx\in\RR^d,~\text{we have}~m_{0,j}(\bx) \neq 0 .
\end{gather*}
We summarize the practical implementation of our conditional screening test in the low-dimensional regime in Algorithm~\ref{algorithmlowd}.  Also assume that, in place of Assumption~\ref{ass:DGP}, we have $\bX\in[-b,b]^d$.  Then the corresponding theoretical development can be obtained a special case of our theoretical results for the high-dimensional regime, by setting $\wt\bfF_i$ to $\bX_{i,-j}$, the preliminary regression and adjusted regression estimators $\wh g_n$ and $\wc g_n$ to $\wh m_n$ and $\wc m_n$ in Algorithm~\ref{algorithmlowd} respectively, and $\deltaf=0$.  For brevity, we omit presenting the theoretical results in the low-dimensional regime.

\begin{algorithm}
\hrulefill
\caption{Partial derivative estimation \& conditional screening test by deep neural network in the low dimensional regime}
\SetKwInput{KwData}{Input}
\KwData{Observed samples $Y_i, \bX_i \in \RR^d$, $i\in \{1,\cdots, n\}$; testing if the $j$-th coordinate $\bX_j$ has additional contributions.}
\KwResult{Preliminary and refined regression estimators $\wh m_n(\cdot)$ and $\wc m_n(\cdot)$, their smoothed derivative estimators $\wh m_{n,j}^\us(\cdot)$ and $\wc m_{n,j}^\us(\cdot)$, $\wh\alpha_n(\cdot)$, and test results. }
\begin{itemize}
\item[1)] Estimate the regression function $m_0$ in \eqref{eq:reg_model} by the neural network estimator $\wh m_n$ defined as
\begin{align}
	\label{eq:def_wh_f_n}
\textstyle \wh m_n=\argmin_{m\in \calFn(d)} \frac{1}{n}\sum_{i=1}^n \left\{ Y_i-m(\bX_i) \right\}^2 .
\end{align}
Conduct cross validation (see Remark~\ref{rmk:bandwidth}), or otherwise, to find a smoothing bandwidth $h$.

\item[2)] Define the loss function $\Rnullhat(\cdot)$ as
\begin{align*}
    \textstyle \Rnullhat(\alpha) = \frac{1}{n} \sum_{i=1}^n \alpha^2(\bX_i) - 2 \frac{1}{n} \sum_{i\in\calI_h} \alpha_j^\us(\bX_i) ,
\end{align*}
which can be approximated via Lemma~\ref{lemma:alpha_j_numerical},  and let the estimator $\wh\alpha_n$ of $\alpha_{t,n}^*$ be
\begin{align*}
    \textstyle \wh\alpha_n = \argmin_{\alpha\in\calFn(d)} \Rnullhat(\alpha).
\end{align*}

\item[3)] Return $\wc m_n = \wh m_n + \deltathat \wh\alpha_n$, where
\begin{align*}
	\textstyle \deltathat = \dfrac{1}{\frac{1}{n} \sum_{i=1}^n \wh\alpha_n^2(\bX_i)} \frac{1}{n} \sum_{i=1}^n \left\{ Y_i - \wh m_n(\bX_i) \right\} \wh\alpha_n(\bX_i) ,
\end{align*}
\item[4)] Let $\|\wh\epsilon \wh\alpha_n\|_{L_2(\PPn)} = [ \frac{1}{n} \sum_{i=1}^n \{ \wh\epsilon_i \wh\alpha_n(\bX_i) \}^2 ]^{1/2}$ be the estimated standard deviation.  Then, for a prescribed significance level $0<\alpha<1$, for the fixed $t$ test, reject the null hypothesis \eqref{hypo_null} in our conditional screening hypotheses if
\begin{eqnarray*}
 \textstyle \frac{ n^{1/2} }{ t \|\wh\epsilon \wh\alpha_n\|_{L_2(\PPn)} } |\frac{1}{n}\sum_{i=\calI_n} [ \exp\{t\,\Psi(\wc m_j^\us(\bX_i)) \} - 1] |> z_{1-\alpha/2} ,
\end{eqnarray*}
where $z_{1-\alpha/2}$ is the $1-\alpha/2$ quantile of a standard normal distribution.
The decision rules for the ensemble tests are analogous to those in \eqref{sup} and \eqref{square} for the high-dimensional regime, and are omitted.
\label{algorithmlowd}
\end{itemize}
\hrulefill
\end{algorithm}

\newpage

\section{Additional application results}
\label{sec:app_results_more}

\begin{figure}[!h]
    \includegraphics[height=19cm, width = 1.0\textwidth]{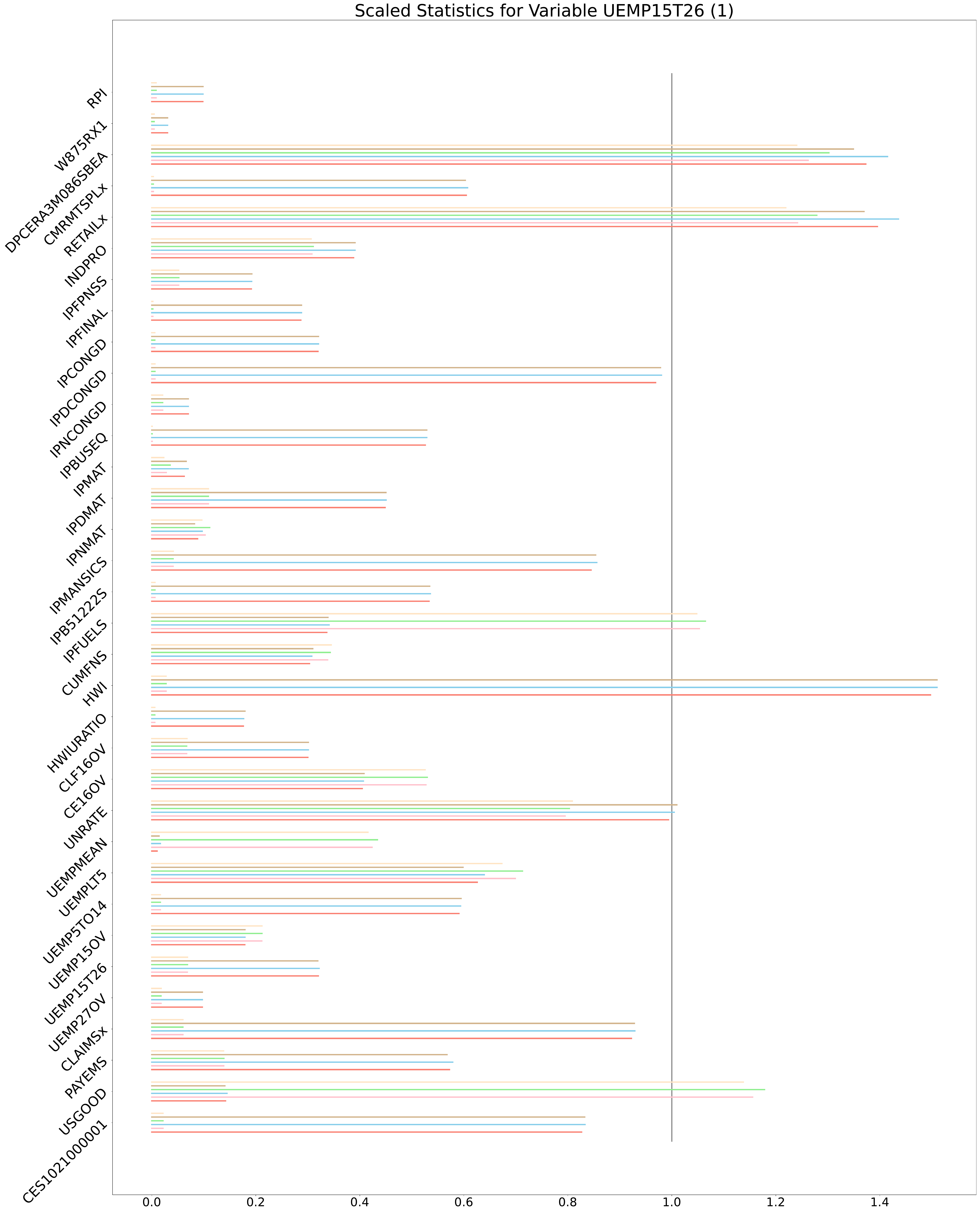}
\caption{FRED-MD dataset: Significance  of variables after neural factor regression at the $5\%$ significance level for the response variable UEMP15T26 (Part 1). Intervals covering one correspond to the significant idiosyncratic variables. Please refer to Figure~\ref{assetpricing} for the description of the intervals.}
\label{FRED10}
\end{figure}

\begin{figure}
    \includegraphics[height=20cm, width = 1.0\textwidth]{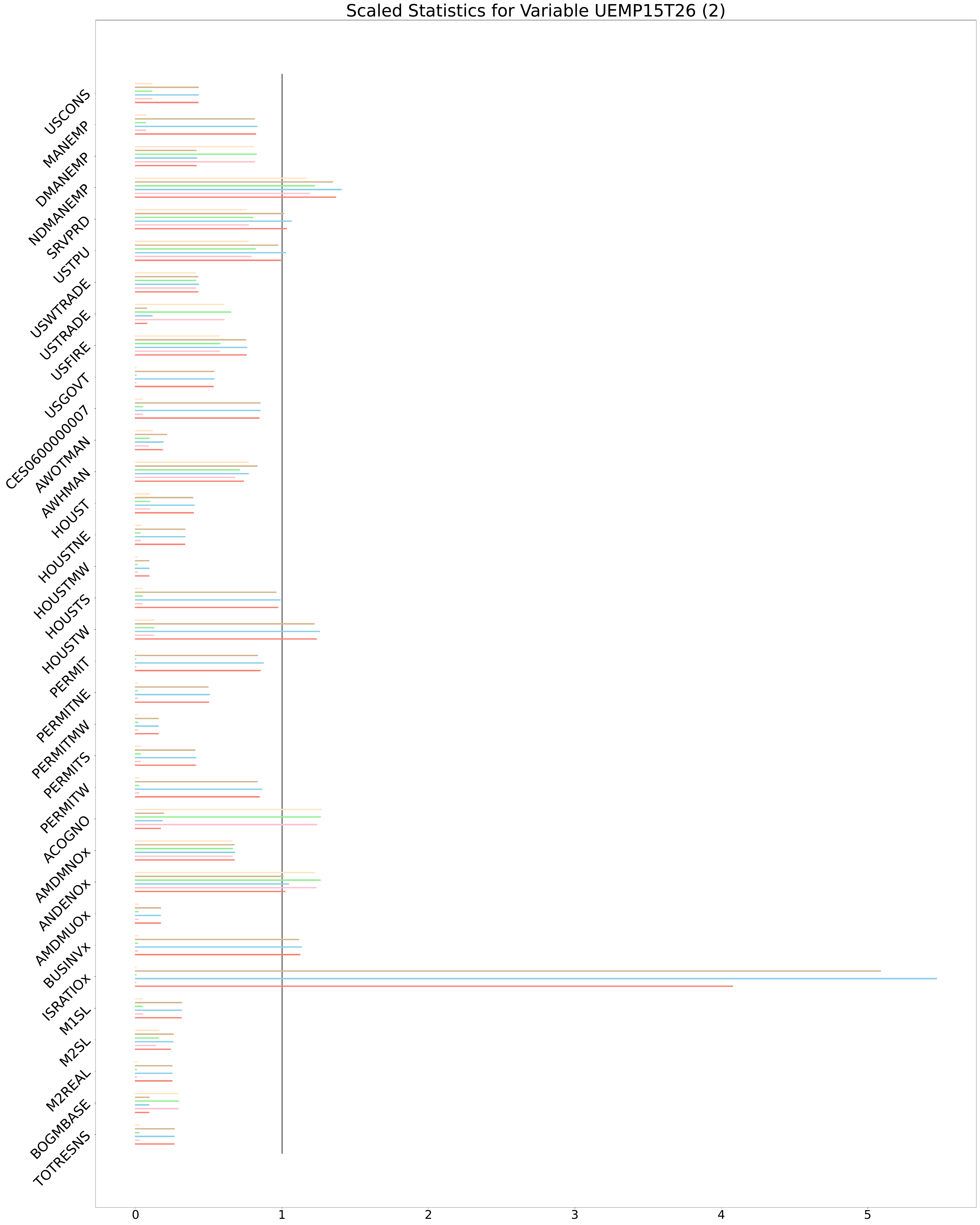}
\caption{FRED-MD dataset: Significance  of variables after neural factor regression at the $5\%$ significance level for the response variable UEMP15T26 (Part 2). Intervals covering one correspond to the significant idiosyncratic variables. Please refer to Figure~\ref{assetpricing} for the description of the intervals.}
\label{FRED11}
\end{figure}

\begin{figure}
    \includegraphics[height=20cm, width = 1.0\textwidth]{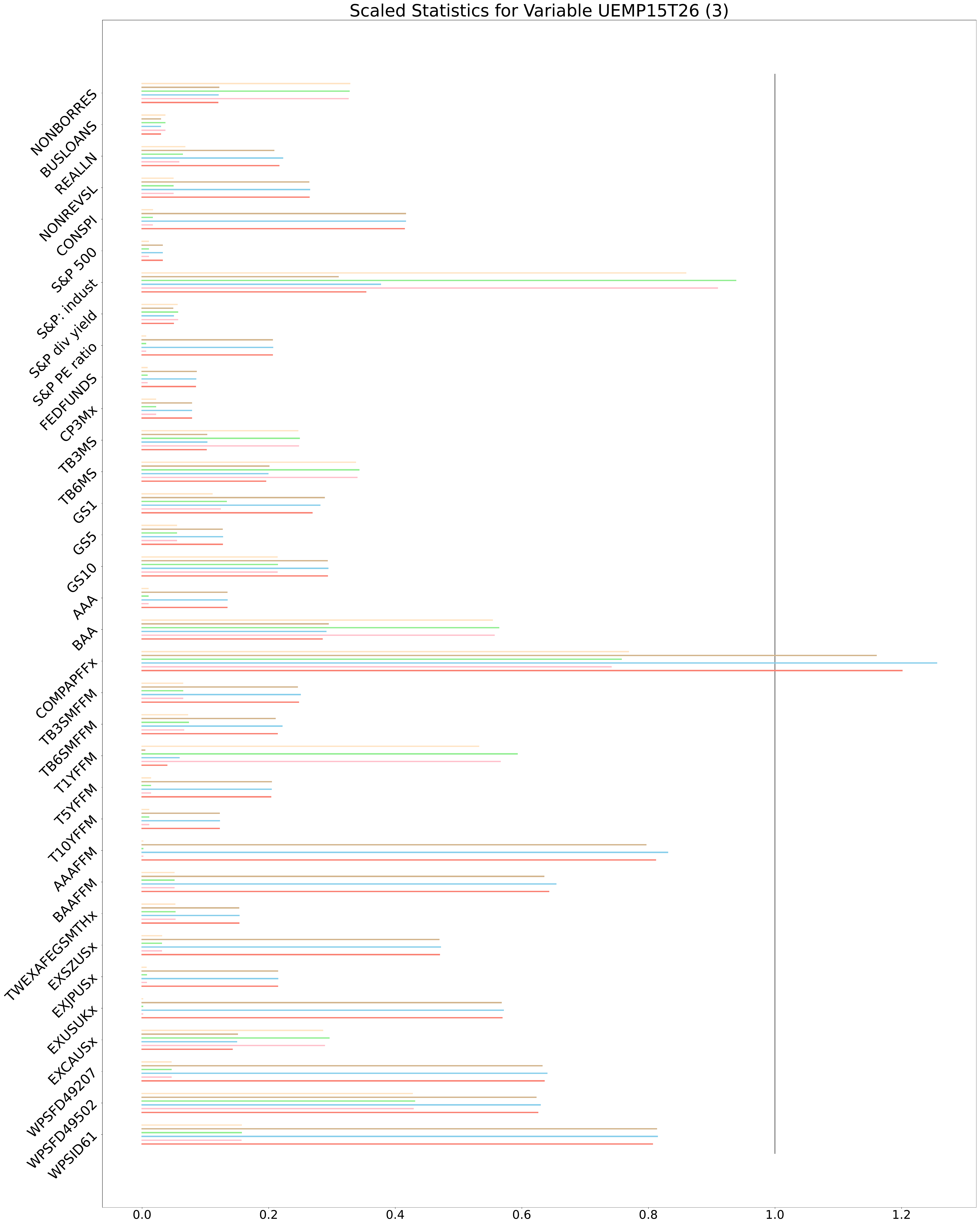}
\caption{FRED-MD dataset: Significance  of variables after neural factor regression at the $5\%$ significance level for the response variable UEMP15T26 (Part 3). Intervals covering one correspond to the significant idiosyncratic variables. Please refer to Figure~\ref{assetpricing} for the description of the intervals.}
\label{FRED12}
\end{figure}

\begin{figure}
    \includegraphics[height=20cm, width = 1.0\textwidth]{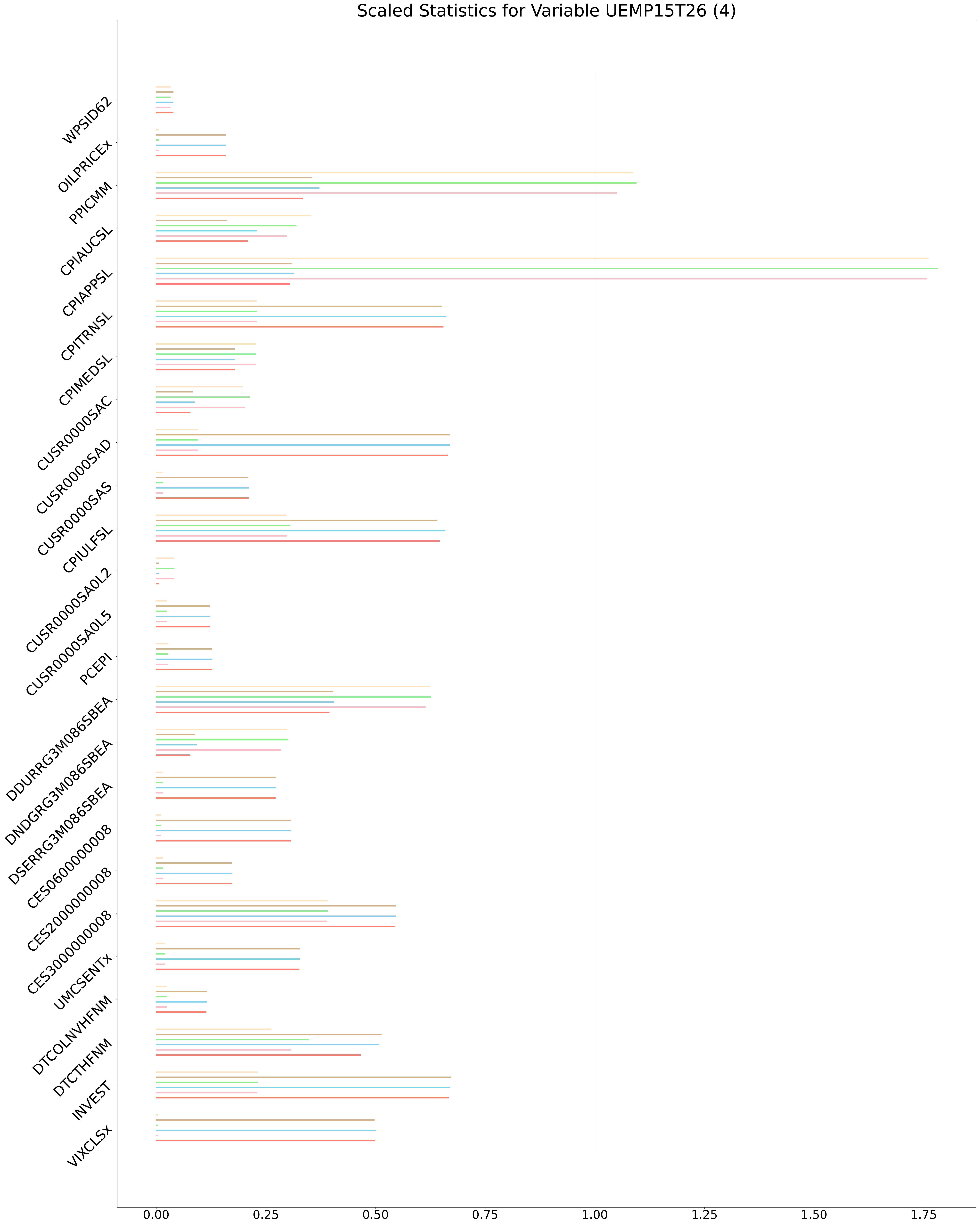}
\caption{FRED-MD dataset: Significance  of variables after neural factor regression at the $5\%$ significance level for the response variable UEMP15T26 (Part 4). Intervals covering one correspond to the significant idiosyncratic variables. Please refer to Figure~\ref{assetpricing} for the description of the intervals.}
\label{FRED13}
\end{figure}

\begin{figure}
    \includegraphics[height=20cm, width = 1.0\textwidth]{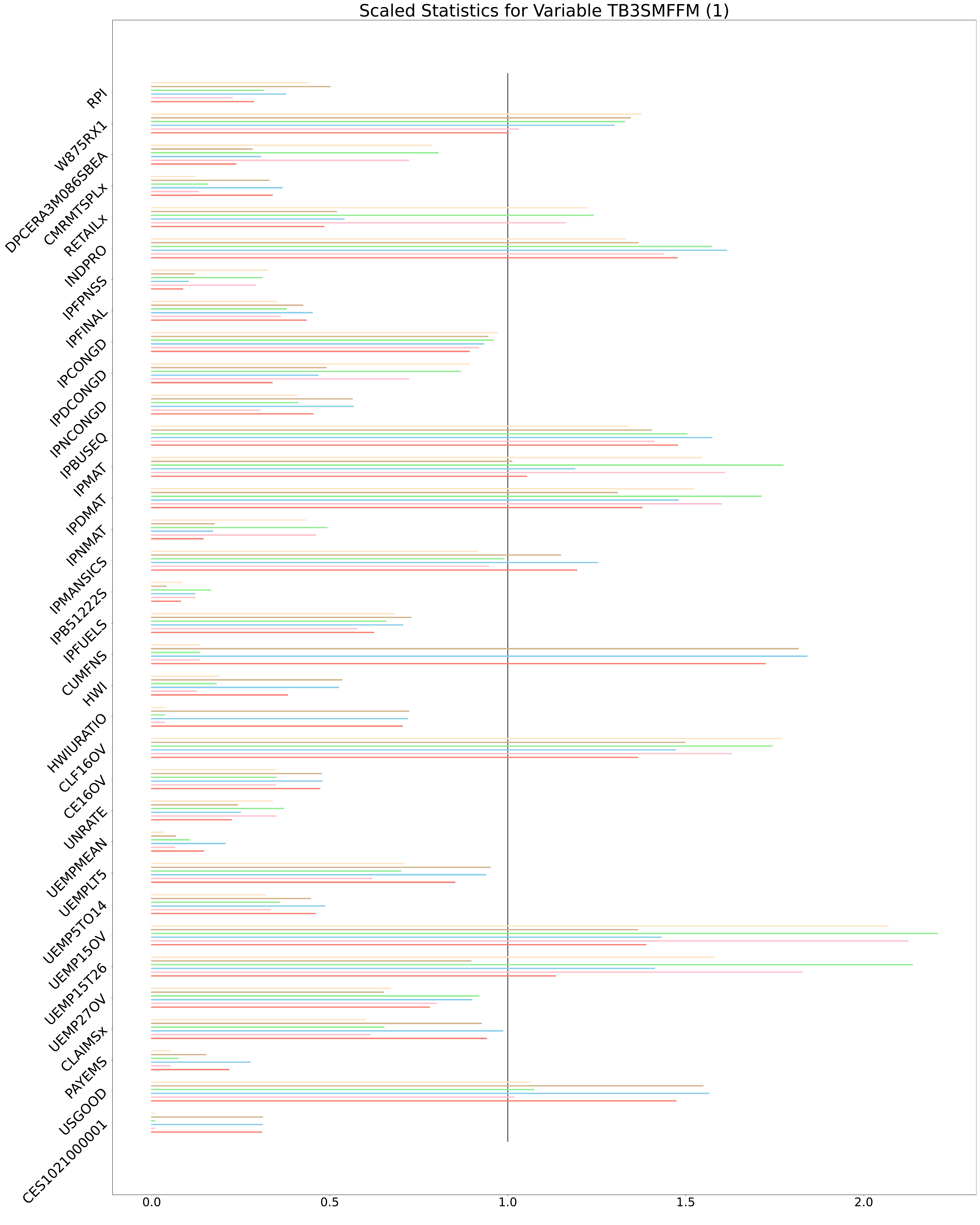}
\caption{FRED-MD dataset: Significance  of variables after neural factor regression at the $5\%$ significance level for the response variable TB3SMFFM (Part 1). Intervals covering one correspond to the significant idiosyncratic variables. Please refer to Figure~\ref{assetpricing} for the description of the intervals.}
\label{FRED1}
\end{figure}

\begin{figure}
    \includegraphics[height=20cm, width = 1.0\textwidth]{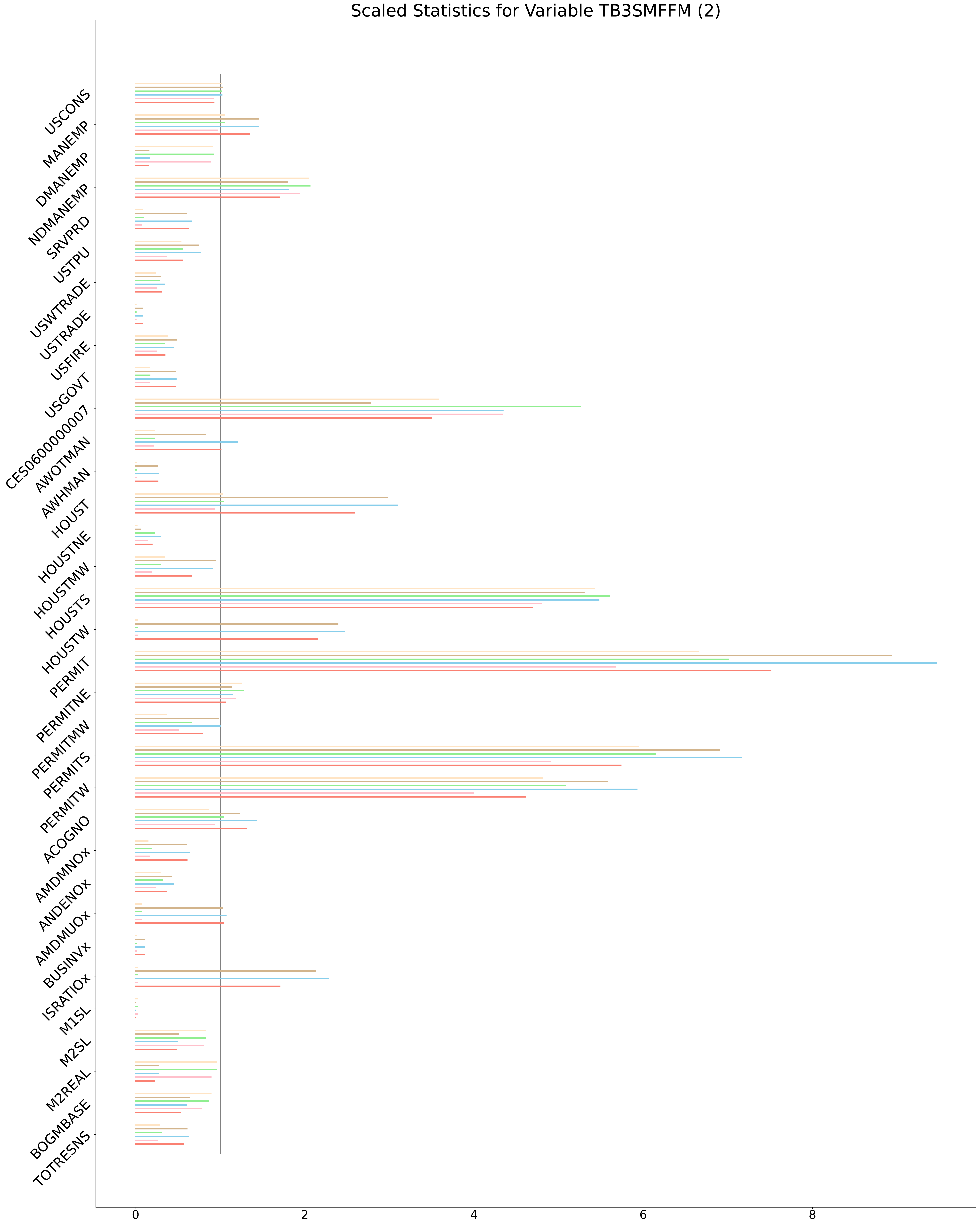}
\caption{FRED-MD dataset: Significance  of variables after neural factor regression at the $5\%$ significance level for the response variable TB3SMFFM (Part 2). Intervals covering one correspond to the significant idiosyncratic variables. Please refer to Figure~\ref{assetpricing} for the description of the intervals.}
\label{FRED2}
\end{figure}

\begin{figure}
    \includegraphics[height=20cm, width = 1.0\textwidth]{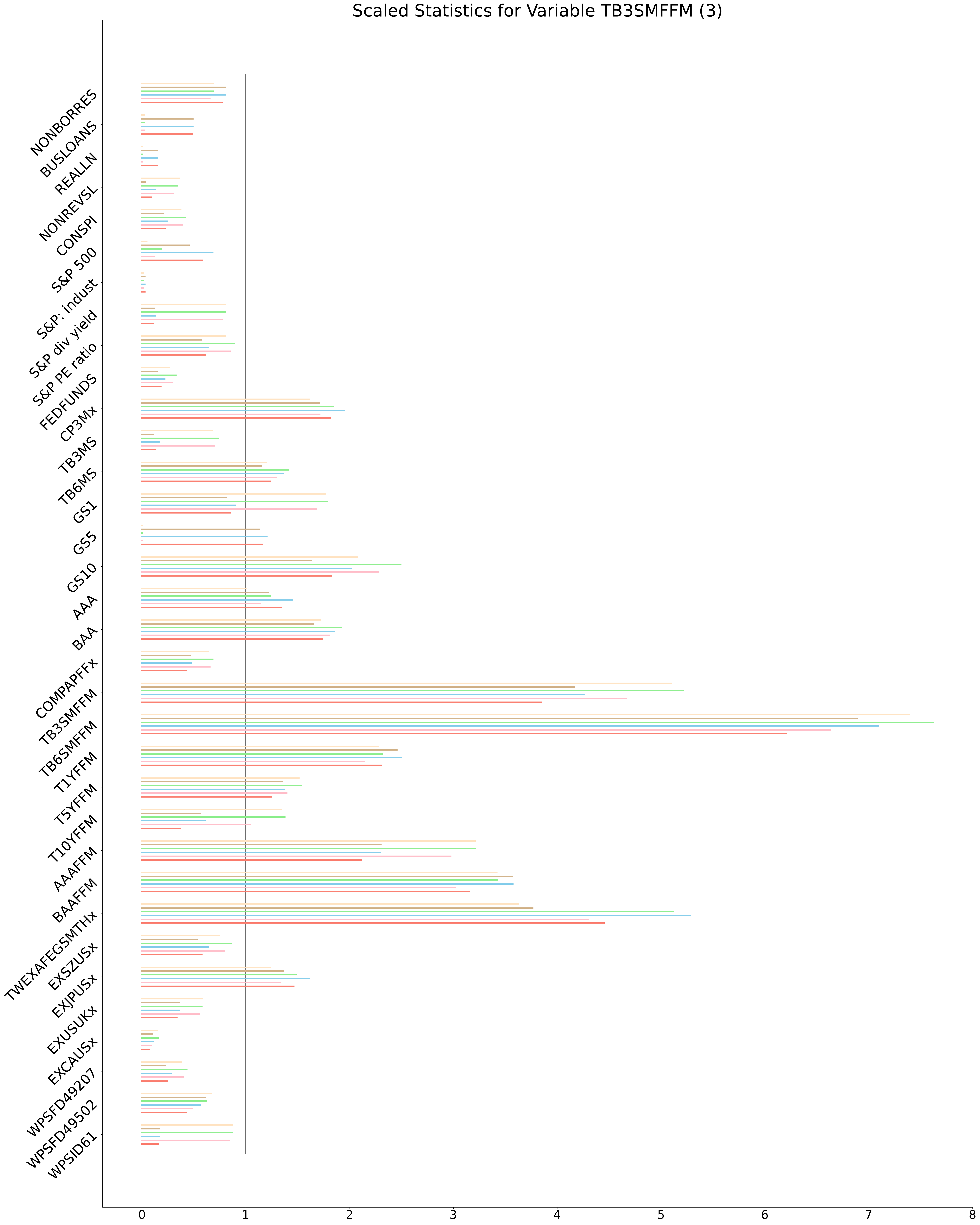}
\caption{FRED-MD dataset: Significance  of variables after neural factor regression at the $5\%$ significance level for the response variable TB3SMFFM (Part 3). Intervals covering one correspond to the significant idiosyncratic variables. Please refer to Figure~\ref{assetpricing} for the description of the intervals.}
\label{FRED3}
\end{figure}

\begin{figure}
    \includegraphics[height=20cm, width = 1.0\textwidth]{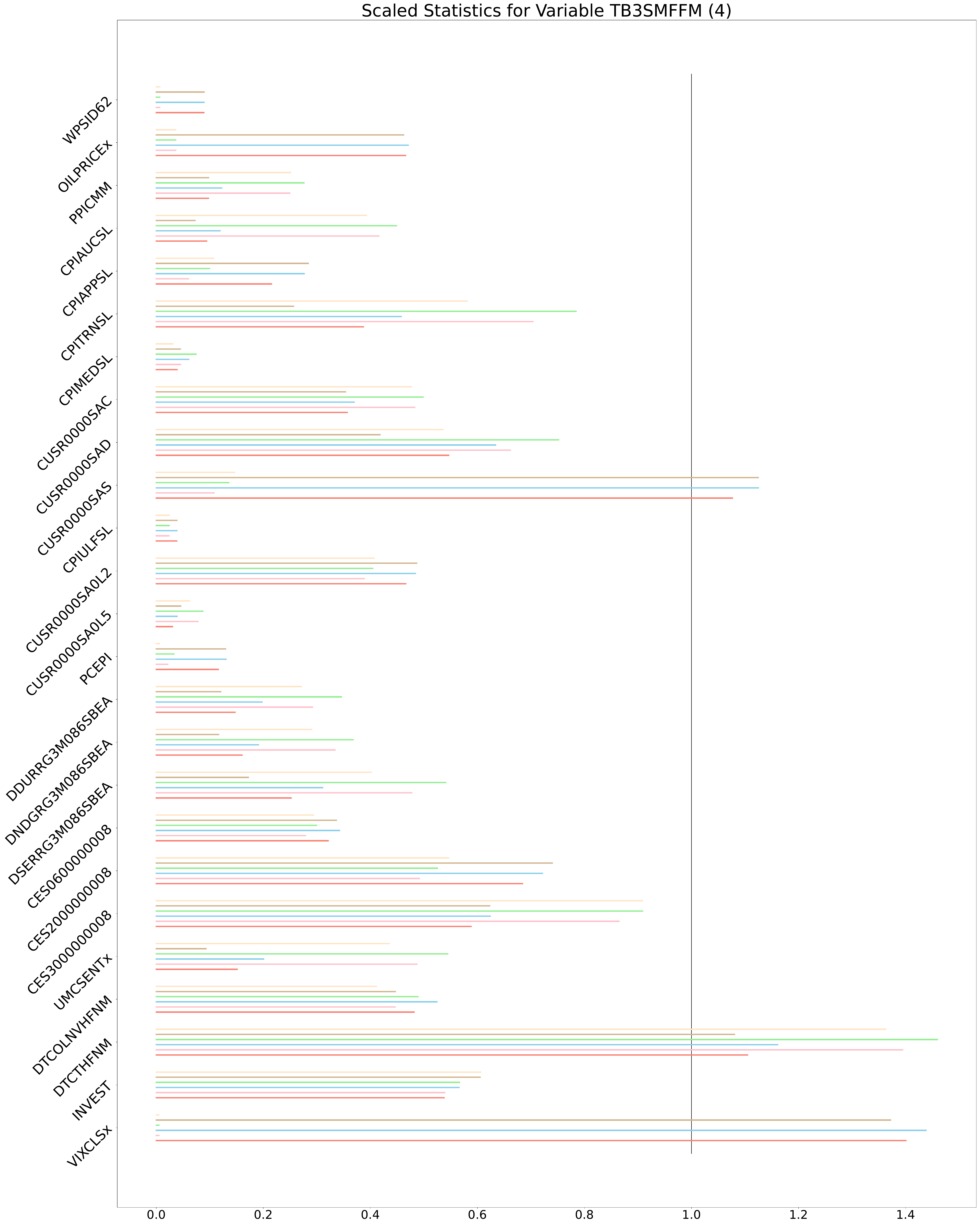}
\caption{FRED-MD dataset: Significance  of variables after neural factor regression at the $5\%$ significance level for the response variable TB3SMFFM (Part 4). Intervals covering one correspond to the significant idiosyncratic variables. Please refer to Figure~\ref{assetpricing} for the description of the intervals.}
\label{FRED4}
\end{figure}

\begin{figure}
    \includegraphics[height=20cm, width = 1.0\textwidth]{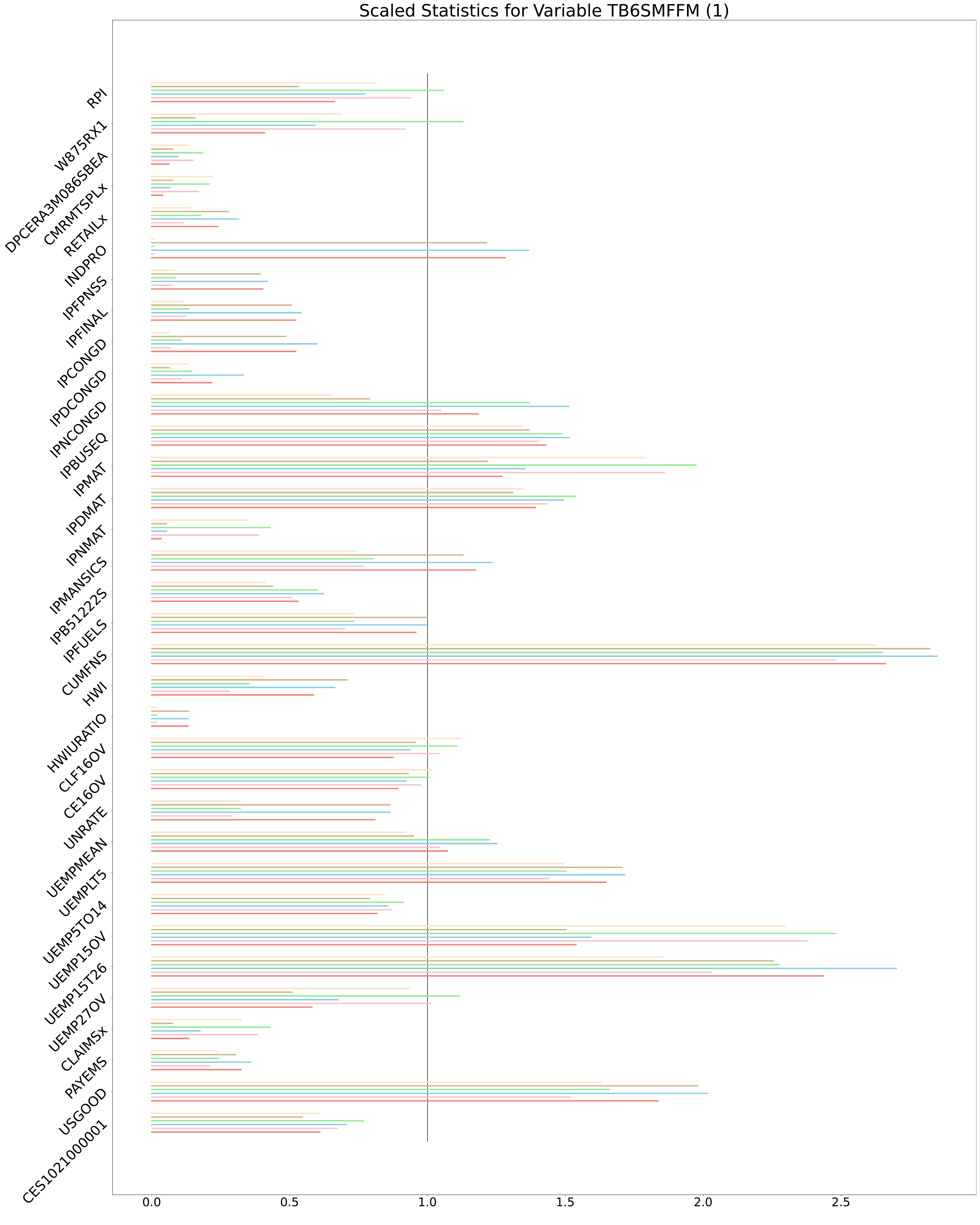}
\caption{FRED-MD dataset: Significance  of variables after neural factor regression at the $5\%$ significance level for the response variable TB6SMFFM (Part 1). Intervals covering one correspond to the significant idiosyncratic variables. Please refer to Figure~\ref{assetpricing} for the description of the intervals.}
\end{figure}

\begin{figure}
    \includegraphics[height=20cm, width = 1.0\textwidth]{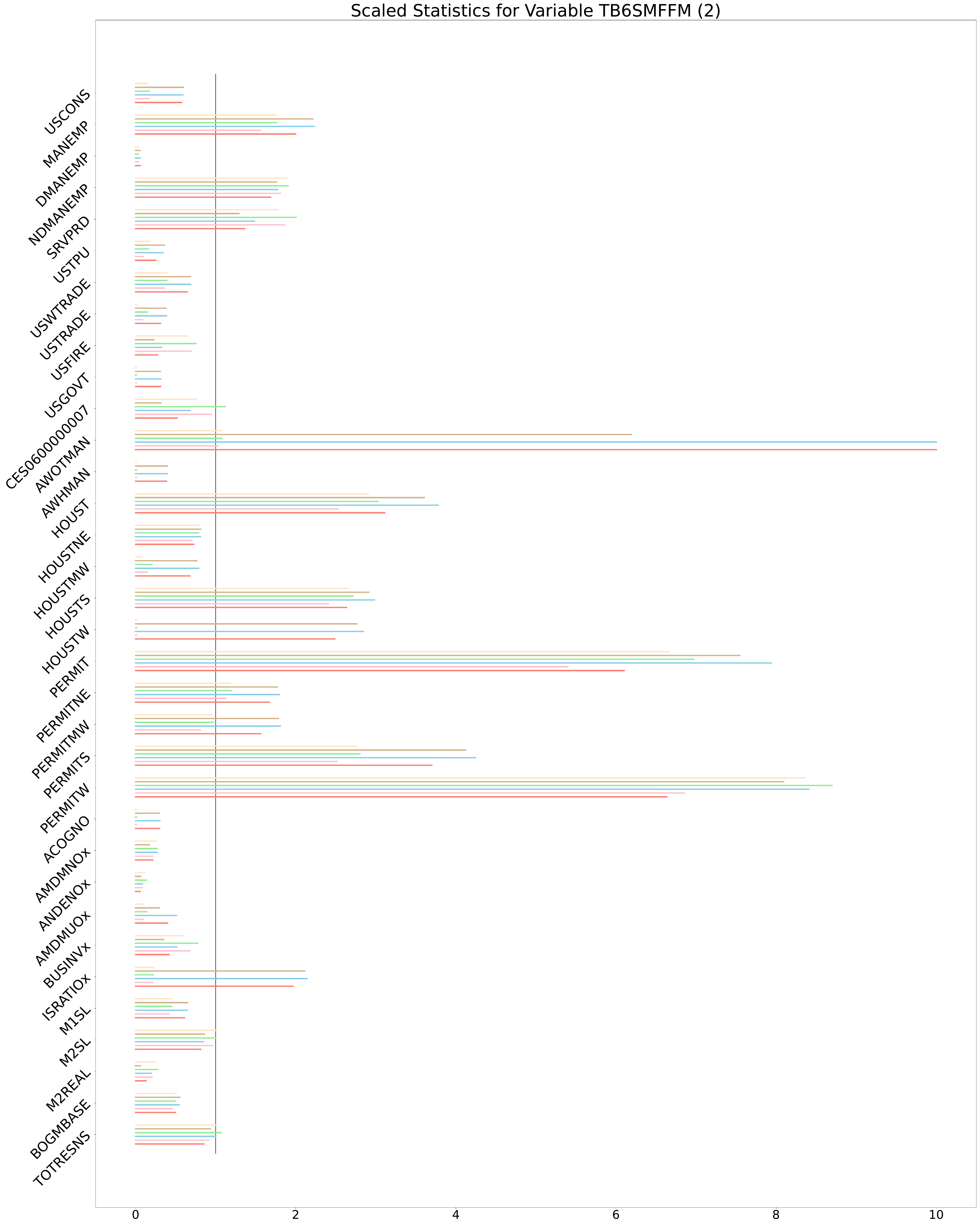}
\caption{FRED-MD dataset: Significance  of variables after neural factor regression at the $5\%$ significance level for the response variable TB6SMFFM (Part 2). Intervals covering one correspond to the significant idiosyncratic variables. Please refer to Figure~\ref{assetpricing} for the description of the intervals.}
\label{FRED6}
\end{figure}

\begin{figure}
    \includegraphics[height=20cm, width = 1.0\textwidth]{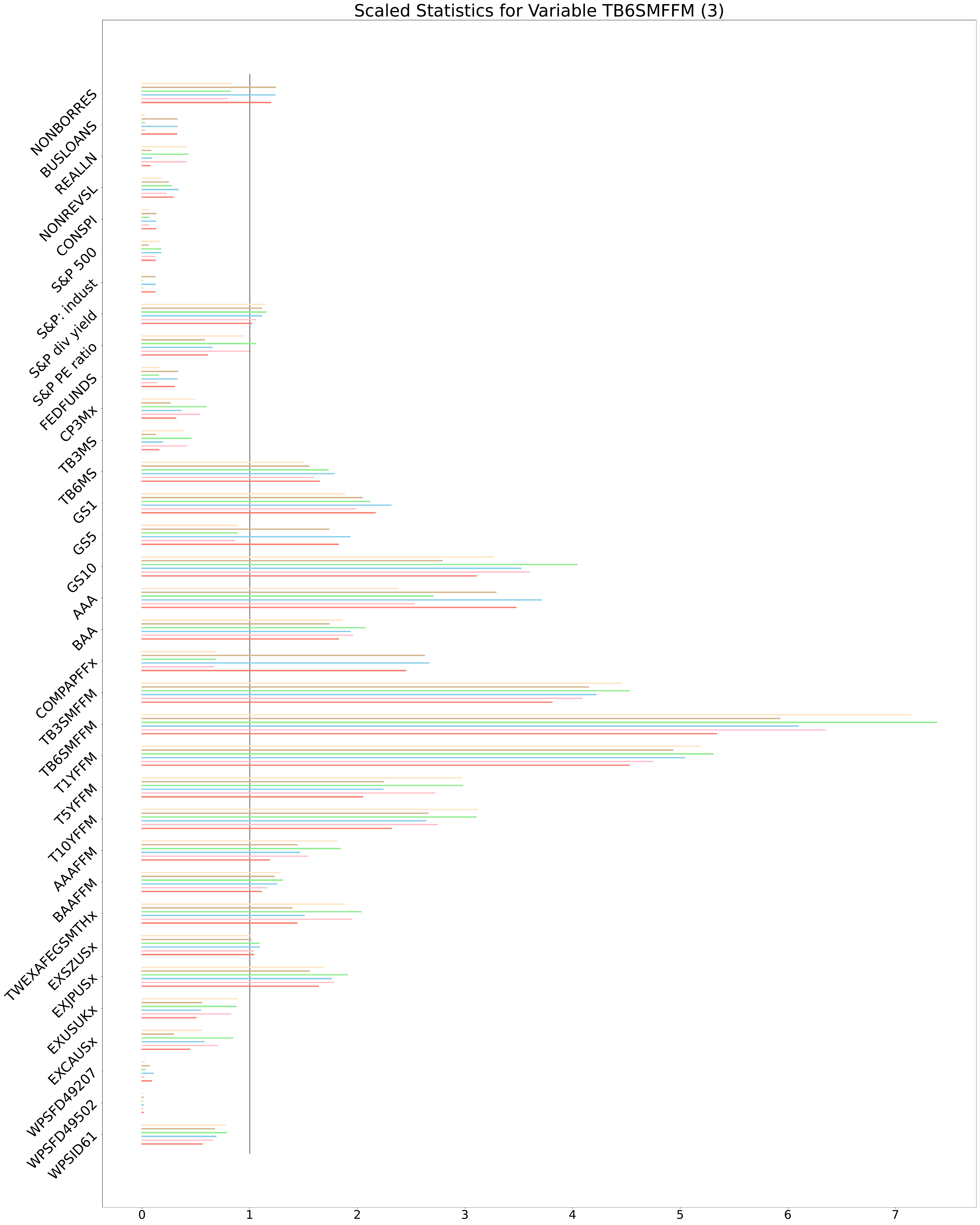}
\caption{FRED-MD dataset: Significance  of variables after neural factor regression at the $5\%$ significance level for the response variable TB6SMFFM (Part 3). Intervals covering one correspond to the significant idiosyncratic variables. Please refer to Figure~\ref{assetpricing} for the description of the intervals.}
\end{figure}

\begin{figure}
    \includegraphics[height=20cm, width = 1.0\textwidth]{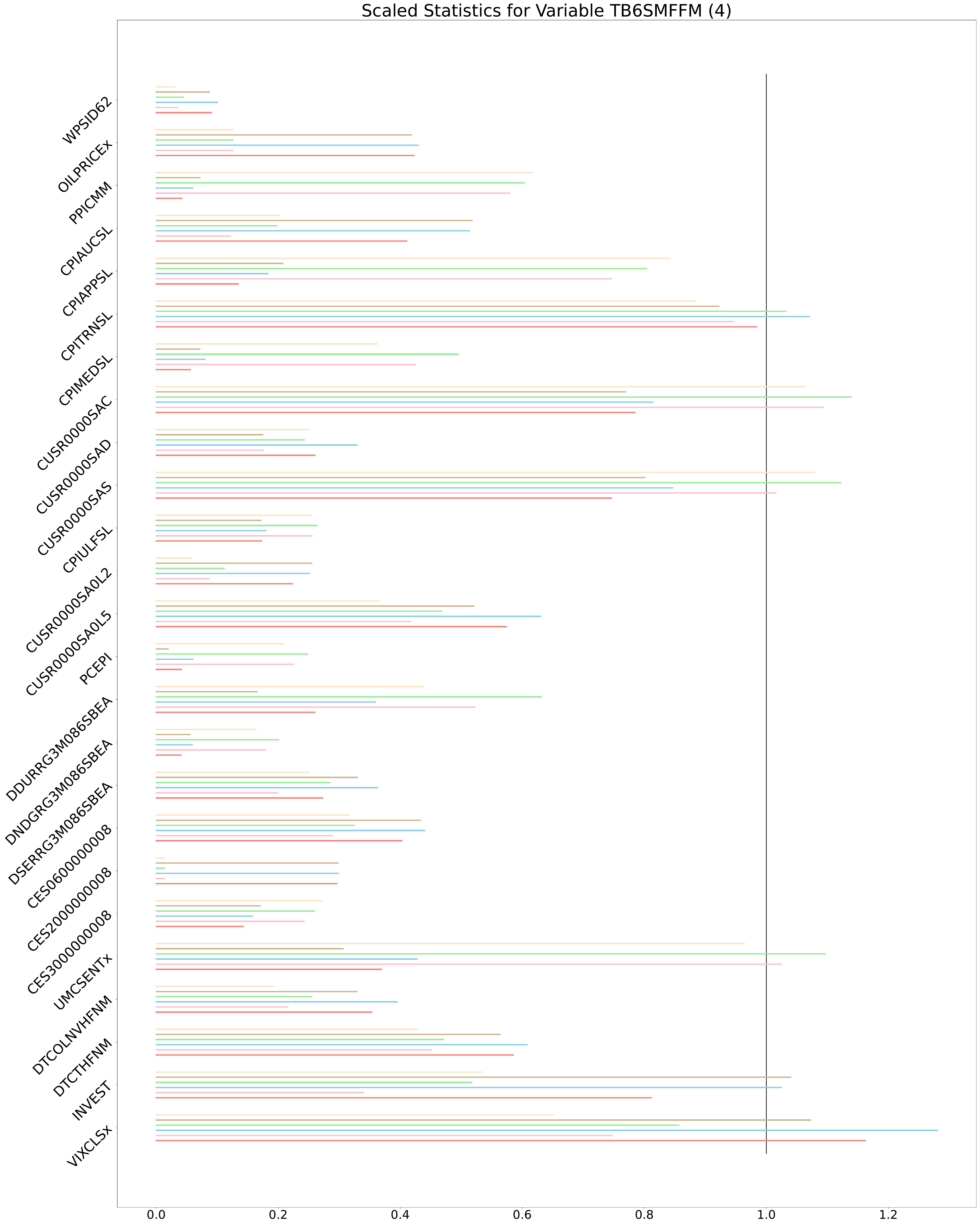}
\caption{FRED-MD dataset: Significance  of variables after neural factor regression at the $5\%$ significance level for the response variable TB6SMFFM (Part 4). Intervals covering one correspond to the significant idiosyncratic variables. Please refer to Figure~\ref{assetpricing} for the description of the intervals.}
\label{FRED8}
\end{figure}

\end{document}